\documentclass[10pt, a4paper, oneside]{article}
\usepackage[top=1.3cm, left=3.2cm, right=3.2cm]{geometry}
\usepackage{etoolbox}
\usepackage{tabto}
\usepackage{booktabs}
        
\usepackage{microtype}
\usepackage{algorithm,algorithmicx,amssymb,amsmath,amsthm,amsfonts,bm,bbm,epsfig,dsfont,authblk}
\usepackage[labelfont=bf,font=normalsize,indention=0.6cm,margin=0cm]{caption}
\usepackage[noend]{algpseudocode}
\usepackage[numbers,sort&compress,longnamesfirst,sectionbib]{natbib}
\usepackage{xspace,nicefrac,hyperref}

\usepackage
[
    subrefformat = simple, 
    labelformat = simple,  
]
{subcaption}         
\usepackage{thm-restate}
\usepackage{todonotes}

\newtheorem{theorem}{Theorem}
\newtheorem{lemma}{Lemma}
\newtheorem{definition}{Definition}
\newtheorem{corollary}{Corollary}

\theoremstyle{plain}

\numberwithin{theorem}{section}
\numberwithin{lemma}{section}
\numberwithin{definition}{section}
\numberwithin{corollary}{section}
\numberwithin{equation}{section}

\newcommand{\eps}{\varepsilon}

\renewcommand{\geq}{\geqslant}
\renewcommand{\le}{\leqslant}
\renewcommand{\ge}{\geqslant}

\newcommand{\Oh}{\mathcal{O}}
\newcommand{\N}{\mathbb{N}}
\newcommand{\R}{\mathbb{R}}

\newcommand{\thmref}[1]{Theorem~\ref{thm:#1}}

\newcommand{\lemref}[1]{Lemma~\ref{lem:#1}}

\newcommand{\corref}[1]{Corollary~\ref{cor:#1}}

\newcommand{\figref}[1]{Figure~\ref{fig:#1}}

\newcommand{\secref}[1]{Section~\ref{sec:#1}}
\newcommand{\eq}[1]{Equation~\eqref{eq:#1}}

\newcommand{\Prob}[1]{\mathrm{Pr}\left[\,#1\,\right]}
\newcommand{\Pro}[2]{\underset{\substack{#1}}{\mathrm{Pr}}\left[\,#2\,\right]}

\newcommand{\Ex}[1]{\mathbb{E}\left[\,#1\,\right]}

\newcommand{\Wlog}{W.\,l.\,o.\,g.\xspace}
\newcommand{\wlo}{w.\,l.\,o.\,g.\xspace}
\newcommand{\ie}{i.\,e.\xspace}
\newcommand{\eg}{e.\,g.\xspace}
\newcommand{\cf}{c.\,f.\xspace}

\newcommand{\aas}{a.\thinspace a.\thinspace s.\xspace}

\newcommand{\n}[1]{\overline{#1}}

\title{The Satisfiability Threshold for Non-Uniform Random 2-SAT\footnote{This paper is partially funded by the project \emph{Scalefree Satisfiability} (project no. 416061626) of the German Research Foundation (DFG).}}

\author[1]{Tobias Friedrich}
\author[1]{Ralf Rothenberger}
\affil[1]{\small Hasso Plattner Institute, University of Potsdam, Potsdam, Germany\\
  \texttt{firstname.lastname@hpi.de}}
\date{}

\begin{document}

\clearpage\maketitle

\thispagestyle{empty}
\begin{abstract}
Propositional satisfiability (SAT) is one of the most fundamental problems in computer science.
Its worst-case hardness lies at the core of computational complexity theory, for example in the form of NP-hardness and the (Strong) Exponential Time Hypothesis.
In practice however, SAT instances can often be solved efficiently.
This contradicting behavior has spawned interest in the average-case analysis of SAT and has triggered the development of sophisticated rigorous and non-rigorous techniques for analyzing random structures.

Despite a long line of research and substantial progress, most theoretical work
on random SAT assumes a \emph{uniform} distribution on the variables.
In contrast, real-world instances often exhibit large fluctuations in variable occurrence.
This can be modeled by a \emph{non-uniform} distribution of the variables,
which can result in distributions closer to industrial SAT instances.

We study satisfiability thresholds of non-uniform random $2$-SAT with $n$~variables and $m$~clauses and with an arbitrary ensemble of probability distributions $(\vec{p}^{(n)})_{n\in\N}$ with $p_1^{(n)}\ge p_2^{(n)}\ge\ldots\ge p_n^{(n)}>0$ for each $n\in\N$.
We show that for $p_{1}^2\in o\left(\sum_{i=1}^n p_i^2\right)$ there is a sharp satisfiability threshold at $m=1/(\sum_{i=1}^n p_i^2)$.
Otherwise, the asymptotic threshold is at $m\in\Theta((1-{\sum_{i=1}^n p_i^2})/(p_1\cdot(\sum_{i=2}^n p_i^2)^{1/2}))$ and the threshold is coarse.
This result generalizes the seminal works by Chvatal and Reed~[FOCS 1992] and by Goerdt~[JCSS 1996].
\end{abstract}

\newpage

\section{Introduction}

Satisfiability of Propositional Formulas (SAT) is one of the most thoroughly researched topics in theoretical computer science.
It was one of the first problems shown to be NP-complete by Cook~\cite{Cook71} and, independently, by Levin~\cite{Levin73}.
Today SAT stands at the core of many techniques in modern complexity theory, for example NP-completeness proofs \cite{Karp72} or running time lower bounds assuming the (Strong) Exponential Time Hypothesis \cite{Bringmann14,CyganNPPRW11,ImpagliazzoP01,ImpagliazzoPZ98}.

In addition to its importance for theoretical research, Propositional Satisfiability is also famously applied in practice.
Despite the theoretical hardness of SAT, many problems arising in practice can be transformed to SAT instances and then solved efficiently with state-of-the-art solvers.
Problems like hard- and software verification, automated planning, and circuit design are often transformed into SAT instances.
Such formulas arising from practical and industrial problems are therefore referred to as \emph{industrial SAT instances}.
The efficiency of SAT solvers on these instances suggests that they have a structure that makes them easier to solve than the theoretical worst-case.

\subsection{Uniform Random k-SAT and the satisfiability threshold conjecture:}

Random $k$-SAT is used to study the average-case complexity of Boolean Satisfiability.
In the model, a random formula $\Phi$ with $n$ variables, $m$ clauses, and $k$ literals per clause is generated in conjunctive normal form.
Each of these formulas has the same uniform probability to be generated.
Therefore, we also refer to this model as \emph{uniform random $k$-SAT}.

One of the most prominent questions related to studying uniform random $k$-SAT is trying to prove the \emph{satisfiability threshold conjecture}.
The conjecture states that for a uniform random $k$-SAT formula $\Phi$ with $n$ variables and $m$ clauses there is a real number $r_k$ such that
 \[
\lim_{n \to \infty} \Pr\left(\Phi \text{~is satisfiable}\right) =
\begin{cases}
  1 & m/n < r_k;\\
  0 & m/n > r_k.
\end{cases}
\]
Chvatal and Reed~\cite{chvatalreed92} and, independently, Goerdt~\cite{Goerdt1996threshold} proved the conjecture for $k=2$ and showed that $r_2=1$. 
For larger values of $k$ upper and lower bounds have been established, \eg, $3.52\le r_3\le4.4898$~\cite{DBLP:journals/tcs/DiazKMP09,HajiaghayiSorkin2003threshold,Kaporis2006probabilistic}.
Methods from statistical mechanics~\cite{mezard2002analytic} were used to derive a numerical estimate of $r_3\approx 4.26$.
Coja-Oghlan and Panagiotou ~\cite{Coja-Oghlan:2014:AKT:2591796.2591822,kostathreshold} showed a bound (up to lower order terms) for $k\ge3$ with 
\mbox{$r_k = 2^k \log 2 - \tfrac12 (1 + \log 2) \pm o_k(1)$}. 
Finally, Ding, Sly, and Sun~\cite{Ding:2015:PSC:2746539.2746619} proved the exact position of the threshold for sufficiently large values of $k$.
Still, for $k$ between 3 and the values determined by Ding, Sly, and Sun the conjecture remains open.

The satisfiability threshold is also connected to the average hardness of solving instances.
For uniform random $k$-SAT for example, the on average hardest instances are concentrated around the threshold~\cite{hardeasy}.

\subsection{Non-Uniform Random SAT}

There is a large body of work which considers other random SAT models, \eg regular random $k$-SAT~\cite{regkSAT2016A,regkSAT2005,regkSAT2016B,regkSAT2010}, random geometric $k$-SAT~\cite{geometricSAT} and $2+p$-SAT~\cite{twoplusp4,twoplusp3,twoplusp1,twoplusp2}.
However, most of these are not motivated by modeling the properties of industrial instances.
One such property is community structure~\cite{industrialSAT1}, \ie some variables have a bias towards appearing together in clauses.
It is clear by definition that such a bias does not exist in uniform random $k$-SAT.
Therefore, Gir\'aldez-Cru and Levy~\cite{giraldez2015modularity} proposed the Community Attachment Model, which creates random formulas with clear community structure.
However, the work of Mull et al.~\cite{communityHardness} shows that instances generated by this model have exponentially long resolution proofs with high probability, making them hard on average for solvers based on conflict-driven clause learning.

Another important property of industrial instances is their degree distribution.
The degree distribution of a formula~$\Phi$ is a function $f\colon\N\to\N$, where $f(x)$ denotes the fraction of different Boolean variables that appear $x$ times in $\Phi$ (negated or unnegated).
Instances created with the uniform random $k$-SAT model have a binomial distribution, while some families of industrial instances appear to follow a power-law distribution~\cite{ABL2009}, \ie $f(x)\sim x^{-\beta}$, where $\beta$ is a constant intrinsic to the instance.
Therefore, Ans\'otegui et al.~\cite{AnsoteguiBL09} proposed a random $k$-SAT model with a power-law degree distribution. 
Empirical studies by the same authors~\cite{AnsoteguiBGL14,AnsoteguiBGL15,ABL2009,AnsoteguiBL09} found that this distribution is beneficial for the runtime of SAT solvers specialized in industrial instances.
However, it looks like instances generated with their model can be solved faster than uniform instances, but not as fast as industrial ones: 
median runtimes around the threshold still seem to scale exponentially for several state-of-the-art solvers~\cite{AAAI17}.

Therefore, we want to consider a generalization of the model by Ans\'otegui et al.~\cite{ABL2009}.
Our model allows instances with \emph{any} given ensemble of variable distributions $(\vec{p}^{(n)})_{n\in\N}$ instead of only power laws:
We draw $m$ clauses of length $k$ over $n$ Boolean variables at random. 
For each clause the $k$ variables are drawn with a probability proportional to $\vec{p}^{(n)}$, the $n$-th distribution in the ensemble. 
Then, they are negated independently with a probability of $\nicefrac12$ each. 
This means, the probability ensemble is part of the model, but the number of variables $n$ determines which distribution from the ensemble we actually use. 
We call this model \emph{non-uniform random $k$-SAT} and denote it by $\mathcal{D}\left(n,k,(\vec{p}^{(n)})_{n\in\N},m\right)$.
Although $\mathcal{D}\left(n,k,(\vec{p}^{(n)})_{n\in\N},m\right)$ cannot capture all properties of industrial instances, e.g. community structure, it can help us to investigate the influence of the degree distribution on the structure and on the computational complexity of such instances in an average-case scenario. 

As one of the steps in analyzing this connection, we would like to find out for which ensembles of variable probability distributions an equivalent of the satisfiability threshold conjecture holds in non-uniform random $k$-SAT.
In previous works we already proved upper and lower bounds on the threshold position~\cite{ESA17} and showed sufficient conditions on sharpness~\cite{SAT18}.
In this work we are interested in actually determining the satisfiability threshold for $k=2$.
It has to be noted that Cooper et al.~\cite{cooper2sat} and Levy~\cite{levy17} already studied thresholds in a similar random 2-SAT model.
The difference is that in their models the degrees are fixed and the random instances determined in a configuration-model-like fashion, while in our model we only have a sequence of expected degrees from which the actual degrees might deviate.
Another difference is that we do a complete analysis of the model we consider, while they have additional constraints on their degree sequences.
However, if we assume the expected degrees that our model implies to be the actual degrees, the thresholds determined by Cooper et al. and by Levy coincide with the ones we derive for our model.

\subsection{Our Results} \label{sec:results}

We investigate the position and behavior of the satisfiability threshold for non-uniform random 2-SAT.
That is, we fix an ensemble of variable probability distributions $({\vec{p}}^{(n)})_{n\in\N}=(p_1^{(n)},p_2^{(n)},\ldots,p_n^{(n)})_{n\in\N}$ and a function $m\colon \N\to\N$ for the number of clauses, depending on the number of variables $n$ and analyze the probability of satisfiability asymptotically as $n$ increases.
To this end, we use the following definition and say that a function $m^\star\colon \N\to\N$ is an \emph{asymptotic satisfiability threshold} if for any function $m\colon \N\to\N$
\[\lim_{n\to\infty}\Pr_{\Phi\sim \mathcal{D}\left(n,k,(\vec{p}^{(n)})_{n\in\N},m(n)\right)}\left(\Phi\text{ satisfiable}\right)=\begin{cases}1&\text{if}\ m\in o(m^\star))\\0&\text{if}\ m\in \omega(m^\star).\end{cases}\]
We also say that an asymptotic satisfiability threshold $m^\star$ is \emph{sharp} if for all $\eps>0$
\[\lim_{n\to\infty}\Pr_{\Phi\sim \mathcal{D}\left(n,k,(\vec{p}^{(n)})_{n\in\N},m(n)\right)}\left(\Phi\text{ satisfiable}\right)=\begin{cases}1&\text{if}\ m\le(1-\eps)\cdot m^\star\\0&\text{if}\ m\ge(1+\eps)\cdot m^\star.\end{cases}\]
If an asymptotic threshold is not sharp, we call it \emph{coarse}.

For an ensemble of probability distributions $(\vec{p}^{(n)})_{n\in\N}=(p_1^{(n)},p_2^{(n)},\ldots,p_n^{(n)})_{n\in\N}$, \wlo we assume $p_1^{(n)}\ge p_2^{(n)} \ge \ldots \ge p_n^{(n)}$ for all $n\in\N$. 
We also interpret $p_i\colon \N\setminus [i-1]\to\R^+$ as a function in $n$ with $p_i(n)={p_i}^{(n)}$.
We will see that the threshold position and its sharpness depend on how the functions of the two highest variable probabilities $p_1$ and $p_2$ behave compared to the other variable probabilities.
Moreover, it depends on the asymptotic behavior of those values.
The squares of $p_1$ and $p_2$ will be compared to the sum of squares of the other variable probabilities, $\sum_{i=1}^n p_i^2$ and $\sum_{i=2}^n p_i^2$.
Note that these sums of squares are functions in $n$ as well.
The conditions on those functions can be checked if we know the ensemble of probability distributions that our model uses.
We are going to show that there are four cases depending on $p_1$, $p_2$, $\sum_{i=1}^n p_i^2$, and $\sum_{i=2}^n p_i^2$:
\begin{enumerate}
\item If $p_1^2\in o\left(\sum_{i=1}^n p_i^2\right)$, then there is a sharp threshold at exactly \[m^\star=\frac{1}{\sum_{i=1}^n p_i^2}.\]
\item If $p_1^2\in\Theta\left(\sum_{i=1}^n p_i^2\right)$ and $p_2^2\in o\left(\sum_{i=2}^n p_i^2\right)$, then the asymptotic threshold is at \[m^\star\in\Theta\left(\frac{1-\sum_{i=1}^n p_i^2}{p_1\cdot\left(\sum_{i=2}^n p_i^2\right)^{1/2}}\right)\] and it is coarse.
The coarseness stems from the emergence of an unsatisfiable sub-formula with $3$ variables and $4$ clauses.
Furthermore, we can show that there is a range of size $\Theta(m^\star)$ around the threshold in which the probability to generate satisfiable instances is a constant bounded away from zero and one in the limit.
\item If $p_1^2\in\Theta\left(\sum_{i=1}^n p_i^2\right)$ and $p_2^2\in\Theta\left(\sum_{i=2}^n p_i^2\right)$, then we can show that the asymptotic satisfiability threshold is at \[m^\star\in\Theta\left(\frac{1-\sum_{i=1}^n p_i^2}{p_1\cdot p_2}\right),\] which is proportional to $1/q_{\max}$, where $q_{\max}$ is the maximum clause probability.
We can also show that this threshold is coarse.
This time the coarseness stems from the emergence of an unsatisfiable sub-formula of size $4$, which contains only the two most probable variables.
Again, we can show that there is a range of size $\Theta(m^\star)$ around the threshold in which the probability to generate satisfiable instances is a constant bounded away from zero and one in the limit.
\item If none of the above cases apply, there is a threshold at \[m^\star\in\Theta\left(\frac{1-\sum_{i=1}^n p_i^2}{\sum_{i=2}^n p_i^2+p_1\cdot\left(\sum_{i=2}^n p_i^2\right)^{1/2}}\right).\]
The threshold is again coarse, but this time the probability in the threshold interval cannot be bounded.
\end{enumerate}
It is important to understand why we want to show that in the second and third case there is a range of size $\Theta(m^\star)$ around the threshold in which the probability to generate satisfiable instances can be bounded away from zero and one by a constant.
This implies that the probability cannot approach zero or one for some functions $m$ that are only a constant factor away from the threshold.
According to our definition of sharp thresholds (see definition~\ref{def:general-sharpness}), that means in those cases the threshold cannot be sharp, but must be coarse.
We will later see that the statements for those two cases also imply coarseness of the threshold in the last case.
Moreover, in all four cases the asymptotic threshold is at 
\[m^\star\in\Theta\left(\frac{1-\sum_{i=1}^n p_i^2}{\sum_{i=2}^n p_i^2+p_1\cdot\left(\sum_{i=2}^n p_i^2\right)^{1/2}}\right).\]
Together with the conditions on $p_1^2$ and $p_2^2$ this threshold function simplifies to the ones we stated in the first three cases, respectively.
The four cases give us a complete dichotomy of coarseness and sharpness for the satisfiability threshold of non-uniform random 2-SAT.
This result generalizes the seminal works by \citet{chvatalreed92} and by \citet{Goerdt1996threshold} to arbitrary ensembles of variable probability distributions and includes their findings as a special case (\cf \secref{2-SAT-examples}). 
We summarize our findings in the following theorem.
\begin{restatable}{theorem}{statemainthm}\label{thm:main}
Given an ensemble of probability distributions $(\vec{p}^{(n)})_{n\in\N}$ with $p_1^{(n)}\ge p_2^{(n)} \ge \ldots \geq p_n^{(n)}$ for all $n\in\N$.
\begin{enumerate}
\item If $p_1^2\in o\left(\sum_{i=1}^n p_i^2\right)$, then non-uniform random 2-SAT with this ensemble of probability distributions has a sharp satisfiability threshold at $m^\star=1/\sum_{i=1}^n p_i^2$. 
\item Otherwise, non-uniform random 2-SAT with this ensemble of probability distributions has a coarse satisfiability threshold at 
\[m^\star\in\Theta\left(\frac{1-\sum_{i=1}^n p_i^2}{\sum_{i=2}^n p_i^2+p_1\cdot\left(\sum_{i=2}^n p_i^2\right)^{1/2}}\right).\]
\end{enumerate}
\end{restatable}

\subsection{Techniques}

For the sharp threshold result, we only show the conditions on sharpness.
These also imply the existence of an asymptotic threshold.
For the coarse threshold results, however, we first have to show the existence of an asymptotic threshold function $m^\star$.
Then, we have to show that for some range of constants $\eps\in[\eps_1,\eps_2]$ the probability to generate a satisfiable instance at $\eps\cdot m^\star$ is a constant bounded away from zero and one in the limit.

We extend and generalize the proof ideas of Chvatal and Reed~\cite{chvatalreed92}.
In order to show a lower bound on the threshold, we investigate the existence of bicycles.
Bicycles were introduced by Chvatal and Reed. 
They are sub-formulas which appear in every unsatisfiable formula.
We can show with a first moment argument, that these do not appear below a certain number of clauses, thus making formulas satisfiable.

In order to show an upper bound on the threshold, we investigate the existence of snakes.
Snakes are unsatisfiable sub-formulas and have also been introduced by Chvatal and Reed.
We can show with a second-moment argument that snakes of certain sizes do appear above a certain number of clauses, thus making formulas unsatisfiable.
However, we need to be careful and distinguish more possibilities of partially mapping snakes onto each other than in the uniform case.
Unfortunately, this method does not work if the two largest variable probabilities are too large asymptotically.
In that case we lower-bound the probability that an unsatisfiable sub-formula containing only those two variables exists.
This can be done with a simple inclusion-exclusion argument and the resulting lemma also works for $k\ge 3$.

\section{Preliminaries}

\subsection{Mathematical Notation}

We use blackboard bold letters to denote number sets. 
$\N$ denotes the set of natural numbers including zero and $\R$ denotes the set of real numbers.
We let $\R^+$ denote the set of positive real numbers.
For any $x,y\in\R$ with $x\le y$ we let $[x,y]=\left\{z\in\R\mid x\le z\le y\right\}$ denote the closed interval of real numbers from $x$ to $y$.
We denote open intervals with round instead of square brackets.
For any $m,n\in\N$ we let $[m\ldots n]=[m,n]\cap\N$ and $[n]= [1\ldots n]$.
Also, we let $\mathcal{P}(\cdot)$ denote the power set and let $\mathcal{P}_k(\cdot)$ denote the set of cardinality-$k$ elements of the power set. 

For a real-valued function $f$ and $c\in\R$ we let $\lim_{x\to c} f(x)$ denote the limit of $f$ as $x$ approaches $c$.
For a sequence $a_1, a_2, \ldots$ of real numbers we let $\lim_{n\to\infty} a_n$ denote the limit of $a_n$ as $n$ approaches infinity.
It holds that $\lim_{n\to\infty} a_n=L$ if and only if for every real number $\eps>0$ there is an $n_0\in\N$ so that for all $n>n_0$ we have $|a_n-L|<\eps$.
Furthermore, we will use Landau notation.
That means, for two real-valued functions $f$ and $g$ defined on the same unbounded subset of $\R^+$ we use the following notation:
\begin{itemize}
\item $f\in\Oh(g)$ \tabto{1.55cm}$\Leftrightarrow$ $\exists\ \eps>0\  \exists\  n_0\  \forall\ n>n_0\colon |f(n)|\le \eps\cdot g(n)$,
\item $f\in\Theta(g)$ \tabto{1.55cm}$\Leftrightarrow$ $\exists\ \eps_1>0\ \exists\  \eps_2>0\  \exists\  n_0\  \forall\  n>n_0\colon \eps_1\cdot g(n)\le f(n)\le\eps_2\cdot g(n)$,
\item $f\in\Omega(g)$ \tabto{1.55cm}$\Leftrightarrow$ $\exists\  \eps>0\  \exists\  n_0\  \forall\  n>n_0\colon f(n)\ge\eps\cdot g(n)$,
\item $f\in o(g)$ \tabto{1.55cm}$\Leftrightarrow$ $\forall\  \eps>0\  \exists\  n_0\  \forall\  n>n_0\colon |f(n)|\le \eps\cdot g(n)$, and
\item $f\in \omega(g)$ \tabto{1.55cm}$\Leftrightarrow$ $\forall\  \eps>0\  \exists\  n_0\  \forall\  n>n_0\colon |f(n)|\ge \eps\cdot |g(n)|$.
\end{itemize}
The definitions of limits and Landau symbols will be used heavily when dealing with satisfiability thresholds in this paper.
Thus, it is important to state those definitions explicitly.
Another definition we use to compare functions is the following.
For two functions $f,g\colon X\to\R$ which are defined on the same domain $X$ we write $f\le g$ iff for all $x\in X$ it holds that $f(x)\le g(x)$.

\subsection{Boolean Satisfiability}

We analyze non-uniform random $k$-SAT on $n$ variables and $m$ clauses. 
We denote by $X_1, \ldots, X_n$ the Boolean variables. 
A clause is a disjunction of $k$ literals $\ell_1 \lor \ldots \lor \ell_k$, where each literal assumes a (possibly negated) variable. 
For a literal $\ell_i$ let $|\ell_i|$ denote the variable of the literal.
A formula $\Phi$ in conjunctive normal form (CNF) is a conjunction of clauses $c_1 \land \ldots \land c_m$. 
A formula is in \emph{k-CNF} if it is in CNF and each clause consists of exactly $k$ literals.
We conveniently interpret a clause $c$ both as a Boolean formula and as a set of literals. 
We say that $\Phi$ is satisfiable if there exists an assignment of variables $X_1, \ldots, X_n$ such that the formula evaluates to $1$.

\subsection{Non-Uniform Model} \label{sec:non-uniform}

In this section we introduce a generalization of random $k$-SAT, which we call \emph{non-uniform random $k$-SAT}.
The model is inspired by power-law random $k$-SAT and geometric random $k$-SAT by \citet{AnsoteguiBL09}.
These two models are also notable special cases of non-uniform random $k$-SAT.
As in random $k$-SAT we draw $m$ clauses independently at random.
However, the clause probabilities are now non-uniform.
This allows us to model different frequency distributions.
Formally, we define the model as follows.

\begin{definition}[Non-Uniform Random $k$-SAT]
\label{def:non-uniform-random-SAT} 
Let $m$, $n$, $k$ be given, and consider an ensemble of probability distributions $({\vec{p}}^{(n)})_{n\in\N}=({p_{1}}^{(n)},\ldots,{p_{n}}^{(n)})_{n\in\N}$, where each distribution ${\vec{p}}^{(n)}$ is defined over $n$ Boolean variables with ${p_{1}}^{(n)},\ldots,{p_{n}}^{(n)}>0$ and $\sum_{i=1}^n {p_{i}}^{(n)} = 1$. 
The \emph{non-uniform random $k$-SAT} model $\mathcal{D}^N(n, k, ({\vec{p}}^{(n)})_{n\in\N},\allowbreak m)$ constructs a random formula $\Phi$ in k-CNF by sampling $m$ clauses independently at random. Each clause is sampled as follows:
\begin{enumerate}
	\item Select $k$ variables independently at random according to the distribution ${\vec{p}}^{(n)}$. Repeat until no variables coincide.
	\item Negate each of the $k$ variables independently at random with probability $1/2$.
\end{enumerate}
\end{definition}

\Wlog we assume ${p_1}^{(n)}\ge {p_2}^{(n)}\ge \ldots {p_n}^{(n)}$ for all $n\in\N$. 
To simplify notation, we denote $p^{(n)}(X_i):=\Pr\left[X=X_i\right]={p_i}^{(n)}$.
Throughout this paper we will consider the limit behavior of probabilities in our ensembles.
Thus, based on an ensemble of discrete probability distributions $\left({\vec{p}}^{(n)}\right)_{n\in\N}$ we define for all $i\in\N$ the functions $p_i\colon \N\setminus [i-1]\to\R^+$ with $p_i(n)={p_i}^{(n)}$.
However, for the sake of brevity we omit the input parameter $n$ of those functions if it is not necessary, \ie most of the expressions we derive are actually functions in $n$ if not stated otherwise.

Non-uniform random $k$-SAT is equivalent to drawing each clause independently at random from the set of all $k$-clauses which contain no variable more than once.
The probability to draw a $k$-clause $c=(\ell_1\vee\ell_2\vee\ldots\vee\ell_k)$ over $n$ variables in this model is
\begin{equation}
q_c=\frac{\prod_{\ell \in c} p(|\ell|)}{2^k\sum_{J\in{\mathcal{P}_k\left(\left\{1,2,\ldots,n\right\}\right)}}{\prod_{j\in{J}}{p_{j}}}}. \label{eq:clause-prob}
\end{equation}
The factor $2^k$ in the denominator comes from the different possibilities to negate variables. 
Note that $k!\sum_{J\in{\mathcal{P}_k\left(\left\{1,2,\ldots,n\right\}\right)}}{\prod_{j\in{J}}{p_{j}}}$ is the probability of choosing a $k$-clause that contains no variable more than once. 
We define
\[C:=\left(k!\cdot\sum_{J\in{\mathcal{P}_k\left(\left\{1,2,\ldots,n\right\}\right)}}{\prod_{j\in{J}}{p_{j}}}\right)^{-1}\]
and write
\begin{equation}
q_c=C\cdot\frac{k!}{2^k} \cdot\prod_{\ell\in c} p(|\ell|). \label{eq:clausesample}
\end{equation}
However, for $k=2$ the factor $C$ simplifies to 
\begin{equation}\label{eq:C}C=\frac{1}{1-\sum_{i=1}^n p_i^2}.\end{equation}
Remember that both $C$ and $q_c$ are actually functions in $n$.
The representation of \eq{clausesample} makes clause probabilities easier to handle.
Since clauses are also drawn independently, the probability to generate a formula $\Phi$ with non-uniform random $k$-SAT essentially comes down to a product of variable probabilities for Boolean variables it contains. 
This makes the analysis of formulas a lot easier.
Throughout the paper we let $q_{\max}$ denote the maximum clause probability as defined in \eq{clausesample}.


The following are notable special cases of non-uniform random $k$-SAT from the literature.
They will serve as examples for applying our main theorem at the end of this work.

\subsubsection{Power-Law Random k-SAT}

Power-law random $k$-SAT was introduced by \citet{AnsoteguiBL09} as a more realistic model for industrial SAT instances.
The variable probabilities in this model follow a discrete power law with power-law exponent $\beta>2$.
More precisely for some fixed $\beta>2$ and some $n\in\N$ the distribution is $\vec{p}^{(n)}=\left(p_1^{(n)},p_2^{(n)}\ldots,p_n^{(n)}\right)$ with
\[p_i^{(n)}=\frac{(n/i)^{\frac{1}{\beta-1}}}{\sum_{j=1}^n (n/j)^{\frac{1}{\beta-1}}}.\]


\citet{AnsoteguiBL09} claim that instances generated with their model exhibit a satisfiability threshold. 
They experimentally determine the threshold position and examine the running time of state-of-the-art SAT solvers on instances generated at the threshold.
They observe that the running time of solvers can be controlled with the power-law exponent $\beta$.
With increasing exponent instances get more similar to those generated by random $k$-SAT.
Thus, solvers specialized in random instances perform better.
With small exponents, the performance of solvers specialized in industrial instances is better.
According to the authors this phenomenon can be used to generate instances on which the performance of state-of-the-art SAT solvers is comparable to the performance of those solvers on industrial instances.

However, \citeauthor{AnsoteguiBL09} do not determine the threshold position rigorously.
Furthermore they do not show theoretical bounds on the running times of the solvers they consider.
The results of this work complement theirs with regard to the threshold behavior of power-law random $2$-SAT.

In order to derive the results for power-law random $2$-SAT, we need the following bounds. 
For the sake of brevity, their proof can be found in the appendix.
\begin{restatable}{lemma}{stateplbounds}\label{lem:pl-aux1}
For the discrete power-law distribution $\vec{p}$ with exponent $\beta>2$ it holds that
\begin{align*}
p_i&=(1+o(1))\cdot\frac{\beta-2}{\beta-1}\cdot n^{-1}\cdot\left(\frac{n}{i}\right)^{1/(\beta-1)},\\
\sum_{i=1}^n {{p_i}}^2 &= \begin{cases}\Theta\left(n^{-2\frac{\beta-2}{\beta-1}}\right)&\text{for }\beta<3\\
\left(1\pm o(1)\right)\cdot\frac14\cdot\frac{\ln n}{n}&\text{for }\beta=3\\
\left(1\pm o(1)\right)\cdot\frac{(\beta-2)^2}{(\beta-3)\cdot(\beta-1)}\cdot n^{-1} & \text{for }\beta>3\text{, and}\end{cases}.
\end{align*}
\end{restatable}

\subsubsection{Geometric Random k-SAT}

Geometric Random $k$-SAT was introduced by \citet{AnsoteguiBL09} as an alternative to Power-law Random $k$-SAT.
In this model the variable probabilities are normalized terms of a geometric series with base $1/b$ for some constant $b>1$.
More precisely for some fixed $b>1$ and some $n\in\N$ the distribution is $\vec{p}^{(n)}=\left(p_1^{(n)},p_2^{(n)}\ldots,p_n^{(n)}\right)$ with
\[p_i^{(n)}=\frac{b\cdot(1-b^{-1/n})}{b-1}\cdot b^{-(i-1)/n}.\]

\citet{AnsoteguiBL09} determine the position of the satisfiability threshold for geometric random $k$-SAT experimentally.
For base parameter $b=1$ the model is equivalent to random $k$-SAT.
However, the authors observe that instances get easier as the base parameter $b$ increases.
As for power-law random $k$-SAT \citeauthor{AnsoteguiBL09} do not provide any rigorous results regarding the satisfiability threshold.
The results of this work complement theirs work with regard to the threshold of Geometric Random 2-SAT.

We will make use of the following lemma in order to derive our results. 
Its proof can be found in the appendix.

\begin{restatable}{lemma}{stategeombounds}\label{lem:g-aux1}
For the discrete geometric distribution $\vec{p}$ with base $b>1$ it holds that
\[\sum_{i=1}^n p_i^2 = \frac{b+1}{b-1}\cdot\frac{1-b^{-1/n}}{1+b^{-1/n}}=(1\pm o(1))\cdot\frac{b+1}{b-1}\cdot\frac{\ln b}{2}\cdot n^{-1}.\]
It also holds that
\[p_1=\frac{b\cdot(1-b^{-1/n})}{b-1}=(1-o(1))\cdot\frac{b\cdot\ln b}{b-1}\cdot n^{-1}.\]
\end{restatable} 

\subsection{Satisfiability Threshold}

If we fix the number of variables $n\in\N$, the clause length $k\in\N$, and a probability ensemble $\left(\vec{p}^{(n)}\right)_{n\in\N}$ and increase the number of clauses $m$ the probability that $\mathcal{D}^N$ generates satisfiable instances decreases.
This is not surprising, since each clause is a constraint on the satisfying assignments of a Boolean formula in $k$-CNF.
Thus, the more clauses a formula has, the more likely it is that the formula is not satisfiable.
This property of satisfiability is called \emph{monotonicity}.
More formally, if we have a sample space $V=\left\{0,1\right\}^N$, we call a property $P\subseteq V$ \emph{monotone} if 
\[\forall x\in P\ \forall y\in V\colon \left(\forall i\in[N]\colon y_i\ge x_i\right)\Rightarrow y\in P.\]
Intuitively, a property is monotone if adding additional elements to something with the property cannot violate it.
This is true for unsatisfiability of Boolean formulas:
If we have a set of clauses that is unsatisfiable, we cannot make it satisfiable by adding more clauses to it.

Monotonicity will play a crucial role in our proofs.
It implies that the probability for a property to hold increases with the scaling parameter we consider.
This holds for unsatisfiability and any other monotone property with regard to parameter $m$ of model $\mathcal{D}^N$ as we will show in the following lemma, the proof of which can be found in the appendix. 
\begin{restatable}{lemma}{statenondec}\label{lem:monotone-drawing}
Fix $n\in\N$, $k\in\N$, and a probability ensemble $\left(\vec{p}^{(n)}\right)_{n\in\N}$.
Let $P$ be any monotone property.
Then, the probability to generate an instance with property $P$ in $\mathcal{D}^N\left(n, k, \left(\vec{p}^{(n)}\right)_{n\in\N}, m\right)$ is non-decreasing in $m$.
\end{restatable}

However, for random $k$-SAT there is a point at which the probability to generate unsatisfiable instances suddenly increases from close to zero to close to one.
This point is called the \emph{satisfiability threshold}.
The range of $m$ in which the probability increases from close to zero to close to one is called the \emph{threshold interval}.
Formally, we define the satisfiability threshold as follows.



\begin{figure}[t!]
\centering
\includegraphics[width=0.45\linewidth]{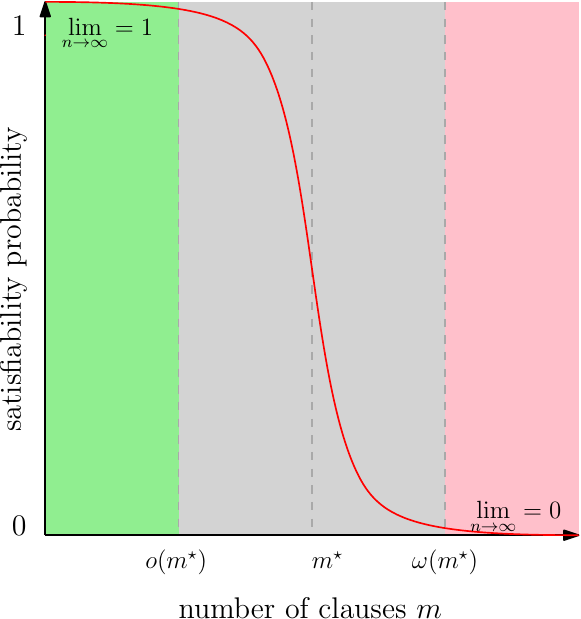}
\caption{Visual representation of a satisfiability threshold with asymptotic threshold function $m^\star$. 
For all functions $m\in o(m^\star)$ the probability tends to one (green region), for all functions $m\in\omega(m^\star)$ the probability tends to zero (red region), for all functions $m\in\Theta(m^\star)$ the probability is not restricted (gray region).}
\label{fig:asymptotic}
\end{figure}

\begin{definition}[Satisfiability Threshold]\label{def:general-threshold}
Fix a constant $k\in \N$ and an ensemble of probability distributions $({\vec{p}}^{(n)})_{n\in\N}$ and let $\mathcal{D}^N(n, k, ({\vec{p}}^{(n)})_{n\in\N},\allowbreak m)$ be a non-uniform random k-SAT model.
Let $m^\star$ and $m'$ be functions, which may depend on the other parameters of $\mathcal{D}$. 
$m^\star$ is an \emph{asymptotic threshold function for satisfiability} of $\mathcal{D}$ if for every $m'$
\begin{equation*}
    \lim_{n\to\infty}\Pro{\Phi\sim \mathcal{D}^N(n, k, ({\vec{p}}^{(n)})_{n\in\N},\allowbreak m'(n))}{\text{$\Phi$ satisfiable}}=
    \begin{cases}
      1, & \text{if}\ m'\in o_n(m^\star) \\
      0, & \text{if}\ m'\in\omega_n(m^\star).
    \end{cases}
\end{equation*}
We say that a satisfiability threshold with respect to $m$ exists if there is an asymptotic threshold function for satisfiability.
\end{definition}

It is important to realize what this definition actually says.
For example, we know that for random $k$-SAT with $k\ge2$ there is a satisfiability threshold and the asymptotic threshold function is $m^\star=n$~\cite{AchlioptasP04, KirousisKKS98}.
This means, if we draw an instance with $m=n^{0.5}\in o(n)$ clauses, then the probability that this instance is satisfiable is close to one.
If we draw an instance with $m=n^{1.2}\in \omega(n)$ clauses, then the probability that this instance is satisfiable is close to zero.
However, if we draw an instance with $m=\eps\cdot n\pm o(n)$ clauses for any constant $\eps>0$, we do not know what happens.
This is in line with our definition of a satisfiability threshold:
We do not care what happens at $m\in\Theta(n)$, as long as the probability to generate satisfiable instances is $1-o(1)$ for $m\in o(n)$ and $o(1)$ for $m\in \omega(n)$.
See Figure~\ref{fig:asymptotic} for a visual representation.

But what if we want to know what happens at $m\in\Theta(n)$?
There are two ways that the probability function could behave in the range $\Theta(n)$.
Either, there is a small interval of size $o(n)$, where it suddenly drops from close to one to close to zero.
If this is the case, we call the threshold \emph{sharp}.
This is what we observe for random $k$-SAT.
Intuitively, a sharp threshold means that the size of the threshold interval grows asymptotically slower than the actual threshold.
Thus, the threshold interval seems to vanish in the limit, which makes the probability function look more and more like a step function.
If the threshold is not sharp, we call it \emph{coarse}.
This could mean that the function decreases more slowly, in an interval of size $\Theta(n)$, in which the probability function is bounded away from zero and one by a constant.
However, it could also mean that the limit of the probability function is not defined in that region.
Formally, we define sharp and coarse thresholds as follows.

\begin{definition}[Sharpness]\label{def:general-sharpness}
Fix a constant $k\in \N$ and an ensemble of probability distributions $({\vec{p}}^{(n)})_{n\in\N}$ and let $\mathcal{D}^N(n, k, ({\vec{p}}^{(n)})_{n\in\N},\allowbreak m)$ be a non-uniform random k-SAT model.
Let $m^\star$ be an asymptotic threshold function of $\mathcal{D}$. 
We call $m^\star$ \emph{sharp} if for every function $m'$ and every constant $\eps>0$
\begin{equation*}
    \lim_{n\to\infty}\Pro{\Phi\sim \mathcal{D}^N(n, k, ({\vec{p}}^{(n)})_{n\in\N},\allowbreak m'(n))}{\text{$\Phi$ satisfiable}}=
    \begin{cases}
      1, & \text{if}\ m'\le(1-\eps)\cdot m^\star \\
      0, & \text{if}\ m'\ge(1+\eps)\cdot m^\star. 
    \end{cases}
\end{equation*}
Otherwise we call $m^\star$ coarse.
\end{definition}

\begin{figure}[t]
\centering
\begin{subfigure}[t]{.45\linewidth}
	\centering
	\includegraphics[width=\linewidth]{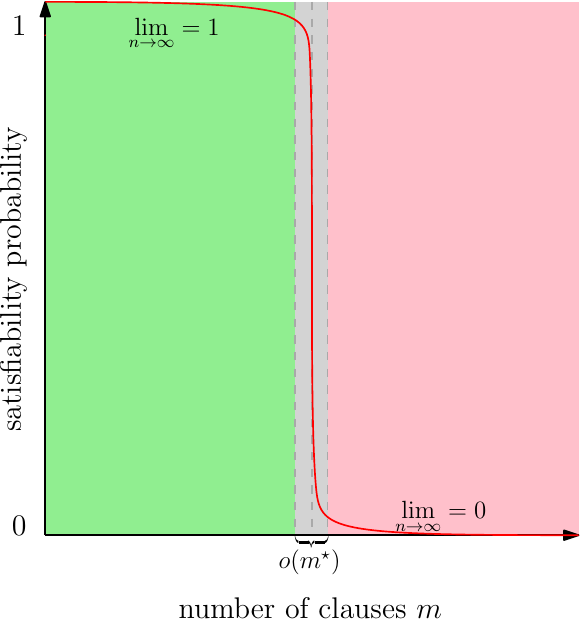}
	\caption{Sharp threshold: $m^\star$ is an asymptotic threshold function. For any constant $\eps>0$ the probability tends to one at $m=(1-\eps)\cdot m^\star$ (green region) and to zero at $m=(1+\eps)\cdot m^\star$ (red region).
	The range where the function is not restricted (gray region) is of size $o(m^\star)$.}
\end{subfigure}
\hfill
\begin{subfigure}[t]{.45\linewidth}
	\centering
	\includegraphics[width=\linewidth]{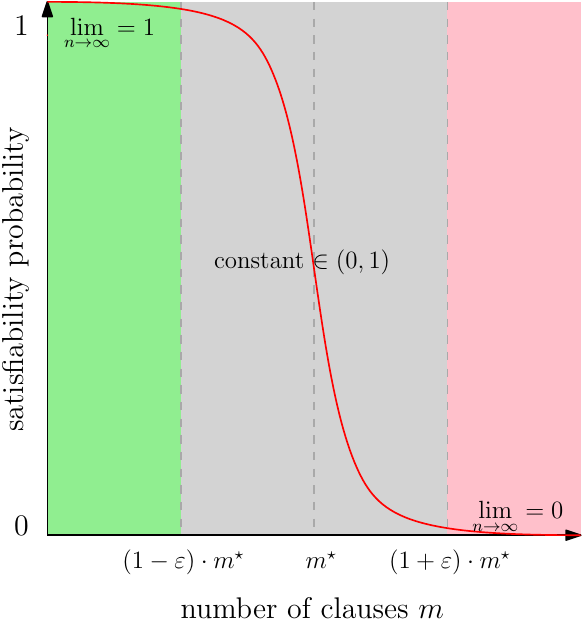}
	\caption{Coarse threshold: $m^\star$ is an asymptotic threshold function. There is a constant $\eps>0$ such that the satisfiability probability is bounded away from zero and one by a constant for all $m\in[(1-\eps)\cdot m^\star, (1+\eps)\cdot m^\star]$.}
\end{subfigure}
\caption{Visual representation of a sharp and a coarse satisfiability threshold with asymptotic threshold function $m^\star$.}
\label{fig:sharp}
\end{figure}

Those two cases are mutually exclusive and one of the two has to hold if an asymptotic threshold function exists.
Thus, if there is a satisfiability threshold, we can always classify it as either sharp or coarse.
See Figure~\ref{fig:sharp} for a visual representation of sharp and coarse thresholds.
\citet{friedgut1999thresholds} proved that random $k$-SAT with $k\ge2$ has a sharp threshold, although he did not determine the exact threshold function.

\section{What we are going to show} \label{sec:what}

First, we are going to discuss which kinds of results we are going to show and why we do not show something more intuitive.
Our results will assume certain relations between the functions $m$ and $m^\star$, $p_1$ and $\sum_{i=1}^n p_i^2$, and $p_2$ and $\sum_{i=2}^n p_i^2$.
Intuitively, those relations would be in terms of Landau notation as is suggested by the results we want to show for non-uniform random 2-SAT.
However, we are only going to assume that functions are smaller or bigger than other functions by some constant factor that is either given or that we can choose.
Instead of assuming $p_1^2\in o(\sum_{i=1}^n p_i^2)$, we only assume that we can choose a constant $\eps_1\in(0,1)$ so that $p_1^2\le\eps_1\cdot\sum_{i=1}^n p_i^2$.
This condition is implied by $p_1^2\in o(\sum_{i=1}^n p_i^2)$ for sufficiently large $n$.
Instead of assuming $p_1^2\in \Theta(\sum_{i=1}^n p_i^2)$, we assume that there is some constant $\eps_1\in(0,1)$ so that $p_1^2\ge\eps_1\cdot\sum_{i=1}^n p_i^2$.
Again, this condition is implied by $p_1^2\in \Theta(\sum_{i=1}^n p_i^2)$ for sufficiently large $n$.
Equivalently, we consistently use the factors $\eps_2$ and $\eps_m$ to define relationships between $p_2$ and $\sum_{i=2}^n p_i^2$, and $m$ and $m^\star$, respectively.
Additionally, we will use the placeholder $\eps$ without any index if we refer to relations from previous results that will only be used in a very local scope.
This is mainly to avoid using the same notation twice.

We choose to use these requirements for our results, because it will allow us to prove something in absence of asymptotic behavior as well, \ie in the fourth case mentioned in \secref{results}.
Remember what we want to show: If we assume the existence of a satisfiability threshold at some function $m^\star$, then 
\begin{equation*}
    \lim_{n\to\infty}\Pro{\Phi\sim \mathcal{D}^N\left(n,2,\left(\vec{p}^{(n)}\right)_{n\in\N},m\right)}{\text{$\Phi$ satisfiable}}=
    \begin{cases}
      1, & \text{if}\ m\in o_n(m^\star) \\
      0, & \text{if}\ m\in \omega_n(m^\star).
    \end{cases}
\end{equation*}
Let us concentrate on the case that $m\in o(m^\star)$.
Remembering the definition of limits, we want for any constant $\eps_P\in(0,1)$ that there is an $n_0\in\N$ so that for all $n\ge n_0$ we have \[\Pro{\Phi\sim \mathcal{D}^N\left(n,2,\left(\vec{p}^{(n)}\right)_{n\in\N},m\right)}{\text{$\Phi$ satisfiable}}\ge\eps_P.\]
We are going to show results of the following flavor:
\begin{enumerate}
\item Assume we can choose $\eps_1\in(0,1)$ with $p_1^2\le\eps_1\cdot \sum_{i=1}^n p_i^2$. Then for any given $\eps_m\in(0,1)$ with $m\le\eps_m\cdot m^\star$ and for any given $\eps_P\in(0,1)$, we can reach a probability of at least $\eps_P$ by choosing $\eps_1$ small enough.\label{enum:first}
\item Assume we are given $\eps_1\in(0,1)$ with $p_1^2\ge\eps_1\cdot \sum_{i=1}^n p_i^2$ and we can choose $\eps_m\in(0,1)$ with $m\le\eps_m\cdot m^\star$. Then, for any given $\eps_P\in(0,1)$, we can choose $\eps_m$ small enough to reach a probability of at least $\eps_P$.
\end{enumerate}
Let us now consider what these two results imply.
Assume we are given some $\eps_P\in(0,1)$ and some $m\in o(m^\star)$.
We have to show that the probability to generate satisfiable instances at $m$ is at least $\eps_P$ for all sufficiently large $n$.
First we note that, if $m\in o(m^\star)$, then for all $\eps_m\in(0,1)$ there is some $n_0\in\N$ so that for all $n\ge n_0$ we have $m\le\eps_m\cdot m^\star$.

If $p_1^2\in o(\sum_{i=1}^n p_i^2)$, then for every $\eps_1\in(0,1)$ there is some $n_0\in\N$ so that $p_1^2\le\eps_1\cdot \sum_{i=1}^n p_i^2$ holds for all $n\ge n_0$.
Due to the first result we can now simply choose some $\eps_m\in(0,1)$ (for example $\eps_m=\frac{1}{2}$) and choose $\eps_1$ small enough so that the resulting probability is at least $\eps_P$.
The requirements $p_1^2\in o(\sum_{i=1}^n p_i^2)$ and $m\in o(m^\star)$ guarantee that there is some $n_0\in\N$ so that both requirements hold for all $n\ge n_0$.

If $p_1^2\in\Theta(\sum_{i=1}^n p_i^2)$, then there are $\eps_1\in(0,1)$ and $n_0\in\N$ so that $p_1^2\ge\eps_1\cdot \sum_{i=1}^n p_i^2$ holds for all $n\ge n_0$.
For this value of $\eps_1$, we can now simply choose an $\eps_m$ small enough so that the resulting probability is at least $\eps_P$.
Again, this requirement is fulfilled for all sufficiently large $n$, since $m\in o(m^\star)$.

The last case is that neither $p_1^2\in o(\sum_{i=1}^n p_i^2)$ nor $p_1^2\in \Theta(\sum_{i=1}^n p_i^2)$.
First, we assume to be able to choose $\eps_1$.
Like in the first case, we choose some $\eps_m\in(0,1)$ (\eg $\eps_m=\frac{1}{2}$) and evaluate how small $\eps_1$ has to be in order to have a probability of at least $\eps_P$ due to our first result.
For sufficiently large $n$ the requirement $m\le\eps_m\cdot m^\star$ will be fulfilled.
However, $p_1^2\le\eps_1\cdot \sum_{i=1}^n p_i^2$ might not.
For all values of $n$ that fulfill the requirement, we are done already.
Thus, we only have to consider what happens if $p_1^2>\eps_1\cdot \sum_{i=1}^n p_i^2$.
For those values of $n$ we can simply use the second result and evaluate an $\eps_m\in(0,1)$ small enough so that we reach the given probability $\eps_P$.
Again, this requirement is fulfilled for sufficiently large $n$ due to $m\in o(m^\star)$.

This was a simplified example of what our results will look like and how we are going to use them to show the statements in \secref{results}.
We will show our results for slightly different functions $m^\star$.
However, these functions will asymptotically coincide as we will see later.

\subsection{How we are going to show it}

Another note on how we will derive our results might be necessary at this point.
As stated in \secref{non-uniform}, the probability to draw a certain clause is proportional to the product of variable probabilities for Boolean variables it contains. For example, the probability to draw a clause $c=(X_i\vee \n{X_j})$ is
\begin{equation}\label{eq:2-sat-clause}
\Pr{(X_i\vee \n{X_j})}=q_c=\frac{C}{2}\cdot p_i\cdot p_j,
\end{equation}
where $C=1/\left(k!\cdot\sum_{J\in{\mathcal{P}_k\left(\left\{1,2,\ldots,n\right\}\right)}}{\prod_{j\in{J}}{p_{j}}}\right)$ is the same for all clauses.
However, for $k=2$ the factor $C$ simplifies to $C=1/\left(1-\sum_{i=1}^n p_i^2\right)$.

Since clauses are also drawn independently, it holds that the probability of drawing a certain formula in non-uniform random $k$-SAT is proportional to the product of variable probabilities for each appearance of a Boolean variable in it.
For example, the probability of drawing $\Phi=(X_i\vee \n{X_j})\wedge(\n{X_h}\vee \n{X_i})$ is 
\[\Pro{\Phi\sim\mathcal{D}^N\left(n, 2, \left(\vec{p}^{(n)}\right)_{n\in\N}, 2\right)}{\Phi=(X_i\vee \n{X_j})\wedge(\n{X_h}\vee \n{X_i})}=\left(\frac{C}{2}\right)^2 \cdot p_h\cdot p_i^2\cdot p_j.\]
We will use this fact heavily in our analysis.
If we want to know the probability of drawing a formula, we only have to know which Boolean variables it contains how often.
Furthermore, we can also use this fact if we do not know the exact variables, but only how often they appear.
For example, if we are looking for a formula $\Phi=(\ell_i\vee \n{\ell_j})\wedge(\n{\ell_h}\vee \n{\ell_i})$, where $\ell_h$, $\ell_i$, and $\ell_j$ are literals of distinct Boolean variables, the probability is proportional to
\[\left(\sum_{h=1}^n p_h\right)\cdot\left(\sum_{\substack{i=1\\i\neq h}}^n p_i^2\right)\cdot\left(\sum_{\substack{j=1\\ j\neq h,\ j\neq i}}^n p_j\right)
\le\left(\sum_{h=1}^n p_h\right)\cdot\left(\sum_{i=1}^n p_i^2\right)\cdot\left(\sum_{j=1}^n p_j\right).\]

The following lemma shows how we can bound expressions of that kind.
It applies to situations where a set of variables all appear the same number of times and we already accounted for the possible ways to arrange them in clauses.
For example, if we want the probability for a formula $\Phi=(\ell_h\vee {\ell_i})\wedge({\ell_i}\vee {\ell_j})\wedge({\ell_j}\vee {\ell_h})$, where $\ell_h$, $\ell_i$, and $\ell_j$ are literals of distinct Boolean variables, the probability is proportional to
\[3!\cdot\sum_{\substack{A\in\mathcal{P}_3([n])}} \prod_{a\in A} p_a^2,\]
where $3!$ accounts for the possibilities to interchange the chosen variables.

\begin{lemma}\label{lem:aux1}
For every set $S\subseteq\left\{1,\ldots,n\right\}$, every integer $i\le|S|$, and every integer $l\ge 1$ it holds that
\[\sum_{A\in \mathcal{P}_i(S)} \prod_{a \in A} p_a^l \le \frac{1}{i!} \left(\sum_{s\in S} p_s^l\right)^i.\]
If $\left(\sum_{s\in S} p_s^l\right)\ge (i-1)\cdot \max\left\{p_s^l \mid s\in S\right\}$, then it also holds that
\[\sum_{A\in \mathcal{P}_i(S)} \prod_{a \in A} p_a^l \ge \frac{1}{i!} \left(\left(\sum_{s\in S} p_s^l\right)-(i-1)\cdot \max\left\{p_s^l \mid s\in S\right\}\right)^i.\]
\end{lemma}
\begin{proof}
For the first part, notice that each product $\prod_{a \in A} p_a^l$ for some $A\in \mathcal{P}_i(S)$ appears $i!$-times in $\left(\sum_{s\in S} p_s^l\right)^i$.
For the second part, $\sum_{A\in \mathcal{P}_i(S)} \prod_{a \in A} p_a^l$ can be expressed as the following nested sum
\[\sum_{A\in \mathcal{P}_i(S)} \prod_{a \in A} p_a^l = \frac{1}{i!}\cdot\sum_{a_1\in A} \left(p_{a_1}^l \cdot \sum_{a_2\in A\setminus\left\{a_1\right\}} \left(p_{a_2}^l \cdot \ldots\cdot\sum_{a_{i}\in A\setminus\left\{a_0, \ldots, a_{i-1}\right\}} p_{a_i}^l\right)\right).\]
This sum essentially captures the choices of elements we have for each term, where $a_j$ is the $j$-th chosen element for $j=1,\ldots,i$.
Since we only forbid repetitions of elements, the $j$-th element can be anything from $S\setminus\left\{a_1,a_2,\ldots,a_{j-1}\right\}$.
Again, we generate each product $i!$ times on the right-hand side.
If we pessimistically assume that forbidden elements have the maximum value $\max\left\{p_s^l \mid s\in S\right\}$, we get
\begin{align*}
&\sum_{A\in \mathcal{P}_i(S)} \prod_{a \in A} p_a^l \\
	& \ge \frac{1}{i!}\sum_{a_1\in A}\left( p_{a_1}^l \cdot \sum_{a_2\in A\setminus\left\{a_1\right\}}\left( p_{a_2}^l \cdot \ldots\cdot\left(\left(\sum_{s\in S} p_s^l\right)-(i-1)\cdot \max\left\{p_s^l \mid s\in S\right\}\right)\right)\right)\\
	& \ge \frac{1}{i!} \left(\left(\sum_{s\in S} p_s^l\right)-(i-1)\cdot \max\left\{p_s^l \mid s\in S\right\}\right)^i.
\end{align*}
Now we also see why the requirement for this second statement is necessary.
\end{proof}

We will use the bounds of the former lemma heavily in the remainder of this work.

\section{Bi-Cycles and a Lower Bound on the Satisfiability Threshold}\label{sec:bicycle}

In this section we introduce the concept of bicycles and derive a lower bound on the position of the satisfiability threshold.

\citet{chvatalreed92} define the following sub-structure of 2-SAT formulas and show that every unsatisfiable formula in 2-CNF contains this substructure.
\begin{definition}[bicycle]
Let $X_1,X_2,\ldots,X_t$ be $t$ distinct Boolean variables and let $w_1,w_2,\ldots,w_t$ be literals such that each $w_l$ is either $X_l$ or $\n{X_l}$.
We define a bicycle of length $t$ to be a sequence of $t+1$ clauses of the form 
\[\left(u,w_1\right),\left(\n{w_1},w_2\right),\ldots,\left(\n{w_{t-1}},w_t\right),\left(\n{w_t},v\right),\]
where $u,v\in\left\{w_1,\ldots,w_t,\n{w}_1,\ldots,\n{w}_t\right\}$.
\end{definition}
Although a bicycle itself might not be unsatisfiable, \citet{chvatalreed92} prove that \emph{every unsatisfiable Boolean formula in 2-CNF must contain a bicycle}.
We can use this knowledge in the following way:
We show that up to a certain number of clauses $m^\star$ the random formulas our model generates \aas do not contain \emph{any} bicycles.
Thus, they must be satisfiable.
In order to bound the probability for bicycles to appear, we use the first moment method.
This means, we bound the expected number of bicycles that appear.
If this number is $o(1)$, we can use Markov's inequality (c.f. \cite{MitzenmacherU05}) to bound the probability of them appearing as desired.
The same approach was used in the proof of Theorem~3 from \cite{chvatalreed92}.

First, we consider the case $p_1^2\in o(\sum_{i=1}^n p_i^2)$.
We define our threshold function to be $m^\star=(\sum_{i=1}^n p_i^2)^{-1}$.
We want to show that this function defines a sharp satisfiability threshold for non-uniform random $2$-SAT.
Remember our definition of a sharp satisfiability threshold.
We need to show that for any constant $\eps_m\in(0,1)$ and all functions $m\le\eps_m\cdot m^\star$ the probability to generate a satisfiable instance is a function tending to one as $n$ increases.
However, as we wrote in the last section, we are going to show something a bit more general.
We will show that, given $\eps_m\in(0,1)$ and $\eps_P\in(0,1)$, we can choose a constant $\eps_1\in(0,1)$ with $p_1^2\le\eps_{1}\cdot(\sum_{i=1}^n p_i^2)$ small enough so that the probability to generate a satisfiable instance is at least $\eps_P$.
If $p_1^2\in o\left(\sum_{i=1}^n p_i^2\right)$, then there is an $n_0\in\N$ such that this condition holds for all $n\ge n_0$.

\begin{lemma} \label{lem:sharp-lb}
Given an ensemble of probability distributions $\left(\vec{p}^{(n)}\right)_{n\in\N}$.
Let $m^\star=1/\sum_{i=1}^n p_i^2$.
Then, for any constant $\eps_m\in(0,1)$ and any constant $\eps_P\in(0,1)$ we can choose $\eps_1\in(0,1)$ such that $p_1^2\le\eps_{1}\cdot(\sum_{i=1}^n p_i^2)$ implies
\[\Pro{\Phi\sim\mathcal{D}^N\left(n,2,\left(\vec{p}^{(n)}\right)_{n\in\N},m\right)}{\Phi\text{ satisfiable}}\ge \eps_P\]
for any $m\le\eps_m\cdot m^\star$.
\end{lemma}

\begin{proof}
To show this result, we show that the expected number of bicycles is at most $1-\eps_P$ for the setting we consider. 
The result then follows by Markov's inequality.

First, choose $n$ arbitrary, but fixed.
We want to evaluate the value of the probability function for this value of $n$ and the number of clauses prescribed by the clause function $m(n)$.
We fix a set $S\subseteq\left[n\right]$ of variables to appear in a bicycle with $\left|S\right|=t\ge 2$ .
The probability that a \emph{specific} bicycle $B$ with these variables appears in $\Phi$ is
\[\Prob{B\text{ in }\Phi} \le \underbrace{\binom{m}{t+1}\cdot(t+1)!}_{\text{positions of $B$ in $\Phi$}}\cdot\Prob{(u\vee w_1)}\cdot\Prob{(\n{w_t}\vee v)}\cdot\prod_{h=1}^{t-1}{\Prob{(\n{w_h}\vee w_{h+1})}}.\]
$\Prob{(w_j\vee w_i)}$ denotes the probability to draw clause $(w_j\vee w_i)$ in non-uniform random $2$-SAT.
There are at most $t!$ possibilities to arrange the $t$ variables in a bicycle and $2^t$ possibilities to choose literals from the $t$ variables. 
For the probability that \emph{any} bicycle with the variables from $S$ appears in $\Phi$ it now holds that
\[\Prob{S\text{-bicycle in }\Phi} \le m^{t+1}\cdot t!\cdot 2^t\cdot\left(\frac{C}{2}\right)^{t+1}\cdot\prod_{i\in S}{p_i^2}\cdot\left(2\cdot\sum_{i\in S}{p_i}\right)^2,\]
where the last factor accounts for the possibilities to choose $u$ and $v$.
It now holds that

\begin{align*}
\Prob{\text{bicycle in }\Phi} 
&\le \sum_{t=2}^{n}\left({\sum_{S\in \mathcal{P}_t([n])}\left(m^{t+1}\cdot t!\cdot 2^t\cdot\left(\frac{C}{2}\right)^{t+1}2^2\cdot\prod_{i\in S}{p_i^2}\cdot\left(\sum_{i\in S}{p_i}\right)^2\right)}\right).\\
\intertext{If we estimate $\sum_{i\in S}{p_i}\le t\cdot p_1$, we get}
&\le 2\cdot\sum_{t=2}^{n}\left(\left(C\cdot m\right)^{t+1}\cdot t!\cdot t^2\cdot p_1^2\cdot\sum_{S\in \mathcal{P}_t([n])}\left(\prod_{i\in S}{p_i^2}\right)\right)\\
\intertext{and with $\sum_{S\in \mathcal{P}_t([n])}\left(\prod_{i\in S}{p_i^2}\right)\le \frac{1}{t!}\cdot\left(\sum_{i=1}^n{p_i^2}\right)^t$ due to \lemref{aux1} this yields}
&\le 2\cdot\sum_{t=2}^{n}\left(\left(C\cdot m \right)^{t+1}\cdot t^2\cdot p_1^2\cdot \left(\sum_{i=1}^n{p_i^2}\right)^t\right).\\
\intertext{Since $m\le\eps_m\cdot m^\star=\frac{\eps_m}{\sum_{i=1}^n{p_i^2}}$, this is}
&\le 2\cdot\frac{p_1^2}{\sum_{i=1}^n{p_i^2}}\cdot\sum_{t=2}^{n}\left(C\cdot \eps_m\right)^{t+1}\cdot t^2.
\end{align*}

Now, it holds that $C=\frac{1}{1-\sum_{i=1}^n p_i^2}\le 1+\frac{p_1}{1-p_1}$, since $\sum_{i=1}^n p_i^2\le p_1$. 
Thus,

\begin{align*}
\Prob{\text{bicycle in }\Phi} 
&\le 2\cdot\frac{p_1^2}{\sum_{i=1}^n{p_i^2}}\cdot\sum_{t=2}^{n}\left(\left(\left(1+\frac{p_1}{1-p_1}\right)\cdot \eps_m\right)^{t+1}\cdot t^2\right)\\
&\le 2\cdot\frac{p_1^2}{\sum_{i=1}^n p_i^2}\cdot\sum_{t=2}^{\infty}\left(\left(\left(1+\frac{p_1}{1-p_1}\right)\cdot \eps_m\right)^{t+1}\cdot t^2\right).
\end{align*}

We know that $p_1\le \sqrt{\eps_1\cdot \sum_{i=1}^n p_i^2}\le \sqrt{\eps_1}$.
Thus, if we choose $\eps_1$ small enough such that 
\[\left(1+\frac{p_1}{1-p_1}\right)\cdot \eps_m\le \left(1+\frac{\sqrt{\eps_1}}{1-\sqrt{\eps_1}}\right)\cdot \eps_m<1,\]
then 
\[\left(\sqrt{\left(1+\frac{\sqrt{\eps_1}}{1-\sqrt{\eps_1}}\right)\cdot \eps_m}\right)^{t+1}\cdot t^2\in o(1).\]
Thus, there is some $t_0$ such that for all $t\ge t_0$ this function is at most $1$.
Therefore,
\begin{align*}
\Prob{\Phi\text{ contains a bicycle}} 
&\le 2\cdot\frac{p_1^2}{\sum_{i=1}^n p_i^2}\cdot\sum_{t=2}^{\infty}\left(\left(1+\frac{p_1}{1-p_1}\right)\cdot \eps_m\right)^{t+1}\cdot t^2\\
&\le 2\cdot\frac{p_1^2}{\sum_{i=1}^n p_i^2}\cdot\left(t_0^3+\sum_{t=t_0}^{\infty}\left(\sqrt{\left(1+\frac{\sqrt{\eps_1}}{1-\sqrt{\eps_1}}\right)\cdot \eps_m}\right)^{t+1}\right)\\
&\le 2\cdot\eps_1\cdot\left(t_0^3+\frac{\sqrt{(1+\frac{\sqrt{\eps_1}}{1-\sqrt{\eps_1}})\cdot \eps_m}}{1-\sqrt{(1+\frac{\sqrt{\eps_1}}{1-\sqrt{\eps_1}})\cdot \eps_m}}\right),
\end{align*}
where the second term was bounded by a geometric series.
If we choose $\eps_1$ sufficiently small, this expression is at most $1-\eps_P$.
\end{proof}

We now turn to the case that $p_1^2\notin o(\sum_{i=1}^n p_i^2)$.
We are going to show that there is an asymptotic threshold at 
\[m^\star=\left(C\cdot \left(\sum_{i=2}^n p_i^2\right)+C\cdot p_1\cdot \left(\sum_{i=2}^n p_i^2\right)^{1/2}\right)^{-1}.\]
However, we are going to show something a bit more general.
We only assume that $p_1^2\ge\eps_1\cdot(\sum_{i=1}^n p_i^2)$ for some constant $\eps_1>0$.
If $p_1^2\in\Theta(\sum_{i=1}^n p_i^2)$, then there is some $n_0\in\N$ such that this holds for all $n\ge n_0$.
Under this condition, we will show that for any $\eps_P\in(0,1)$ we can choose an $\eps_m\in(0,1)$ with $m\le\eps_m\cdot(C\cdot p_1\cdot (\sum_{i=2}^n p_i^2)^{1/2})^{-1}$ so that the probability to generate a satisfiable instance is at least $\eps_P$.
If $m\in o\left((C\cdot p_1\cdot (\sum_{i=2}^n p_i^2)^{1/2})^{-1}\right)$, this condition is met for all sufficiently large $n$.
However, it is also met if $m\in o(m^\star)$.
We will show this in more detail in \secref{main-theorem}.

\begin{lemma}\label{lem:bicycle-coarse}
Given an ensemble of probability distributions $\left(\vec{p}^{(n)}\right)_{n\in\N}$ with $p_1^2\ge\eps_1\cdot(\sum_{i=1}^n p_i^2)$ for some constant $\eps_1\in(0,1)$.
Let $m^\star=(C\cdot p_1\cdot (\sum_{i=2}^n p_i^2)^{1/2})^{-1}$.
Then, for any $\eps_P\in(0,1)$ we can choose an $\eps_m\in(0,1)$ such that the probability to generate a satisfiable formula $\Phi\sim\mathcal{D}^N(n,2,(\vec{p}^{(n)})_{n\in\N},m)$ is at least $\eps_P$ if $m\le\eps_m\cdot m^\star$.
\end{lemma}
\begin{proof}
As in the proof of \lemref{sharp-lb} it holds that
\begin{align}
\Prob{\Phi\text{ unsat}}
&\le \Prob{\text{bicycle in }\Phi}\notag\\
&\le \sum_{t=2}^{n}\left(\sum_{S\in \mathcal{P}_t([n])}\left(m^{t+1}\cdot t!\cdot 2^t\cdot\left(\frac{C}{2}\right)^{t+1}2^2\cdot\prod_{i\in S}{p_i^2}\cdot\left(\sum_{i\in S}{p_i}\right)^2\right)\right)\notag\\
&\le 2\cdot\sum_{t=2}^{n}\left(\left(C\cdot m\right)^{t+1}\cdot t!\cdot\sum_{S\in \mathcal{P}_t([n])}\left(\prod_{i\in S}{p_i^2}\right)\cdot\left(\sum_{i\in S}p_i\right)^2\right) \label{eq:bicycle-bound}.
\end{align}
We can analyze the term $\sum_{S\in \mathcal{P}_t([n])}\left(\left(\prod_{i\in S}{p_i^2}\right)\cdot\left(\sum_{i\in S}p_i\right)^2\right)$ in more detail.
By doing a case distinction between the terms with $p_1\in S$ and $p_1\notin S$ we get
\begin{align}
\sum_{S\in \mathcal{P}_t([n])}\left(\left(\prod_{i\in S}{p_i^2}\right)\cdot\left(\sum_{i\in S}p_i\right)^2\right)
& \le p_1^2\cdot t^2\cdot p_1^2 \cdot \frac{1}{(t-1)!}\cdot\left(\sum_{i=2}^n p_i^2\right)^{t-1} + t^2 \cdot p_{2}^2 \cdot \frac{1}{t!}\cdot\left(\sum_{i=2}^n p_i^2\right)^{t}\notag\\
\intertext{and since $p_1^2\ge\eps_1\cdot(\sum_{i=1}^n p_i^2)\ge\eps_1\cdot(\sum_{i=2}^n p_i^2)$ and $p_{2}\le p_1$ this yields}
& \le \left(1+1/\eps_1\right)\cdot t^3\cdot p_1^4 \cdot \frac{1}{t!}\cdot\left(\sum_{i=2}^n p_i^2\right)^{t-1}.\notag
\end{align}
It holds that $p_1^4\cdot\left(\sum_{i=2}^n p_i^2\right)^{t-1}\le\left(\frac{1}{\sqrt{\eps_1}}\cdot p_1\cdot\left(\sum_{i=2}^n p_i^2\right)^{1/2}\right)^{t+1}$ for $t\ge3$. 
This yields
\begin{equation}
\sum_{S\in \mathcal{P}_t([n])}\left(\left(\prod_{i\in S}{p_i^2}\right)\cdot\left(\sum_{i\in S}p_i\right)^2\right) \le  \left(1+1/\eps_1\right)\cdot\frac{t^3}{t!}\cdot\left(\frac{1}{\sqrt{\eps_1}} \cdot p_1\cdot\left(\sum_{i=2}^n p_i^2\right)^{1/2}\right)^{t+1}\label{eq:detailed-bound}
\end{equation}
for $t\ge3$.

For $t=2$ we know that each of the three 2-clauses in the bicycle must contain both variables.
Thus,
\begin{align}
\sum_{S\in{\mathcal{P}_2([n])}}\Prob{S\text{-bicycle in }F}
&\le m^3\cdot \left(C/2\right)^3\cdot 2!\cdot 2^2 \cdot 2^2\sum_{\substack{{i,j\in V}\colon\\i\neq j}}{p_i^3\cdot p_j^3}\notag\\
& \le (C\cdot m)^{3}\cdot t^2\cdot p_1^3\cdot \left(\sum_{i=2}^n p_i^3\right)+ (C\cdot m)^{3}\cdot t^2\left(\sum_{i=2}^n p_i^3\right)^2\notag\\
\intertext{and since $\sum_{i=2}^n p_i^3\le\left(\sum_{i=2}^n p_i^2\right)^{3/2}$ due to the monotonicity of vector norms, this is at most}
& \le (C\cdot m)^{3}\cdot t^2\cdot p_1^3\cdot \left(\sum_{i=2}^n p_i^2\right)^{3/2}+(C\cdot m)^{3}\cdot t^2\cdot\left(\sum_{i=2}^n p_i^2\right)^{3}\notag\\
\intertext{and due to our condition $p_1^2\ge\eps_1\cdot(\sum_{i=1}^n p_i^2)\ge\eps_1\cdot(\sum_{i=2}^n p_i^2)$, we get}
& \le \left(1+\frac{1}{\eps_1^{3/2}}\right)\cdot (C\cdot m)^{3}\cdot t^2\cdot p_1^3\cdot \left(\sum_{i=2}^n p_i^2\right)^{3/2}\notag\\
& \le (C\cdot m)^{t+1}\cdot \left(1+1/\eps_1\right)\cdot t^3 \cdot\left(\frac{1}{\sqrt{\eps_1}} \cdot p_1\cdot\left(\sum_{i=2}^n p_i^2\right)^{1/2}\right)^{t+1}.\label{eq:detailed-bound2}
\end{align}

We can now plug \eq{detailed-bound} and \eq{detailed-bound2} into \eq{bicycle-bound} to get
\begin{align*}
\Prob{\Phi\text{ unsat}}
&\le 2\cdot \left(1+1/\eps_1\right)\sum_{t=2}^{n}\left(\left(C\cdot m\cdot \frac{1}{\sqrt{\eps_1}} \cdot p_1\cdot \left(\sum_{i=2}^n p_i^2\right)^{1/2}\right)^{t+1}\cdot t^3\right)\\
&\le 2\cdot \left(1+1/\eps_1\right)\sum_{t=2}^{\infty}\left(\left(C\cdot m\cdot \frac{1}{\sqrt{\eps_1}} \cdot p_1\cdot \left(\sum_{i=2}^n p_i^2\right)^{1/2}\right)^{t+1}\cdot t^3\right).
\end{align*}
We can now choose $m\le\eps_m\cdot m^\star$ for a constant $\eps_m\in(0,1)$ to be determined later.
Then,
\begin{align*}
\Prob{\Phi\text{ unsat}}
&\le 2\cdot \left(1+1/\eps_1\right)\sum_{t=2}^{\infty}\left(\frac{\eps_m}{\sqrt{\eps_1}}\right)^{t+1}\cdot t^3\\
&\le 2\cdot \left(1+1/\eps_1\right)\cdot\frac{\eps_m}{\sqrt{\eps_1}}\sum_{t=2}^{\infty}\left(\frac{\eps_m}{\sqrt{\eps_1}}\right)^{t}\cdot t^3.
\end{align*}
If we choose $\eps_m$ small enough so that $\frac{\eps_m}{\sqrt{\eps_1}}<1$, it holds that $\left(\sqrt{\frac{\eps_m}{\sqrt{\eps_1}}}\right)^{t}\cdot t^3=o(1)$.
Thus, there is a $t_0\in\N$ so that this function is at most $1$ for all $t\ge t_0$.
As in the proof of \lemref{sharp-lb}, we have
\begin{align*}
\Prob{\Phi\text{ unsat}}
&\le 2\cdot \left(1+1/\eps_1\right)\cdot\frac{\eps_m}{\sqrt{\eps_1}}\cdot\sum_{t=2}^{\infty}\left(\frac{\eps_m}{\sqrt{\eps_1}}\right)^{t}\cdot t^3\\
&\le 2\cdot \left(1+1/\eps_1\right)\cdot\frac{\eps_m}{\sqrt{\eps_1}}\cdot\left(t_0^4+\sum_{t=t_0}^{\infty}\left(\sqrt{\frac{\eps_m}{\sqrt{\eps_1}}}\right)^{t}\right)\\
&\le 2\cdot \left(1+1/\eps_1\right)\cdot\frac{\eps_m}{\sqrt{\eps_1}}\cdot\left(t_0^4+\frac{\sqrt{\frac{\eps_m}{\sqrt{\eps_1}}}}{1-\sqrt{\frac{\eps_m}{\sqrt{\eps_1}}}}\right).
\end{align*}
We can now choose $\eps_m$ small enough so that this probability is at most $1-\eps_P$.
\end{proof}

\lemref{sharp-lb} and \lemref{bicycle-coarse} imply the statements we want for $m\le \eps_m\cdot m^\star$ and $m\in o(m^\star)$ respectively.
We will show this formally in \secref{main-theorem}.

\section{Snakes and the Second Moment Method}\label{sec:snake}

The two lemmas from the previous section provide a lower bound on the satisfiability threshold for non-uniform random 2-SAT.
By using the second moment method, we can also derive an upper bound. 
This proof is inspired by Chvatal and Reed~\cite[Theorem~4]{chvatalreed92}, who provide us with the following definition.

\begin{definition}[snake]
A \emph{snake} of size $t\ge2$ is a sequence of literals $(w_1, w_2, \ldots, w_{2t-1})$ over distinct variables.
Each snake $A$ is associated with a set $F_A$ of $2t$ clauses $\left(\n{w_i}, w_{i+1}\right)$, $0\le i\le 2t-1$, such that $w_0=w_{2t}=\n{w_t}$.
\end{definition}
We will also call the variable $|w_t|$ of a snake its \emph{central} variable.
Note that the set of clauses $F_A$ defined by a snake $A$ is unsatisfiable.
Also, the snakes 
\begin{align*}
& (w_1,\ldots, w_{t-1}, w_{t}, w_{t+1}, \ldots, w_s), \\
&(\n{w_{t-1}}, \n{w_{t-2}}, \ldots, \n{w_1}, w_{t}, w_{t+1}, \ldots, w_s), \\
& (w_1, \ldots, w_{t-1}, w_{t}, \n{w_s}, \n{w_{s-1}}, \ldots, \n{w_{t+1}}),\text{ and}\\
& (\n{w_{t-1}}, \n{w_{t-2}}, \ldots, \n{w_1}, w_{t}, \n{w_s}, \n{w_{s-1}}, \ldots, \n{w_{t+1}})
\end{align*}
create the same set of clauses.

The \emph{variable-variable incidence graph (VIG)} for a formula $\Phi$ is a simple graph $G_\Phi=(V_\Phi,E_\Phi)$ with $V_\Phi$ consisting of all variables appearing in $\Phi$ and two variables being connected by an edge if they appear together in at least one clause of $\Phi$. 
An example for a snake's VIG can be seen in \figref{snake}.
We will use this representation later in the proof of \lemref{threshold-sharp}.

\begin{figure}
	\centering \includegraphics{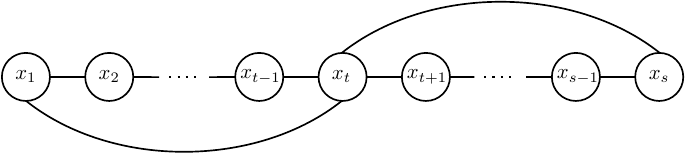}
	\caption{Variable-variable-incidence graph of a snake $w_1,w_2,\ldots,w_s$ where $|w_i|=x_i$ (the variable of the literal $w_i$) for $1\le i\le s=2t-1$.}\label{fig:snake}
\end{figure}

In order to show our upper bounds, we will prove that snakes of a certain length $t$ appear with sufficiently high probability in a random formula $\Phi\sim\mathcal{D}(n,2,(\vec{p}^{(n)})_{n\in\N},m)$.
To this end we utilize the second moment method:
If $X\ge0$ is a random variable with finite variance, then
\[\Prob{X>0}\ge\frac{\Ex{X}^2}{\Ex{X^2}}.\]
We define the following indicator variables for each snake $A$ of size $t$
\[X_A=
\begin{cases}
1&\text{if $F_A$ appears exactly once in $\Phi$}\\
0&\text{otherwise}
\end{cases}\]
and their sum $X_t=\sum_{A}{X_A}$.
Throughout the rest of this chapter we let $X_t$ denote the number of snakes of size $t\ge2$ whose associated clauses appear exactly once in a non-uniform random 2-SAT formula $\Phi\sim\mathcal{D}^N(n,2,(\vec{p}^{(n)})_{n\in\N},m)$.

As before, if we define $t\colon \N\to\R^+$ as a function in $n$, it holds that $\Ex{X_t^2}$ and $\Ex{X_t}$ are functions in $n$ as well.
For carefully chosen functions $t$ we will show that $\Ex{X_t^2}\le(1+\eps_E)\cdot \Ex{X_t}^2$ for a sufficiently small constant $\eps_E>0$.
Note that for the probability to generate an unsatisfiable instance, it is sufficient to show a large enough lower bound for \emph{any} value of $t$.
Thus, we will consider several values of $t$, one of which is guaranteed to give us a bound as desired for sufficiently large values of $n$.
More precisely, there are two values of $t$ that are relevant for us, $t=2$ and $t=f^{1/78}$, where we define 
\[f=\frac{\sum_{i=1}^n p_i^2}{p_1^2}.\]
Note that $t$ and $f$ are both functions in $n$ as are $\left(\sum_{i=1}^n p_i^2\right)$ and $p_1$.

$t=2$ will provide the desired result if there is a constant $\eps_1\in(0,1)$ such that $p_1^2\ge \eps_1\cdot \sum_{i=1}^n p_i^2$ and if we can choose a sufficiently small $\eps_2\in(0,1)$ such that $p_2^2\le\eps_2\cdot\sum_{i=2}^n p_i^2$.
This especially includes the case $p_1^2\in\Theta(\sum_{i=1}^n p_i^2)$ and $p_2^2\in o(\sum_{i=2}^n p_i^2)$.

$t=f^{1/78}$ will provide the desired result if we can choose a sufficiently small $\eps_1\in(0,1)$ with $p_1^2\le \eps_1\cdot \sum_{i=1}^n p_i^2$.
This includes the case $p_1^2\in o(\sum_{i=1}^n p_i^2)$.

However, if there are constants $\eps_1,\eps_2\in(0,1)$ so that $p_1^2\ge \eps_1\cdot \sum_{i=1}^n p_i^2$ and $p_2^2\ge \eps_1\cdot \sum_{i=2}^n p_i^2$, we can show a lower bound directly without having to use the second moment method.
We will handle this case in \secref{coarse}.

Now, if we want to use the second moment method, we first have to ensure that the expected number of snakes of a certain size is large enough.
The following lemma provides a lower bound on this expected number.

\begin{lemma}\label{lem:exp-snakes}
Let $2\le t \le (2\cdot q_{\max})^{-1}$.
Then it holds that
\begin{eqnarray*}
\Ex{X_t}	&\ge	& \frac12\cdot(m-2t)^{2t}\cdot C^{2t}\cdot \left(1-2t\cdot q_{\max}\cdot m\right)\cdot\left(\sum_{i=1}^n p_i^4\right)\cdot\left(\sum_{i=2}^n\left( p_i^2-(2t-3)\cdot p_{2}^2\right)\right)^{2t-2}.
\end{eqnarray*}
\end{lemma}
\begin{proof}
It holds that
\[\Ex{X_t} = \sum_{\substack{\text{snake}\\A=(w_1,\ldots,w_{2t-1})}}\left(\binom{m}{2t}\cdot(2t)!\cdot\prod_{i=0}^{2t-1}\Pr{\n{w_i},w_{i+1}}\cdot\left(1-\sum_{c\in F_A}\Pr{c}\right)^{m-2t}\right)\]
and according to \eq{2-sat-clause} it holds that $\Prob{(\n{w_i},w_{i+1})}=\frac{C}{2}\cdot p(|w_i|)\cdot p(|w_{i+1}|)$. 
Together with the fact that $\sum_{c\in F_A}\Prob{c}\le 2t\cdot q_{\max}$, we get
\begin{equation}
\Ex{X_t}	\ge  (m-2t)^{2t}\left(1-2t\cdot q_{\max}\right)^{m-2t}\left(\frac{C}{2}\right)^{2t}\cdot\hspace{-0.5cm}\sum_{\substack{\text{snake}\\A=(w_1,\ldots,w_{2t-1})}}\left(p(|w_t|)^4\cdot\prod_{\substack{i=1\\i\neq t}}^{2t-1}p(|w_i|)^2\right).\label{eq:exp-snake1}
\end{equation}
Now we count how many snakes of size $t$ there are. 
First, we choose a central variable $X_j$.
Then, we choose a set $S$ of $2t-2$ \emph{different variables}.
From those variables we can create $2^{2t-1}\cdot(2t-2)!$ different snakes by choosing signs for the $2t-1$ variables and by permuting the order of the $2t-2$ non-central variables.
\[\sum_{\substack{\text{snake}\\A=(w_1,\ldots,w_{2t-1})}}\left(p(|w_t|)^4\cdot\prod_{\substack{i=1\\i\neq t}}^{2t-1}p(|w_i|)^2\right) \ge 2^{2t-1}(2t-2)!\cdot\sum_{j=1}^{n} \left(p_j^4\cdot\hspace{-1ex}\sum_{\substack{S\subseteq{[n]\setminus\left\{j\right\}}\colon\\|S|=2t-2}}\prod_{s\in S}p_s^2\right).\]
Due to \lemref{aux1} we have 
\[\sum_{\substack{S\subseteq{[n]\setminus\left\{j\right\}}\colon\\|S|=2t-2}}\prod_{s\in S}p_s^2 \ge \frac{1}{(2t-2)!}\left(\sum_{i=2}^{n}\left(p_i^2-(2t-3)\cdot p_{2}^2\right)\right)^{2t-2}\] and thus
\begin{multline}
\sum_{\substack{\text{snake}\\A=(w_1,\ldots,w_{2t-1})}}\left(p(|w_t|)^4\cdot\prod_{\substack{i=1\\i\neq t}}^{2t-1}p(|w_i|)^2\right)
\ge 2^{2t-1}\cdot\left(\sum_{j=1}^{n} p_j^4\right)\cdot\left(\sum_{i=2}^{n}\left(p_i^2-(2t-3)\cdot p_{2}^2\right)\right)^{2t-2}.\label{eq:exp-snake2}
\end{multline}
Due to Bernoulli's inequality, it also holds that
\begin{equation}
\left(1-2t \cdot q_{\max}\right)^{m-2t} \ge \left(1-2t \cdot q_{\max}\cdot(m-2t)\right), \label{eq:exp-snake3}
\end{equation}
if $2t\cdot q_{\max}\le 1$.
Plugging \eq{exp-snake2} and \eq{exp-snake3} into \eq{exp-snake1} we get the result as desired.
\end{proof}

\subsection{The coarse threshold case}\label{sec:coarse-snake}

We want to prove an upper bound on the non-uniform random 2-SAT threshold. 
To get to know the proof technique, we start with the much simpler case that there is a constant $\eps_1\in(0,1)$ with $p_1^2\ge\eps_1\cdot\sum_{i=1}^n p_i^2$ and that we can choose a sufficiently small $\eps_2\in(0,1)$ such that $p_2^2\le\eps_2\cdot\sum_{i=2}^n p_i^2$.
For this case, we set 
\[m^\star=\left(C\cdot p_1\cdot\left(\sum_{i=2}^n p_i^2\right)^{1/2}\right)^{-1}.\]
In order to show the desired result, we need the following lower bound on $m^\star$.
\begin{lemma}\label{lem:m-lower-bound}
Given an ensemble of probability distributions $(\vec{p}^{(n)})_{n\in\N}$ so that there is a constant $\eps_1\in(0,1)$ with $p_1^2\ge\eps_1\cdot\sum_{i=1}^n p_i^2$ and a constant $\eps_2\in(0,1)$ so that $p_2^2\le \eps_2\cdot \sum_{i=2}^n p_i^2$.
Let $m^\star=(C\cdot p_1\cdot(\sum_{i=2}^n p_i^2)^{1/2})^{-1}$.
Then, 
\[m^\star\ge\frac{1-\eps_2^{1/2}}{\eps_2^{1/4}}.\]
\end{lemma}
\begin{proof}
First, we notice
\begin{align*}
\sum_{i=2}^n p_i^2 
&\le p_2\cdot \sum_{i=2}^n p_i. \\
\intertext{Due to the requirement $p_2^2\le\eps_2\cdot \sum_{i=2}^n p_i^2$, we get}
&\le \eps_2^{1/2} \cdot \left(\sum_{i=2}^n p_i^2\right)^{1/2}\cdot \sum_{i=2}^n p_i. \\
\intertext{The monotonicity of vector norms yields $\left(\sum_{i=2}^n p_i^2\right)^{1/2}\le\sum_{i=2}^n p_i$ and thus}
&\le \eps_2^{1/2} \cdot \left(\sum_{i=2}^n p_i\right)^{2} \\
&=\eps_2^{1/2}\cdot (1-p_1)^2.
\end{align*}
It holds that
\begin{align*}
m^\star 
& = \frac{1-\sum_{i=1}^n p_i^2}{p_1\cdot \left(\sum_{i=2}^n p_i^2\right)^{1/2}}.\\
\intertext{We can now use the inequality $\sum_{i=2}^n p_i^2\le\eps_2^{1/2}\cdot (1-p_1)^2$ to get}
& \ge \frac{1-p_1^2-\sum_{i=2}^n p_i^2}{p_1\cdot \eps_2^{1/4}\cdot (1-p_1)}\\
& \ge \frac{1-p_1^2-\eps_2^{1/2}\cdot (1-p_1)^2}{p_1\cdot \eps_2^{1/4}\cdot (1-p_1)}\\
& = \frac{(1-p_1)\cdot(1+p_1-\eps_2^{1/2}\cdot (1-p_1))}{p_1\cdot \eps_2^{1/4}\cdot (1-p_1)}\\
& = \frac{1+p_1-\eps_2^{1/2}\cdot (1-p_1)}{p_1\cdot \eps_2^{1/4}} > \frac{1-\eps_2^{1/2}}{\eps_2^{1/4}}.\qedhere
\end{align*}
\end{proof}
The former lemma especially says that $m^\star$ can be arbitrarily large if $\eps_2\in(0,1)$ is sufficiently small.

We want to show that for any constant $\eps_m>0$ at $m\ge\eps_m\cdot m^\star$ there is a constant $\eps_P\in(0,1)$ such that the probability that a randomly generated instance contains a snake of size $t=2$ is at least $\eps_P>0$.
In that case, the only degree of freedom we have is choosing a constant $\eps_2$ arbitrarily small.
Together with our previous results this implies that the probability to generate an unsatisfiable instance is a constant bounded away from zero and one at $m\in\Theta(m^\star)$.
However, if we can also choose $\eps_m>0$ arbitrarily large, we can show that this result holds for any constant $\eps_P\in(0,1)$.
This implies that the probability to generate an unsatisfiable instance approaches one if $m\in\omega(m^\star)$.

In order to derive those results, we first show a lower bound on the expected number of snakes of size $t=2$.
Let us discuss what our lemma is going to state.
We assume that there is an $\eps_1\in(0,1)$ with $p_1^2\ge\eps_1\cdot\sum_{i=1}^n p_i^2$ and that we can choose $\eps_2\in(0,1)$ arbitrarily small.
This setting captures the case $p_1^2\in\Theta(\sum_{i=1}^n p_i^2)$ and $p_2^2\in o(\sum_{i=2}^n p_i^2)$.
We want to show a result for all functions $m\in\Omega(m^\star)$.
Thus, we assume $m=\eps_m \cdot m^\star$ for some $\eps_m>0$.
We can show that, given $\eps_1$ and $\eps_m$, we can choose $\eps_2$ small enough such that $\Ex{X_2}>\eps_E$ for any constant $\eps_E<\eps_m^4$. 
This will imply $\Ex{X_2}\in\Omega(1)$ later.
If $\eps_1$ and some $\eps_E>0$ are given and we can choose $\eps_m$ and $\eps_2$, then we can show that we can choose those values such that $\Ex{X_2}>\eps_E$ for any $\eps_E$ given.
This will imply $\Ex{X_2}\in\omega(1)$ later.
However, we will show that these results only hold for $m=\eps_m\cdot m^\star$.
These are also the values of $m$ for which we will show bounds on the probability to generate unsatisfiable instances.
For higher values of $m$, for example for $m\in\omega(m^\star)$, these bounds still hold due to the monotonicity of unsatisfiability in non-uniform random $k$-SAT (\cf \lemref{monotone-drawing}).

\begin{lemma} \label{lem:exp-snakes-coarse}
Given an ensemble of probability distributions $(\vec{p}^{(n)})_{n\in\N}$ so that there is a constant $\eps_1\in(0,1)$ with $p_1^2\ge\eps_1\cdot\sum_{i=1}^n p_i^2$ and let $m^\star=(C\cdot p_1\cdot(\sum_{i=2}^n p_i^2)^{1/2})^{-1}$.
The following statements hold:
\begin{enumerate}
\item Given a constant $\eps_m>0$ with $m=\eps_m\cdot m^\star$ and a constant $\eps_E\in(0,1)$, then we can choose a constant $\eps_2\in(0,1)$ such that $p_2^2\le \eps_2\cdot \sum_{i=2}^n p_i^2$ implies
\begin{equation*}
\Ex{X_2} \ge \left(1-\eps_E\right)\cdot\frac12\cdot m^{4}\cdot C^4\cdot p_1^4\cdot\left(\sum_{i=2}^n p_i^2\right)^{2}.
\end{equation*}
\item Given a constant $\eps_m>0$ with $m=\eps_m\cdot m^\star$ and a constant $\eps_E\in(0,\frac12\cdot\eps_m^4)$, then we can choose a constant $\eps_2\in(0,1)$ such that $p_2^2\le \eps_2\cdot \sum_{i=2}^n p_i^2$ implies $\Ex{X_2}\ge\eps_E$.
\item Given a constant $\eps_E>0$, then we can choose a constant $\eps_m>0$ with $m=\eps_m\cdot m^\star$ sufficiently large and a constant $\eps_2\in(0,1)$ sufficiently small such that $p_2^2\le \eps_2\cdot \sum_{i=2}^n p_i^2$ implies $\Ex{X_2}\ge\eps_E$.
\end{enumerate}
\end{lemma}
\begin{proof}
For the first statement, note that
\begin{equation}\label{eq:qmax}
(4\cdot q_{\max})^{-1}=\frac14 \cdot \frac{1}{C\cdot p_1\cdot p_2}=\frac14\cdot \frac{\left(\sum_{i=2}^n p_i^2\right)^{1/2}}{p_2} m^\star\ge\frac{1}{4\cdot \eps_2^{1/2}}\cdot m^\star>\frac{1-\eps_2^{1/2}}{4\cdot \eps_2^{3/4}}
\end{equation}
due to \lemref{m-lower-bound}
This means, we can choose $\eps_2$ small enough, such that $t=2\le\left(2\cdot q_{\max}\right)^{-1}$.
This allows us to use \lemref{exp-snakes} with $t=2$, which yields
\[\Ex{X_2} \ge  \frac12\cdot(m-4)^{4}\cdot C^{4}\cdot \left(1-4\cdot q_{\max}\cdot m\right)\cdot\left(\sum_{i=1}^n p_i^4\right)\cdot\left(\sum_{i=2}^n p_i^2-p_{2}^2\right)^{2}.\]
We now get
\begin{equation*}
\left(\sum_{i=2}^{n}p_i^2-p_2^2\right)^{2} \ge \left(\sum_{i=2}^{n}p_i^2\right)^{2}\cdot\left(1-\frac{p_{2}^2}{\sum_{i=2}^n p_i^2}\right)^{2} \ge (1-\eps_2)^2\cdot\left(\sum_{i=2}^{n}p_i^2\right)^{2} ,
\end{equation*}
where we used $p_2^2\le\eps_2\cdot\sum_{i=2}^n p_i^2$.
Equivalently,
\begin{equation*}
\left(m-4\right)^{4} \ge m^{4}\cdot\left(1-\frac{4}{m}\right)^4\ge m^{4}\cdot \left(1-\frac{4\cdot\eps_2^{1/4}}{\eps_m\cdot (1-\eps_2^{1/2})}\right)^4,
\end{equation*}
which holds since $m=\eps_m \cdot m^\star \ge \eps_m\cdot \frac{1-\sqrt{\eps_2}}{\eps_2^{1/4}}$ due to \lemref{m-lower-bound}.
Since $m=\eps_m\cdot m^\star$ and due to \eq{qmax} we get
\[1-4\cdot q_{\max}\cdot m\ge1-\frac{4\cdot\eps_2^{1/2}\cdot m}{m^\star}=1-4\cdot\eps_2^{1/2}\cdot\eps_m.\]

Since $\left(\sum_{i=1}^n p_i^4\right)\ge p_1^4$, the expected value now simplifies to
\begin{align*}
\Ex{X_2}
&\ge \left(1-\frac{4\cdot\eps_2^{1/4}}{\eps_m\cdot (1-\eps_2^{1/2})}\right)^4\cdot(1-\eps_2)^2\cdot\left(1-4\cdot\eps_2^{1/2}\cdot\eps_m\right)\cdot\frac{m^4}{2}\cdot C^4\cdot p_1^4\cdot\left(\sum_{i=2}^n p_i^2\right)^{2}.
\end{align*}
We can see that for any choice of $\eps_m$ and $\eps_1$, the leading factor gets closer to one as $\eps_2$ gets closer to zero.
Thus, for any $\eps_m>0$, $\eps_1\in(0,1)$, and $\eps_E\in(0,1)$ we can choose a sufficiently small $\eps_2$ to guarantee
\[\Ex{X_2}\ge(1-\eps_E)\cdot\frac12\cdot m^{4}\cdot C^4\cdot p_1^4\cdot\left(\sum_{i=2}^n p_i^2\right)^{2}.\]
This establishes the first statement.

For the second statement, suppose we are given an $\eps_m>0$ with $m=\eps_m\cdot m^\star$.
Then, 
\begin{align*}
\Ex{X_2} 
&\ge (1-\eps)\cdot\frac12\cdot m^{4}\cdot C^4\cdot p_1^4\cdot\left(\sum_{i=2}^n p_i^2\right)^{2}\\
&=(1-\eps)\cdot\frac12\cdot (m/m^\star)^4=(1-\eps)\cdot\frac12\cdot \eps_m^4
\end{align*}
for some constant $\eps$ that decreases with decreasing $\eps_2$.
The smaller we choose $\eps_2$, the closer this function gets to $\frac12\cdot\eps_m^4$.
Thus, for any $\eps_E\in(0,\frac12\cdot\eps_m^4)$ we can achieve $\Ex{X_2}\ge \eps_E$.
This establishes the second statement.

For the third statement suppose we are given an $\eps_E>0$ and we can choose $\eps_m>0$ with $m=\eps_m\cdot m^\star$.
Again, we get
\begin{align*}
\Ex{X_2} &\ge (1-\eps)\cdot \frac12\cdot\eps_m^4
\end{align*}
for some constant $\eps$ that decreases for fixed $\eps_m$ and decreasing $\eps_2$.
First, we choose $\eps_m$ such that $\frac12\cdot\eps_m^4>\eps_E$.
Now we know that we can make $\eps_2$ small enough so that the expected value is at least $\eps_E$.
\end{proof}

We are now ready to prove that random formulas are unsatisfiable with some positive constant probability at $m\in\Theta((C\cdot p_1\cdot(\sum_{i=2}^n p_i^2)^{1/2})^{-1})$.
More precisely, we will show that, given $\eps_1$ and $\eps_m$, we can choose an $\eps_2$ sufficiently small such that there is a constant $\eps_P\in(0,1)$ which bounds the probability to generate unsatisfiable instances from below.
Moreover, this value $\eps_P$ depends only on $\eps_1$ and $\eps_m$ and not on $n$.
This means, this lower bound does not approach zero or one as $n$ increases.

In the proof we consider $\Prob{X_A=1 \wedge X_B=1}$, the probability that both snake $A$ and snake $B$ appear exactly once in a random formula.
We distinguish several cases depending on how many clauses $F_A$ and $F_B$ have in common.
Then, we analyze the probability of $F_A \cup F_B$ appearing exactly once.
In order to do so, we assume that some snake $A$ and the shared clauses of $A$ and $B$ have already been chosen.
Then, we construct $B$, incorporating the shared clauses from $A$.

\begin{lemma}\label{lem:coarse-threshold1}
Given an ensemble of probability distributions $(\vec{p}^{(n)})_{n\in\N}$ and a constant $\eps_1\in(0,1)$ with $p_1^2\ge\eps_1\cdot(\sum_{i=1}^n p_i^2)$.
Let $m^\star=(C\cdot p_1\cdot (\sum_{i=2}^n p_i^2)^{1/2})^{-1}$.
Then, for any constant $\eps_m>0$ with $m=\eps_m\cdot m^\star$ and any \[\eps_P<\frac{\eps_m^4}{\eps_m^4+3\cdot\eps_m^2\left(1+\frac{1}{\eps_1}+\frac{1}{\eps_1^2}\right)+8}\] 
we can choose a constant $\eps_2\in(0,1)$ with $p_2^2\le\eps_2\cdot\sum_{i=2}^n p_i^2$ such that the probability to generate an unsatisfiable formula $\Phi\sim\mathcal{D}^N(n,2,\left(\vec{p}^{(n)}\right)_{n\in\N},m)$ is at least $\eps_P$.
\end{lemma}
\begin{proof}
First, we want to show that given $\eps_m>0$ with $m=\eps_m\cdot m^\star$ and $\eps_1\in(0,1)$ with $p_1^2\ge\eps_1\cdot(\sum_{i=1}^n p_i^2)$, there is an $\eps_P\in(0,1)$ such that 
\[\Prob{X_2>0}\ge\frac{\Ex{X_2}^2}{\Ex{X_2^2}}\ge\eps_P.\]
Since \lemref{exp-snakes-coarse} gives us a lower bound on $\Ex{X_2}$, we only need to consider $\Ex{X_2^2}$ now.
We use the same approach as \citet{chvatalreed92} and split the expected value into two parts as follows
\begin{align*}
\Ex{X_2^2} &=\sum_{A}{\sum_{B}\Prob{X_A=1 \wedge X_B=1}}\\
&=\sum_{A}{\left(\sum_{B\colon B\not{\sim}A}\Prob{X_A=1 \wedge X_B=1}+\sum_{B\colon B\sim A}\Prob{X_A=1 \wedge X_B=1}\right)},
\end{align*}
where $B\sim A$ denotes $F_A\cap F_B\neq\emptyset$.
We will show that the part for $B\nsim A$ is at most $(1+\eps_E)\cdot\Ex{X_2}^2$ for some arbitrarily small constant $\eps_E>0$ and that there is a constant $\eps_F>0$ such that the other part is at most $\eps_F\cdot\Ex{X_2}^2$.

First let us consider the part for $B\nsim A$.
It holds that
\begin{equation*}
\Prob{X_A=1 \wedge X_B=1}=\binom{m}{8} \cdot 8! \cdot \left(\prod_{c\in F_A}{\Prob{c}}\right)\cdot\left(\prod_{c\in F_B}{\Prob{c}}\right)\cdot\left(1-\sum_{c\in F_A\cup F_B}{\Prob{c}}\right)^{m-8},
\end{equation*}
while
\begin{equation}
\Prob{X_A=1}=\binom{m}{4}\cdot4!\cdot\left(\prod_{c\in F_A}{\Prob{c}}\right)\cdot\left(1-\sum_{c\in F_A}{\Prob{c}}\right)^{m-4}.
\end{equation}
Since $\binom{m}{8}\cdot8!\le \left(\binom{m}{4}\cdot4!\right)^2$ this readily implies
\begin{align*}
&\Prob{X_A=1 \wedge X_B=1} \\
& \le \Prob{X_A=1}\cdot\Prob{X_B=1}\frac{\left(1-\sum_{c\in F_A\cup F_B}{\Prob{c}}\right)^{m-8}}{\left(1-\sum_{c\in F_A}{\Prob{c}}\right)^{m-4}\left(1-\sum_{c\in F_B}{\Prob{c}}\right)^{m-4}}\\
\intertext{and, due to $\left(1-\sum_{c\in F_A}{\Prob{c}}\right)\cdot\left(1-\sum_{c\in F_B}{\Prob{c}}\right)\ge 1-\sum_{c\in F_A\cup F_B}{\Prob{c}}$, we have}
&\Prob{X_A=1 \wedge X_B=1} \\
& \le \Prob{X_A=1}\cdot\Prob{X_B=1}\cdot\left(1-\sum_{c\in F_A}{\Prob{c}}\right)^{-4}\left(1-\sum_{c\in F_B}{\Prob{c}}\right)^{-4}.
\end{align*}
Again, we can use Bernoulli's inequality to show
\begin{align*}
\left(1-\sum_{c\in F_A}{\Prob{c}}\right)^{4}\left(1-\sum_{c\in F_B}{\Prob{c}}\right)^{4}
&\ge(1-4\cdot q_{\max})^8\\
&\ge1-32\cdot q_{\max}\\
&\ge1-\frac{32\cdot\eps_2^{3/4}}{1-\sqrt{\eps_2}},
\end{align*}
where the last inequality follows with $q_{\max}<\frac{\eps_2^{3/4}}{1-\sqrt{\eps_2}}$ due to \eq{qmax}.
For any fixed $\eps_m$ this expression can be made arbitrarily close to one if we choose a sufficiently small $\eps_2$.
This establishes 
\begin{align}
\sum_{A}{\sum_{B\colon B\nsim A}\Prob{X_A=1 \wedge X_B=1}}
&\le\frac{1}{1-32\cdot q_{\max}}\sum_{A}{\sum_{B\colon B\nsim A}\Prob{X_A=1}\cdot \Prob{X_B=1}}\notag\\
&\le\frac{1-\sqrt{\eps_2}}{1-\sqrt{\eps_2}-32\cdot \eps_2^{3/4}}\cdot\Ex{X_2}^2\notag\\
& = (1+\eps_E)\cdot \Ex{X_2}^2\label{eq:snake-count-coarse-0}
\end{align}
for a constant $\eps_E$ that we can make arbitrarily small by making $\eps_2$ small enough.

Now we turn to the case that $B\sim A$.
We want to show that there is a constant $\eps_F>0$ such that this second sum is at most $\eps_F\cdot\Ex{X_2}^2$.
Let $l=|F_{A}\cap F_B|$.
The first and simplest case is $F_{A}=F_{B}$. 
This obviously happens if $A=B$, but also for three other snakes.
So it holds that
\begin{align}
\sum_{A}\sum_{\substack{B\colon\\ |F_{A}\cap F_B|=4}}\Prob{X_A=1 \wedge X_B=1}
&=4\cdot\Ex{X_2}=\frac{4}{\Ex{X_2}}\cdot\Ex{X_2}^2\notag\\
\intertext{and since we can achieve $\Ex{X_2}\ge \eps$ for any constant $\eps\in(0,\tfrac12\cdot\eps_m^4)$ due to \corref{threshold-coarse2} by making $\eps_2$ sufficiently small, we get }
\sum_{A}\sum_{\substack{B\colon\\ |F_{A}\cap F_B|=4}}\Prob{X_A=1 \wedge X_B=1}&\le \frac{4}{\eps}\cdot\Ex{X_2}^2=\eps_F\cdot \Ex{X_2}^2\label{eq:snake-count-coarse-1}
\end{align}
for any constant $\eps_F>8/\eps_m^4$.
This captures the case $l=4$.

For $1\le l\le3$ it holds that
\begin{align}
&\sum_{A}\sum_{\substack{B\colon\\ |F_{A}\cap F_B|=l}}\Prob{X_A=1 \wedge X_B=1}\notag\\
&\le \binom{m}{8-l}\cdot (8-l)!\cdot\left(1-\sum_{c\in{F_A \cup F_B}}\Prob{c}\right)^{m-8+l}\cdot 2^3\cdot 2! \cdot\left(\frac{C}{2}\right)^4\cdot \notag\\
&\cdot \left(\sum_{i=1}^n\left( p_i^4 \cdot \sum_{\substack{S\subseteq ([n]\setminus\left\{i\right\}):\\|S|=2}}\prod_{s\in S} p_s^2\right)\right)\cdot\sum_{\substack{B\colon\\ |F_{A}\cap F_B|=l}} \prod_{c\in F_B\setminus F_A} \Prob{c} \label{eq:square-snakes-coarse1}
\end{align}
where we accounted for the $8-l$ possible positions of clauses from $F_A\cup F_B$ in $\Phi$, for the $2^3\cdot 2!$ possibilities to create a snake $A$ from chosen variables if the central variable is determined already, and for the ways to choose those variables.
Now we want to bound the term 
\[\sum_{i=1}^n\left( p_i^4 \cdot \sum_{\substack{S\subseteq ([n]\setminus\left\{i\right\}):\\|S|=2}}\prod_{s\in S} p_s^2\right).\]
In order to do so we distinguish between the cases that $p_1$ appears in the snake as the central variable, a non-central variable or not at all to show the following 
\begin{align*}
& \sum_{i=1}^n \left(p_i^4 \cdot \sum_{\substack{S\subseteq ([n]\setminus\left\{i\right\}):\\|S|=2}}\prod_{s\in S} p_s^2\right)\\
& \le p_1^4\cdot \left(\sum_{i=2}^n p_i^2\right)^2 + \left(\sum_{i=2}^n p_i^4\right)\cdot p_1^2 \cdot \left(\sum_{i=2}^n p_i^2\right) + \left(\sum_{i=2}^n p_i^4\right)\cdot \left(\sum_{i=2}^n p_i^2\right)^2.\\
\intertext{Again, the monotonicity of vector norms implies $\sum_{i=2}^n p_i^4 \le \left(\sum_{i=2}^n p_i^2\right)^2$ and thus}
& \le p_1^4\cdot \left(\sum_{i=2}^n p_i^2\right)^2 + p_1^2 \cdot \left(\sum_{i=2}^n p_i^2\right)^3 + \left(\sum_{i=2}^n p_i^2\right)^4\\
& \le \left(1+\frac{1}{\eps_1}+\frac{1}{\eps_1^2}\right)\cdot p_1^4\cdot \left(\sum_{i=2}^n p_i^2\right)^2,
\end{align*}
where we used the prerequisite $\sum_{i=2}^n p_i^2\le\sum_{i=1}^n p_i^2\le p_1^2/\eps_1$.
If we plug this into \eq{square-snakes-coarse1}, we get
\begin{multline}
\sum_{A}\sum_{\substack{B\colon\\ |F_{A}\cap F_B|=l}}\Prob{X_A=1 \wedge X_B=1}\\
\le \left(1+\frac{1}{\eps_1}+\frac{1}{\eps_1^2}\right)\cdot m^{8-l}\cdot C^4\cdot p_1^4\cdot \left(\sum_{i=2}^n p_i^2\right)^2\cdot\sum_{\substack{B\colon\\ |F_{A}\cap F_B|=l}} \prod_{c\in F_B\setminus F_A} \Prob{c}. \label{eq:square-snakes-coarse2}
\end{multline}
Now we consider the cases $l\in \left\{1,2,3\right\}$.
We assume that $A$ is chosen already and that we want to construct all snakes $B$ that contain exactly $l$ clauses from $A$.
However, the bounds we derive will be independent of the actual choice of $A$.
Thus, we can simply plug them into \eq{square-snakes-coarse2}.
Remember that a snake of size $2$ contains the four clauses
\[\left(w_2,w_1\right),\left(\n{w_1},w_2\right),\left(\n{w_2},w_3\right),\left(\n{w_3},\n{w_2}\right)\]
for literals $w_1$, $w_2$, and $w_3$ of distinct Boolean variables.

For $l=1$ we know one shared clause which has to contain $B$'s central variable $x$ and one of $B$'s non-central variables $y$.
Here, we overestimate that for $B$ we choose any two variables from $A$, one as $B$'s central variable and one as a non-central variable.
The other non-central variable $z$ only has to be different from $x$ and $y$.
If we have a look at the clauses a snake of size $2$ contains, we see that the central variable appears 4 times, while the other two variables appear two times each.
However, the central and one of the non-central variables already appear in a shared clause from $A$ and no clause of a snake is supposed to appear more than once in the formula.
Thus, the central variable $x$ appears an additional 3 times, $y$ appears an additional one time, and $z$ appears an additional two times.
Formally, it holds that
\begin{equation*}
\sum_{\substack{B\colon\\ |F_{A}\cap F_B|=1}} \prod_{c\in F_B\setminus F_A} \Prob{c}\le \left(\frac{C}{2}\right)^{3}\cdot \sum_{x\in (S\cup\left\{i\right\})} \left(p_x^3 \cdot \sum_{y\in{(S\cup\left\{i\right\})\setminus\left\{x\right\}}}\left( p_y \cdot \sum_{z\in [n]\setminus \left\{x,y\right\}} p_z^2\right)\right),
\end{equation*}
where $i$ is the central variable and $S$ are the other variables of $A$.
Again, we can do a case distinction depending on the appearances of $p_1$.
We can see that $p_1$ can appear as one of the three variables only.
Also, the variables $S\cup\left\{i\right\}$ of $A$ are predetermined and $|S\cup\left\{i\right\}|=3$.
That means, if $1$ is not part of $S\cup\left\{i\right\}$ or not chosen from it, the set contains at most 3 other indices, whose associated variables have probabilities of at most $p_2$ each.
We now distinguish 4 cases: $x=1$, $y=1$, $z=1$, and $\left\{x,y,z\right\}\cap\left\{1\right\}=\emptyset$.
The terms of the following expression represent those cases.
It holds that
\begin{align*}
&\sum_{x\in (S\cup\left\{i\right\})} p_x^3 \cdot \left(\sum_{y\in{(S\cup\left\{i\right\})\setminus\left\{x\right\}}} p_y \cdot \left(\sum_{z\in [n]\setminus \left\{x,y\right\}} p_z^2\right)\right)\\
&\le p_1^3\cdot 2p_2\cdot\sum_{i=2}^n p_i^2 + 3p_2^3\cdot p_1 \cdot \sum_{i=2}^n p_i^2 + 3p_2^3\cdot 2p_2\cdot p_1^2+ 3p_2^3\cdot 2p_2\cdot\sum_{i=2}^n p_i^2\\
&\le 17\cdot p_1^3\cdot p_2 \cdot \sum_{i=2}^n p_i^2,
\end{align*}
where we used $p_2\le p_1$ and $p_2^2\le\sum_{i=2}^n p_i^2$.
Together with \eq{square-snakes-coarse2}, it now holds that
\begin{align*}
&\sum_{A}\sum_{\substack{B\colon\\ |F_{A}\cap F_B|=1}}\Prob{X_A=1 \wedge X_B=1}\\
&\le \frac{17}{8}\cdot\left(1+\frac{1}{\eps_1}+\frac{1}{\eps_1^2}\right)\cdot m^{7}\cdot C^7\cdot p_1^7\cdot p_2\cdot\left(\sum_{i=2}^n p_i^2\right)^3\\
&=\sqrt{\eps_2}\cdot\frac{17}{8}\cdot\left(1+\frac{1}{\eps_1}+\frac{1}{\eps_1^2}\right)\cdot m^{7}\cdot C^7\cdot p_1^7\cdot \left(\sum_{i=2}^n p_i^2\right)^{7/2},
\end{align*}
since $p_2^2\le \eps_2\cdot\sum_{i=2}^n p_i^2$.
In the first statement of \lemref{exp-snakes-coarse} we show that for any given $\eps_1$, $\eps_m$, and $\eps\in(0,1)$, we can choose $\eps_2$ small enough such that 
\[\Ex{X_2} \ge (1-\eps)\cdot\frac12\cdot m^{4}\cdot C^4\cdot p_1^4\cdot \left(\sum_{i=2}^n p_i^2\right)^{2}.\]
This implies
\begin{align}
&\sum_{A}\sum_{\substack{B\colon\\ |F_{A}\cap F_B|=1}}\Prob{X_A=1 \wedge X_B=1}\notag\\
&\le \frac{\sqrt{\eps_2}}{(1-\eps)^2}\cdot\frac{17}{2}\cdot\left(1+\frac{1}{\eps_1}+\frac{1}{\eps_1^2}\right)\cdot\frac{\Ex{X_2}^2}{m\cdot C\cdot p_1\cdot \left(\sum_{i=2}^n p_i^2\right)^{1/2}}.\notag\\
\intertext{Since $m\cdot C\cdot p_1\cdot \left(\sum_{i=2}^n p_i^2\right)^{1/2}=m/m^\star=\eps_m$, we get}
&\sum_{A}\sum_{\substack{B\colon\\ |F_{A}\cap F_B|=1}}\Prob{X_A=1 \wedge X_B=1}\notag\\
&\le \frac{\sqrt{\eps_2}}{(1-\eps)^2}\cdot\frac{17}{2}\cdot\left(1+\frac{1}{\eps_1}+\frac{1}{\eps_1^2}\right)\cdot\frac{\Ex{X_2}^2}{\eps_m}\notag\\
&=\eps_F\cdot \Ex{X_2}^2\label{eq:snake-count-coarse-2}
\end{align}
for any $\eps_F>0$ if we choose $\eps_2$ small enough ($\eps$ decreases as $\eps_2$ does).

Now we consider $l=2$.
Again, it is helpful to visualize the clauses a snake of size 2 consists of:
\[\left(w_2,w_1\right),\left(\n{w_1},w_2\right),\left(\n{w_2},w_3\right),\left(\n{w_3},\n{w_2}\right).\]
With two shared clauses, two cases can happen.
Either all three variables of $A$ appear in the two shared clauses or only two do.
In the first case, one variable of $A$ appears in $B$ twice, while the other two appear only once.
However, this information is already enough to completely determine how all other clauses of $B$ have to look.
It implies that the variable that appears twice is the central variable both in $A$ and in $B$, since only the central variable appears in clauses with both other variables.
Moreover, the two shared clauses already imply $A=B$ and thus $l=4$.
This means, this case cannot happen!
Thus, we only have to consider the second case, in which two variables from $A$ each appear twice in the shared clauses.
Again, the shared clauses already determine that the central variable from $A$ also is the central variable in $B$, since only the literals of the central variable appear with the same sign in both clauses.
In $B$ this central variable has to appear an additional two times and a new variable $x\in\left([n]\setminus \left(S\cup\left\{i\right\}\right)\right)$ has to appear two times as well.
The other variable of $B$ does not appear again, since it already appeared two times in shared clauses.
More formally,
\begin{align*}
\sum_{\substack{B\colon\\ |F_{A}\cap F_B|=2}} \prod_{c\in F_B\setminus F_A} \Prob{c} 
& = \left(\frac{C}{2}\right)^2\cdot p_i^2\cdot\sum_{x\in [n]\setminus \left(S\cup\left\{i\right\}\right)}^n p_x^2.\\
\intertext{By considering the possible appearances of $p_1$ again, we get}
\sum_{\substack{B\colon\\ |F_{A}\cap F_B|=2}} \prod_{c\in F_B\setminus F_A} \Prob{c} 
&\le \left(\frac{C}{2}\right)^2\cdot\left( p_1^2\cdot\sum_{i=2}^n p_i^2+p_2^2\cdot p_1^2+p_2^2\cdot\sum_{i=2}^n p_i^2\right)\\
&\le 3\cdot\left(\frac{C}{2}\right)^2\cdot p_1^2\cdot\sum_{i=2}^n p_i^2.
\end{align*}
Again with \eq{square-snakes-coarse2}, it holds that
\begin{align*}
\sum_{A}\sum_{\substack{B\colon\\ |F_{A}\cap F_B|=2}}\Prob{X_A=1 \wedge X_B=1}
&\le \frac{3}{4}\cdot \left(1+\frac{1}{\eps_1}+\frac{1}{\eps_1^2}\right)\cdot m^{6}\cdot C^6\cdot p_1^6\cdot \left(\sum_{i=2}^n p_i^2\right)^3.
\end{align*}
Since we can choose $\eps_2$ small enough such that $\Ex{X_2} \ge (1-\eps)\cdot\frac12\cdot m^{4}\cdot C^4\cdot p_1^4\cdot \left(\sum_{i=2}^n p_i^2\right)^{2}$ for any $\eps\in(0,1)$, we get
\begin{align}
\sum_{A}\sum_{\substack{B\colon\\ |F_{A}\cap F_B|=2}}\Prob{X_A=1 \wedge X_B=1}
&\le \frac{3}{(1-\eps)^2}\cdot \left(1+\frac{1}{\eps_1}+\frac{1}{\eps_1^2}\right)\cdot \frac{\Ex{X_2}^2}{m^{2}\cdot C^2\cdot p_1^2\cdot \left(\sum_{i=2}^n p_i^2\right)}\notag\\
&=\frac{3}{(1-\eps)^2}\cdot \left(1+\frac{1}{\eps_1}+\frac{1}{\eps_1^2}\right)\cdot \frac{\Ex{X_2}^2}{\eps_m^2}\\
&=\eps_F\cdot \Ex{X_2}^2 \label{eq:snake-count-coarse-3}
\end{align}
for some constant $\eps_F>\tfrac{3}{\eps_m^2}\cdot \left(1+\tfrac{1}{\eps_1}+\tfrac{1}{\eps_1^2}\right)$.

The last case is $l=3$.
This case can not happen, since $3$ shared clauses already fully determine the last clause, which also has to align with one of $A$, \ie we do not have any degree of freedom to make $F_A\neq F_B$.

Putting \eq{snake-count-coarse-1}, \eq{snake-count-coarse-2}, and \eq{snake-count-coarse-3} together, establishes that we can choose $\eps_2$ sufficiently small to make
\[\sum_{A}{\sum_{B\colon B\sim A}\Prob{X_A=1 \wedge X_B=1}}\le\eps_F\cdot\Ex{X_2}^2\]
for any constant $\eps_F>\frac{3}{\eps_m^2}\cdot \left(1+\frac{1}{\eps_1}+\frac{1}{\eps_1^2}\right)+\frac{8}{\eps_m^4}$.
Together with \eq{snake-count-coarse-0}, this gives us
\begin{align*}
\Ex{X_2^2}
&=\sum_{A}{\left(\sum_{B\colon B\not{\sim}A}\Prob{X_A=1 \wedge X_B=1}+\sum_{B\colon B\sim A}\Prob{X_A=1 \wedge X_B=1}\right)}\\
&\le(1+\eps_E+\eps_F)\cdot\Ex{X_2}^2
\end{align*}
for any constant $\eps_E>0$ and any $\eps_F>\frac{3}{\eps_m^2}\cdot \left(1+\frac{1}{\eps_1}+\frac{1}{\eps_1^2}\right)+\frac{8}{\eps_m^4}$ and implies 
\[\Prob{X_2>0}\ge\frac{\Ex{X_2}^2}{\Ex{X_2^2}}\ge\frac{1}{1+\eps_E+\eps_F}=\eps_P\]
for any
\[\eps_P<\frac{\eps_m^4}{\eps_m^4+3\cdot\eps_m^2\cdot\left(1+\frac{1}{\eps_1}+\frac{1}{\eps_1^2}\right)+8}.\]
\end{proof}

The following is a corollary of the former lemma and complements it.
It shows that for any $\eps_1\in(0,1)$ and any $\eps_P\in(0,1)$ non-uniform random 2-SAT formulas are unsatisfiable with probability at least $\eps_P$ if we can choose $\eps_m$ with $m=\eps_m\cdot m^\star$ sufficiently large and $\eps_2$ with $p_2^2\le\eps_2\cdot\sum_{i=2}^n p_i^2$ sufficiently small.
This captures the case $m\in\omega(m^\star)$.

\begin{corollary}\label{cor:threshold-coarse2}
Given an ensemble of probability distributions $(\vec{p}^{(n)})_{n\in\N}$ and a constant $\eps_1\in(0,1)$ with $p_1^2\ge\eps_1\cdot(\sum_{i=1}^n p_i^2)$.
Let $m^\star=(C\cdot p_1\cdot (\sum_{i=2}^n p_i^2)^{1/2})^{-1}$.
For any constant $\eps_P\in(0,1)$ we can choose a constant $\eps_m>0$ with $m=\eps_m\cdot m^\star$ and a constant $\eps_2\in(0,1)$ with $p_2^2\le\eps_2\cdot\sum_{i=2}^n p_i^2$ such that the probability to generate an unsatisfiable formula $\Phi\sim\mathcal{D}^N(n,2,(\vec{p}^{(n)})_{n\in\N},m)$ is at least $\eps_P$.
\end{corollary}
\begin{proof}
The corollary is a simple application of the former lemma.
Suppose we are given $\eps_1$ and $\eps_P$.
We can now choose an $\eps_m$ large enough such that
\[\frac{\eps_m^4}{\eps_m^4+3\cdot\eps_m^2\cdot\left(1+\frac{1}{\eps_1}+\frac{1}{\eps_1^2}\right)+8}>\eps_P\]
for the given $\eps_P$.
Due to \lemref{coarse-threshold1} we can then choose $\eps_2$ small enough to generate unsatisfiable instances with probability at least $\eps_P$.
\end{proof}

The former lemma and corollary together with \lemref{bicycle-coarse} establish that in the case of $p_1^2\in \Theta(\sum_{i=1}^n p_i^2)$ and $p_2^2\in o(\sum_{i=2}^n p_i^2)$ the asymptotic threshold is at $m\in\Theta((C\cdot p_1\left(\sum_{i=2}^n p_i^2)^{1/2})^{-1}\right)$ and that it is coarse. 
We will show this formally in \secref{main-theorem}.

\subsection{The sharp threshold case}\label{sec:sharp-snake}

In the last section we analyzed the case $p_1^2\in\Theta(\sum_{i=1}^n p_i^2)$ and $p_2^2\in o(\sum_{i=2}^n p_i^2)$.
Now we tend to the case $p_1^2\in o(\sum_{i=1}^n p_i^2)$.
In this section we are going to show that there is a sharp threshold at $m^\star=(\sum_{i=1}^n p_i^2)^{-1}$.
From \lemref{sharp-lb} we already know that formulas are \emph{satisfiable} with probability $1-o(1)$ if $m\le\eps_m\cdot m^\star$ for any constant $\eps_m\in(0,1)$.
It remains to show that they are \emph{unsatisfiable} with probability $1-o(1)$ if $m\ge\eps_m\cdot m^\star$ for any constant $\eps_m>1$.
This will establish a sharp threshold at $m^\star=(\sum_{i=1}^n p_i^2)^{-1}$.
More generally, we will show that, given $\eps_P\in(0,1)$ and $\eps_m>1$ so that $m=\eps_m\cdot m^\star$, we can choose $\eps_1$ with $p_1^2\le\eps_1\cdot\sum_{i=1}^n p_i^2$ sufficiently small so that the probability to generate an unsatisfiable instance is at least $\eps_P$.

We use the same technique as in the last section to prove this result, the second moment method.
As before, in order to show a sufficiently large probability for unsatisfiability, we first have to show that the expected number of snakes of a certain size $t$ can be made arbitrarily large.
The following lemma establishes this property for the suitably chosen value $t=f^{1/78}$, where $f=\left(\sum_{i=1}^n p_i^2\right)/p_1^2\ge 1/\eps_1$. 

\begin{lemma} \label{lem:exp-snakes-sharp}
Given an ensemble of probability distributions $(\vec{p}^{(n)})_{n\in\N}$ and let $t=f^{1/78}$, where $f=(\sum_{i=1}^n p_i^2)/p_1^2$.
Let $m^\star=(\sum_{i=1}^n p_i^2)^{-1}$ and let $m=\eps_m\cdot m^\star$ for some given constant $\eps_m>1$. 
Given a constant $\eps_E\in(0,1)$, we can choose a constant $\eps_1\in(0,1)$ such that $p_1^2\le \eps_1\cdot \sum_{i=1}^n p_i^2$ implies
\begin{equation*}
\Ex{X_t} \ge \left(1-\eps_E\right)\cdot\frac12\cdot m^{2t}\cdot \left(\sum_{i=1}^n p_i^4\right)\cdot\left(\sum_{i=2}^n p_i^2\right)^{2t-2}.
\end{equation*}
Furthermore, for any given $\eps_E>0$, we can choose $\eps_1\in(0,1)$ small enough to guarantee $\Ex{X_t}\ge \eps_E$.
\end{lemma}
\begin{proof}
It holds that $p_1\le1$, $\sum_{i=1}^n p_i^2\le1$, and $C=1-\sum_{i=1}^n p_i^2\le1$.
Thus, 
\[t=f^{1/78} = \frac{\left(\sum_{i=1}^n p_i^2\right)^{1/78}}{p_1^{1/39}}\le p_1^{-1/39}\le p_1^{-1}\le 1/(C\cdot p_1 \cdot p_2)=(2\cdot q_{\max})^{-1}.\]
Also,
\[t=f^{1/78}\ge\frac{1}{\eps_1^{1/78}}\ge2\]
if $\eps_1$ is sufficiently small.
Therefore, we can apply \lemref{exp-snakes} to get
\begin{multline*}
\Ex{X_t} \\\ge \frac12\cdot(m-2t)^{2t}\cdot C^{2t}\cdot (1-2t\cdot q_{\max}\cdot m)\cdot\left(\sum_{i=1}^n p_i^4\right)\cdot\left(\sum_{i=2}^n \left(p_i^2-(2t-3)\cdot p_{2}^2\right)\right)^{2t-2}.
\end{multline*}
We are going to show that this is at least 
\[\left(1-\eps_E\right)\cdot\frac12\cdot m^{2t}\cdot \left(\sum_{i=1}^n p_i^4\right)\cdot\left(\sum_{i=2}^n p_i^2\right)^{2t-2}\]
for some $\eps_E\in(0,1)$ that we can make arbitrarily small by making $\eps_1$ sufficiently small.

First, we see that
\begin{align*}
\left(\sum_{i=2}^{n}(p_i^2-(2t-3)\cdot p_2^2)\right)^{2t-2} & \ge \left(\sum_{i=1}^{n}(p_i^2-(2t-2)\cdot p_1^2)\right)^{2t-2}\\
&= \left(\sum_{i=1}^{n}p_i^2\right)^{2t-2}\cdot\left(1-\frac{(2t-2)\cdot p_1^2}{\sum_{i=1}^n p_i^2}\right)^{2t-2}.
\end{align*}
It holds that
\begin{align*}
\left(1-\frac{(2t-2)\cdot p_1^2}{\sum_{i=1}^n p_i^2}\right)^{2t-2} 
& \ge \left(1-\frac{2\cdot f^{1/78}\cdot p_1^2}{\sum_{i=1}^n p_i^2}\right)^{2t-2}\\
& = \left(1-2\cdot\left(\frac{p_1^2}{\sum_{i=1}^n p_i^2}\right)^{77/78}\right)^{2t-2}\\
\intertext{where we used our definitions of $t$ and $f$. We can see that $\left(\frac{p_1^2}{\sum_{i=1}^n p_i^2}\right)^{77/78}\le\eps_1^{77/78}$. By choosing $\eps_1$ sufficiently small, this allows us to use Bernoulli's inequality and get}
\left(1-\frac{(2t-2)\cdot p_1^2}{\sum_{i=1}^n p_i^2}\right)^{2t-2} 
& \ge \left(1-2\cdot\left(\frac{p_1^2}{\sum_{i=1}^n p_i^2}\right)^{77/78}\cdot2t\right)\\
& = \left(1-4\cdot\left(\frac{p_1^2}{\sum_{i=1}^n p_i^2}\right)^{38/39}\right)\\
& \ge \left(1-4\cdot\eps_1^{38/39}\right).
\end{align*}
We can make this factor arbitrarily close to one if we choose $\eps_1$ sufficiently small.

Equivalently,
\begin{align*}
\left(m-2t\right)^{2t} & = m^{2t}\cdot\left(1-\frac{2t}{m}\right)^{2t}\\
& = m^{2t}\cdot\left(1-2\cdot\frac{f^{1/78}\cdot \left(\sum_{i=1}^n p_i^2\right)}{\eps_m}\right)^{2t},
\end{align*}
where we used $m=\eps_m/\sum_{i=1}^n p_i^2$ and $t=f^{1/78}$.
It holds that $f\cdot p_1^2=\sum_{i=1}^n p_i^2\le p_1$, which implies $f\le p_1^{-1}$ and $t\le p_1^{-1/78}$.
Furthermore, we still know that $p_1\le\eps_1^{1/2}$, since $p_1^2\le\eps_1\cdot \sum_{i=1}^n p_i^2$ and $\sum_{i=1}^n p_i^2\le1$.
This implies
\begin{align*}
\left(1-2\cdot\frac{f^{1/78}\cdot \left(\sum_{i=1}^n p_i^2\right)}{\eps_m}\right)^{2t}
& \ge\left(1-2\cdot\frac{p_1^{77/78}}{\eps_m}\right)^{2t}\ge\left(1-2\cdot\frac{\eps_1^{77/156}}{\eps_m}\right)^{2t}.
\end{align*}
Again, we can make $\eps_1$ small enough to use Bernoulli's inequality and get
\[\left(1-2\cdot\frac{p_1^{77/78}}{\eps_m}\right)^{2t}\ge\left(1-2\cdot\frac{2t\cdot p_1^{77/78}}{\eps_m}\right)\ge\left(1-4\cdot\frac{p_1^{38/39}}{\eps_m}\right)\ge\left(1-4\cdot\frac{\eps_1^{19/39}}{\eps_m}\right).\]
Thus, for any fixed $\eps_m$ we can make this factor arbitrarily close to one by making $\eps_1$ sufficiently small.

Since we know that $C=(1-\sum_{i=1}^n p_i^2)^{-1}$ and $\sum_{i=1}^n p_i^2\le p_1\le\sqrt{\eps_1}$, this implies $1\le C\le (1-\sqrt{\eps_1})^{-1}$.
We also know that $t=f^{1/78}\ge\eps_1^{-1/78}$, which implies
\begin{align*}
1-2t\cdot q_{\max}\cdot m 
& = 1-f^{1/78}\cdot C\cdot p_1\cdot p_2 \cdot \frac{\eps_m}{\sum_{i=1}^n p_i^2}\\
& \ge 1-f^{1/78}\cdot\frac{\eps_m}{1-\sqrt{\eps_1}}\cdot \frac{p_1^2}{\sum_{i=1}^n p_i^2}\\
& = 1-f^{-77/78}\cdot\frac{\eps_m}{1-\sqrt{\eps_1}}\\
& \ge 1-\frac{\eps_m\cdot\eps_1^{77/78}}{1-\eps_1^{1/2}}.
\end{align*}
Here, we also use $p_2\le p_1$ and the definition $f=(\sum_{i=1}^n p_i^2)/p_1^2$.
Thus, for any given $\eps_m$ we can make this expression arbitrarily close to one by choosing $\eps_1$ sufficiently small.
 
For any given $\eps_m$ we can now choose an $\eps_E\in(0,1)$ and get 
\[\Ex{X_t}\ge\left(1-\eps_E\right)\cdot\frac12\cdot m^{2t}\cdot \left(\sum_{i=1}^n p_i^4\right)\cdot\left(\sum_{i=2}^n p_i^2\right)^{2t-2}\]
by making $\eps_1$ sufficiently small.
We want to show that for every $\eps_m$ we can make the expected value arbitrarily large by choosing $\eps_1$ small enough.
It holds that
\[m^2\cdot \left(\sum_{i=1}^n p_i^4\right)\ge m^2\cdot p_1^4=\frac{\eps_m^2\cdot p_1^4}{\left(\sum_{i=1}^n p_i^2\right)^2}=\frac{\eps_m^2}{f^2},\]
where we used $m=\eps_m\cdot m^\star=\eps_m\cdot\left(\sum_{i=1}^n p_i^2\right)^{-1}$.
With the same fact it holds that
\[\left(m\cdot\sum_{i=1}^n p_i^2\right)^{2t-2}=\eps_m^{2t-2}.\]
Since we know that $t=f^{1/78}$, it holds that
\[\Ex{X_t} \ge (1-\eps_E)\cdot\frac12\cdot \frac{\eps_m^{2\cdot f^{1/78}}}{f^2}.\]
We can now make this expression as large as we need it if $f$ is sufficiently large, because we assumed $\eps_m>1$.
Since $f\ge1/\eps_1$ this is the case if $\eps_1$ is sufficiently small.
Thus, for any given $\eps_E>0$ we can choose $\eps_1$ sufficiently small to guarantee $\Ex{X_t}\ge \eps_E$.
\end{proof}

We now turn to the application of the second moment method.
Again, we want to show that $\Prob{X_A=1 \wedge X_B=1}$ for snakes $A$ and $B$ with shared clauses ($F_A\cap F_B\neq\emptyset$) is relatively small compared to $\Ex{X_t}^2$. 
To this end, we have to consider different possibilities for the shared clauses to influence $\Prob{X_A=1 \wedge X_B=1}$.
In the proofs of the former case this was rather easy, since we only considered the smallest possible snakes of size $t=2$.
Now the distinction becomes a bit more difficult.
We will distinguish several cases: If the number of shared clauses is at least $t$ then $\Prob{X_A=1 \wedge X_B=1}$ is by roughly a factor of $\eps_m^t$ smaller than $\Ex{X_t}^2$.
If the shared clauses form a variable-variable-incidence graph with at least two connected components, then there are enough variable appearances pre-defined for $B$ to make $\Prob{X_A=1 \wedge X_B=1}$ sufficiently small.
The last case is that the shared clauses form only one connected component, which is a lot smaller than $t-1$.
In that case we have to carefully consider what happens to the central variable of $B$, since this variable appears most times in $B$ and the many appearances take degrees of freedom away from other variables, therefore making $\Prob{X_A=1 \wedge X_B=1}$ small.
Here, we only show the result for some $\eps_m>1$ so that $m=\eps_m\cdot\left(\sum_{i=1}^n p_i^2\right)^{-1}$.
For $m\in \omega((\sum_{i=1}^n p_i^2)^{-1})$ it follows by the monotonicity of unsatisfiability as we will see later.

\begin{lemma} \label{lem:threshold-sharp}
Given an ensemble of probability distributions $(\vec{p}^{(n)})_{n\in\N}$.
Let $m^\star=(\sum_{i=1}^n p_i^2)^{-1}$ and let $m=\eps_m\cdot m^\star$ for some given constant $\eps_m>1$. 
Given an $\eps_P\in(0,1)$, we can choose a constant $\eps_1\in(0,1)$ sucht that $p_1^2\le \eps_1\cdot \sum_{i=1}^n p_i^2$ implies
\[\Pro{\Phi\sim \mathcal{D}^N(n,2,(\vec{p}^{(n)})_{n\in\N},m)}{\Phi\text{ unsatisfiable}}\ge \eps_P.\]
\end{lemma}
\begin{proof}
Again, we utilize the second moment method.
We want to show that for any given $\eps_P\in(0,1)$ we can choose $\eps_1\in(0,1)$ sufficiently small so that some $F_A$ for a snake $A$ of size $t$ appears in $\Phi$ with probability at least $\eps_P$. 
This especially implies that we can make this probability arbitrarily close to one.
This will hold for $t=f^{1/78}$, where $f=(\sum_{i=1}^n p_i^2)/p_1^2$.
We will later see why we chose $t$ this way.
Again, we define $X_A$ as an indicator variable for the event that the formula $F_A$ associated with snake $A$ appears exactly once in $\Phi$
and $X_t=\sum_{\text{snake $A$ of size $t$}}{X_A}$.
As in the proof of \corref{threshold-coarse2} we want to show that for any $\eps_E>0$ we can choose $\eps_1$ small enough so that $\Ex{X_t^2}\le(1+\eps_E)\cdot\Ex{X_t}^2$.
Then, the second moment method gives us 
\[\Prob{X_t>0}\ge\frac{\Ex{X_t}^2}{\Ex{X_t^2}}\ge\frac{1}{1+\eps_E}.\]
Thus, for any given $\eps_P\in(0,1)$ we can simply choose $\eps_E=\frac{1}{\eps_P}-1$ to get the result as desired.
We again split the expected value into two sums
\begin{align*}
\Ex{X_t^2} & =\sum_{A}{\sum_{B}\Prob{X_A=1 \wedge X_B=1}}\\
&=\sum_{B\colon B\nsim A}\Prob{X_A=1 \wedge X_B=1}+\sum_{B\colon B\sim A}\Prob{X_A=1 \wedge X_B=1},
\end{align*}
where $B\sim A$ denotes $F_A\cap F_B\neq\emptyset$. 
We will now consider the parts over $B\nsim A$ and $B\sim A$ separately, starting with $B\nsim A$.

As in the proof of \corref{threshold-coarse2}, we want to show that for any $\eps_E>0$, we can choose $\eps_1$ such that
\begin{equation}\label{eq:snakesIndep}
\sum_{A}\sum_{B\colon B\nsim A}\Prob{X_A=1 \wedge X_B=1}\le(1+\eps_E)\cdot \Ex{X_t}^2.
\end{equation}
It holds that
\begin{multline*}
\Prob{X_A=1 \wedge X_B=1}\\
=\binom{m}{4t} \cdot(4t)!\cdot \left(\prod_{c\in F_A}{\Prob{c}}\right)\cdot\left(\prod_{c\in F_B}{\Prob{c}}\right)\cdot\left(1-\sum_{c\in F_A\cup F_B}{\Prob{c}}\right)^{m-4t},
\end{multline*}
while
\begin{equation*}
\Prob{X_A=1}=\binom{m}{2t}\cdot(2t)!\cdot\left(\prod_{c\in F_A}{\Prob{c}}\right)\cdot\left(1-\sum_{c\in F_A}{\Prob{c}}\right)^{m-2t}.
\end{equation*}
This already gives us
\begin{multline*}
\Prob{X_A=1 \wedge X_B=1}\\
\le \Prob{X_A=1}\cdot\Prob{X_B=1}\cdot\frac{\left(1-\sum_{c\in F_A\cup F_B}{\Prob{c}}\right)^{m-4t}}{\left(1-\sum_{c\in F_A}{\Prob{c}}\right)^{m-2t}\left(1-\sum_{c\in F_B}{\Prob{c}}\right)^{m-2t}},
\end{multline*}
since $\binom{m}{4t}\cdot(4t)!\le \left(\binom{m}{2t}\cdot(2t)!\right)^2$.
Due to 
\[\left(1-\sum_{c\in F_A}{\Prob{c}}\right)\cdot\left(1-\sum_{c\in F_B}{\Prob{c}}\right)\ge 1-\sum_{c\in F_A\cup F_B}{\Prob{c}},\]
and $\sum_{c\in F_A}{\Prob{c}}, \sum_{c\in F_B}{\Prob{c}}\le 2t\cdot q_{\max}$ we have
\begin{multline*}
\frac{\left(1-\sum_{c\in F_A\cup F_B}{\Prob{c}}\right)^{m-4t}}{\left(1-\sum_{c\in F_A}{\Prob{c}}\right)^{m-2t}\left(1-\sum_{c\in F_B}{\Prob{c}}\right)^{m-2t}}\\
\le \left(1-\sum_{c\in F_A}{\Prob{c}}\right)^{-2t}\cdot\left(1-\sum_{c\in F_B}{\Prob{c}}\right)^{-2t}\le (1-2t\cdot q_{\max})^{-4t}.
\end{multline*}
Since we know $t\le (2\cdot q_{\max})^{-1}$ from the previous lemma, we can use Bernoulli's inequality to get
\[(1-2t\cdot q_{\max})^{-4t}\le (1-8t^2\cdot q_{\max})^{-1}.\]
We know that $t^2=f^{1/39}=(\sum_{i=1}^n p_i^2)^{1/39}/p_1^{2/39}$ and that $q_{\max}=\frac{1}{2}\cdot C\cdot p_1\cdot p_2=\tfrac{p_1\cdot p_2}{2\cdot(1-\sum_{i=1}^n p_i^2)}$.
Together with $p_2\le p_1\le \eps_1^{1/2}$ and $\sum_{i=1}^n p_i^2\le p_1\le \eps_1^{1/2}$ this yields
\begin{align*}
(1-8t^2\cdot q_{\max})^{-1}
&=\left(1-4\cdot\frac{(\sum_{i=1}^n p_i^2)^{1/39}\cdot p_1\cdot p_2}{p_1^{2/39}\cdot (1-\sum_{i=1}^n p_i^2)}\right)^{-1}\\
&\le\left(1-4\cdot\frac{p_1^{77/39}}{1-p_1}\right)^{-1}\\
&\le\left(1-4\cdot\frac{\eps_1^{77/78}}{1-\eps_1^{1/2}}\right)^{-1}\le 1+\eps_E
\end{align*}
for any $\eps_E>0$ if we make $\eps_1$ sufficiently small.
We now get
\[\Prob{X_A=1 \wedge X_B=1} \le \left(1+\eps_E\right)\cdot\Prob{X_A=1}\cdot\Prob{X_B=1}\]
for $A\nsim B$ and thus
\begin{align*}
\sum_{A}{\sum_{B\colon B\nsim A}\Prob{X_A=1 \wedge X_B=1}}
&\le \left(1+\eps_E\right)\cdot\sum_{A}{\sum_{B\colon B\nsim A}\Prob{X_A=1} \cdot\Prob{X_B=1}}\\
&\le (1+\eps_E)\cdot\Ex{X_2}^2.
\end{align*}

Second, we look at snakes $B\sim A$.
For those we want to show
\begin{equation}
\sum_{A}{\sum_{B\colon B\sim A}\Prob{X_A=1 \wedge X_B=1}}\le \eps_E\cdot\Ex{X_t}^2\label{eq:snakesDep}
\end{equation}
for any $\eps_E>0$ if we make $\eps_1$ sufficiently small.
First, let us consider $F_A=F_B$.
As in the case of $t=2$ it holds that there are exactly $4$ snakes with the same set of clauses.
Thus,
\[\sum_{A}{\sum_{B\colon F_B=F_A}\Prob{X_A=1 \wedge X_B=1}}=4\cdot\Ex{X_t}=\frac{4}{\Ex{X_t}}\cdot\Ex{X_t}^2.\]
\lemref{exp-snakes-sharp} tells us that for any $\eps>0$ we can choose $\eps_1$ such that $\Ex{X_t}\ge \eps$.
Therefore, for any $\eps_E>0$ we can choose $\eps_1$ sufficiently small to get
\[\sum_{A}{\sum_{B\colon F_B=F_A}\Prob{X_A=1 \wedge X_B=1}}=\frac{4}{\Ex{X_t}}\cdot\Ex{X_t}^2\le\frac{4}{\eps}\cdot\Ex{X_t}^2\le \eps_E \cdot \Ex{X_t}^2.\]

The remaining analysis is a bit more complicated than in the case of $t=2$, since we can not always surely say how many variables of snake $B$ are predefined by shared clauses.
As before, we are classifying snakes $B\sim A$ according to the number $l=\left|F_A\cap F_B\right|$ of shared clauses, but also according to the number $j$ of nodes in the variable-variable incidence graph $G_{F_A\cap F_B}$.
Note that the number of variables that $F_A$ and $F_B$ have in common (regardless of signs) could be greater!
In fact, they could share all their variables without having a single clause in common.
However, right now we are only interested in ways to incorporate clauses from $F_A$ as shared clauses into $F_B$.
To that end, we only need to consider the variables from these clauses as shared variables.
For a representation, see \figref{snake-intersection}.

\begin{figure}[t]
\centering
 \begin{subfigure}[t]{0.45\textwidth}
         \centering
         \includegraphics[width=\textwidth]{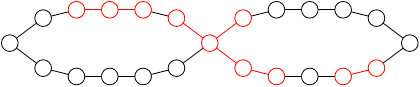}				
         \caption{If the shared clauses form a forest, the number of connected components $c$ is the difference between the number of different variables in shared clauses $j$ and the number of shared clauses $l$, $c=j-l$.}
     \end{subfigure}
     \hfill
     \begin{subfigure}[t]{0.45\textwidth}
         \centering
         \includegraphics[width=\textwidth]{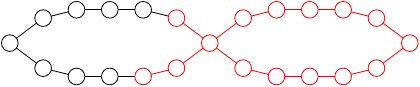}
         \caption{In order for the shared clauses to form a cycle, at least $t$ of the $2t$ clauses have to be shared. In that case the number of connected components is $c=j-l+1$.}
     \end{subfigure}
\caption{Visual representation of the variable-variable incidence graph $G_{F_A}$ for a snake $A$. Each node represents a variable of the snake, while each edge represents a clause of $A$ containing those variables. The node of degree $4$ represents the central variable of $A$. Shared clauses with snake $B$ are highlighted in red, \ie red edges represent clauses that appear both in $F_A$ and $F_B$. However, those edges do not necessarily appear at the same position in $G_{F_B}$.}
\label{fig:snake-intersection}
\end{figure}

Suppose now that snake $A$ and the shared clauses are fixed. 
We let $j$ denote the number of variables in shared clauses.
We know that there are $2t-1-j$ \emph{free} variables in $B$, \ie variables which are not predetermined to appear in $B$ by shared clauses.
Furthermore we can give an upper bound on the number $c$ of connected components of $G_{F_A\cap F_B}$.
It is easy to see that $c\le j-l$ for $l<t$ ($G_{F_A\cap F_B}$ is a forest), $c\le j-l+1$ for $t\le l <2t$ (we could create one cycle), and $c=j-l+2$ for $l=2t$ ($F_A=F_B$).
These cases are also visualized in \figref{snake-intersection}.
Fixing $l$ and $j$ it holds that
\begin{align}
	& \sum_{\substack{\text{snakes $A$, $B$}\colon\\ |E(G_{F_A\cap F_B})|=l,\ |V(G_{{F_A\cap F_B}})|=j}}{\Prob{X_A=1 \wedge X_B=1}} \notag\\
	& \le \binom{m}{4t-l}\cdot (4t-l)!\cdot\left(\frac{C}{2}\right)^{4t-l}\cdot 2^{2t-1}\cdot(2t-2)!\cdot\left(\sum_{\substack{S_A\subseteq [n]\colon\\|S_A|=2t-2}}\prod_{x\in S_A}p(x)^2\right)\cdot \notag\\
	& \cdot\left(\sum_{y\in[n]}{p(y)^4}\right)\cdot4\cdot\left(\binom{2t+2}{2(j-l)+2}\right)^2\cdot c!\cdot2^{c}\cdot 2t\cdot(2t-1-j)!\cdot2^{2t-1-j}\cdot\notag\\
	& \cdot\left(\sum_{\substack{S_B\subseteq [n]\colon\\|S_B|=2t-1-j}}\prod_{x\in S_B}p(x)^2\right)\cdot p_1^{2(j-l+1)}\cdot\left(1-\sum_{c\in F_A\cup F_B}\Prob{c}\right)^{m-(4t-l)}.\label{eq:snake-estimate}
\end{align}
Before we upper bound this expression even further, let us explain where it comes from. 
There are $\binom{m}{4t-l}\cdot (4t-l)!$ positions for the $4t-l$ clauses of $F_A\cup F_B$ in the $m$-clause formula $\Phi$.
There are at most $2^{2t-2}\cdot(2t-2)!$ possibilities of forming different snakes (signs and positions) from the $2t-2$ variables of $A$, excluding $y=|w_t|$, and two possible signs for $y=|w_t|$.
In snake $A$ each variable appears exactly twice, except for $y=|w_t|$, which appears four times.
Now we want to count the ways of mapping $G_{F_A\cap F_B}$ to $G_{F_A}$ and $G_{F_B}$.
Following the argumentation from \cite{chvatalreed92} we can see that there are $2\binom{2t+2}{2j-2l+2}$ possible mappings for $G_{F_A}$ and $G_{F_B}$, respectively.
These mappings fix the shared clauses we choose from $A$ as well as the positions where shared clauses can appear in $B$, but not where exactly which clause will appear.
This is what we consider next.
We know that $G_{F_A\cap F_B}$ contains $c$ connected components.
If they are of same length, they can be interchanged in $c!$ ways.
Furthermore, each component might be flipped, \ie the sign of every literal in the component and their order in $B$ can be inverted. 
For components which are paths, this does not change the set of shared clauses they originate from.
Nevertheless, there is still the possibility of having one component which is not a path.
For this component there are at most $2t$ ways of mapping it onto its counterpart (if it is a cycle) due to \cite{chvatalreed92}.
Now we know the shared clauses from $F_A$ and the exact position of these clauses in $F_B$ as well as positions reserved for non-determined variables in snake $B$.
The remaining $2t-1-j$ non-determined variables from $B$ can be chosen arbitrarily.
Also, there are $2^{2t-1-j}\cdot(2t-1-j)!$ possibilities for them to fill out the blanks of snake $B$.
Each of these variables appears at least twice in $B$ only.
The remaining at most $2(j-l+1)$ appearances of variables in $F_B$ are determined by the previous choices and give an additional factor of at most $p_1^{2(j-l+1)}$.
Note that the case that one of our free variables in $B$ is a central variable is also captured by this upper bound, since $\sum_{i=1}^n p_i^4\le p_1^2\cdot \sum_{i=1}^n p_i^2$.
The other $m-(4t-l)$ clauses of $\Phi$ are supposed to be different from those in $F_A\cup F_B$, so that both $F_A$ and $F_B$ appear exactly once.

Now we want to simplify that expression.
It holds that 
\[\left(1-\sum_{c\in F_A\cup F_B}\Prob{c}\right)^{m-(4t-l)} \le 1\] and that 
\[C^{4t-l}\le\left(1+\frac{\sum_{i=1}^n p_i^2}{1-\sum_{i=1}^n p_i^2}\right)^{4t}\le\exp\left(4t\cdot\frac{\sum_{i=1}^n p_i^2}{1-\sum_{i=1}^n p_i^2}\right).\]
We know that $t=f^{1/78}\le p_1^{-1/78}$ and that $\sum_{i=1}^n p_i^2\le p_1\le \eps_1^{1/2}$. This implies
\[C^{4t-l}\le\exp\left(4t\cdot\frac{\sum_{i=1}^n p_i^2}{1-\sum_{i=1}^n p_i^2}\right)\le\exp\left(4\cdot\frac{p_1^{77/78}}{1-p_1}\right)\le\exp\left(4\cdot\frac{\eps_1^{77/156}}{1-\eps_1^{1/2}}\right).\]
For any $\eps>0$ we can choose $\eps_1$ small enough such that this expression is at most $1+\eps$.
Again
\[\sum_{\substack{S\subseteq [n]\colon\\|S|=x}}\prod_{s\in S}p(s)^2 \le \frac{1}{x!}\left(\sum_{i=1}^n p_i^2\right)^x\]
according to \lemref{aux1}.
This step also cancels out the factors $(2t-2)!$ and $(2t-1-j)!$.
Also, all factors of $2$ that appear cancel out with $c\le j-l+2$.
We will also use the following estimation
\[\left(\binom{2t+2}{2(j-l)+2}\right)^2\cdot c!\le \frac{(2t+2)^{4(j-l+1)}}{(2(j-l+1)!)^2}\cdot(j-l+2)! \le (2t+2)^{4(j-l+1)}\le (3t)^{4(j-l+1)}.\]
This holds since $j\ge l-1$ and $t\ge2$.
However, $j=l-1$ only happens if $F_A=F_B$.
Since we already considered this case, we will further assume $j\ge l$.
Plugging everything back into \eq{snake-estimate} we get
\begin{align}
	& \sum_{\substack{\text{snakes $A$, $B$}\colon\\ |E(G_{F_A\cap F_B})|=l,\ |V(G_{{F_A\cap F_B}})|=j}}{\Prob{X_A=1 \wedge X_B=1}} \notag\\
	& \le 4\cdot (1+\eps)\cdot m^{4t-l}\cdot (3t)^{5(j-l+1)}\cdot\left(\sum_{i=1}^n p_i^4\right)\cdot\left(\sum_{i=1}^n p_i^2\right)^{4t-j-3}\cdot p_1^{2(j-l+1)} \label{eq:snake-estimate2}
\end{align}
for some $\eps>0$ that decreases as $\eps_1$ does.

We will distinguish three cases now, depending on the value of $j-l$.
First $j-l=0$, then $j-l\ge2$ and finally $j-l=1$.
For each of these cases we want to show that for any $\eps_E>0$ we can choose $\eps_1$ small enough so that
\[\sum_{\substack{\text{snakes $A$, $B$}\colon\\ |E(G_{F_A\cap F_B})|=l,\ |V(G_{{F_A\cap F_B}})|=j}}{\Prob{X_A=1 \wedge X_B=1}}\le\eps_E\cdot\frac{\Ex{X_t}^2}{t^2}.\]
Since $1\le l \le 2t$ and $2\le j\le 2t-1$, we will get an additional factor of $4t^2$ when summing over all snakes $A\sim B$.
If we consider all cases, including $F_A=F_B$, this adds up to
\[\sum_{A}\sum_{B\colon B\sim A}{\Prob{X_A=1 \wedge X_B=1}}\le16\cdot\eps_E\cdot\Ex{X_t}^2.\]
Still, for any chosen $\eps_E>0$ we can choose $\eps_1\in(0,1)$ small enough to make this expression at most $\eps_E\cdot\Ex{X_t}^2$ as desired.

Now let us consider the first case, $j=l$.
This can only happen if $G_{F_A\cap F_B}$ contains a cycle, as we can see in \figref{snake-intersection}.
However, $G_{F_A\cap F_B}$ can only contain a cycle if $l\ge t$. 
Due to \eq{snake-estimate2} it holds that
\begin{align*}
	& \sum_{\substack{\text{snakes $A$, $B$}\colon\\ |E(G_{F_A\cap F_B})|=l,\ |V(G_{{F_A\cap F_B}})|=l}}{\Prob{X_A=1 \wedge X_B=1}} \\
	& \le 4\cdot (1+\eps)\cdot m^{4t-l}\cdot (3t)^{5}\cdot\left(\sum_{i=1}^n p_i^4\right)\cdot\left(\sum_{i=1}^n p_i^2\right)^{4t-l-3}\cdot p_1^{2}
\end{align*}
Remember that due to \lemref{exp-snakes-sharp} for any $\eps\in(0,1)$ we can choose $\eps_1\in(0,1)$ small enough so that
\[\Ex{X_t}^2 \ge (1-\eps)\cdot\frac14\cdot m^{4t}\cdot\left(\sum_{i=1}^n p_i^4\right)^2\cdot\left(\sum_{i=1}^n p_i^2\right)^{4t-4}.\]
Thus,
\begin{align*}
	&\sum_{\substack{\text{snakes $A$, $B$}\colon\\ |E(G_{F_A\cap F_B})|=l,\ |V(G_{{F_A\cap F_B}})|=l}}{\Prob{X_A=1 \wedge X_B=1}} \\
	&\le 16\cdot \frac{1+\eps}{1-\eps}\cdot 3^5 \cdot t^5\cdot\left(m\cdot\sum_{i=1}^n p_i^2\right)^{-l}\cdot\frac{p_1^2\cdot\sum_{i=1}^n p_i^2}{\sum_{i=1}^n p_i^4}\cdot \Ex{X_t}^2,
\end{align*}
Here, we can choose $\eps$ arbitrarily small by making $\eps_1$ sufficiently small.
Due to $m=\eps_m/\sum_{i=1}^n p_i^2$, $l\ge t$, and $\sum_{i=1}^n p_i^4\ge p_1^4$ this yields
\begin{align*}
	\sum_{\substack{\text{snakes $A$, $B$}\colon\\ |E(G_{F_A\cap F_B})|=l,\ |V(G_{{F_A\cap F_B}})|=l}}{\Prob{X_A=1 \wedge X_B=1}}
	&\le 16\cdot \frac{1+\eps}{1-\eps}\cdot 3^5 \cdot t^5\cdot\eps_m^{-t}\cdot f\cdot \Ex{X_t}^2.
\end{align*}
Since $t=f^{1/78}$ and $\eps_m>1$, we can make this expression at most $\eps_E\cdot \Ex{X_t}^2$ for any $\eps_E>0$ by making $f$ sufficiently large.
Due to $f\ge 1/\eps_1$, we can also make $\eps_1$ sufficiently small.
This gives us the result for the first case as desired.

The second case we consider is $j-l\ge2$.
It holds that
\begin{align*}
	& \sum_{\substack{\text{snakes $A$, $B$}\colon\\ |E(G_{F_A\cap F_B})|=l,\ |V(G_{{F_A\cap F_B}})|\ge l+2}}{\Prob{X_A=1 \wedge X_B=1}} \\
	& \le 4\cdot (1+\eps)\cdot m^{4t-l}\cdot (3t)^{5(j-l+1)}\cdot\left(\sum_{i=1}^n p_i^4\right)\cdot\left(\sum_{i=1}^n p_i^2\right)^{4t-j-3}\cdot p_1^{2(j-l+1)}.
\end{align*}
As before, \lemref{exp-snakes-sharp} tells us that for any $\eps\in(0,1)$ we can choose $\eps_1\in(0,1)$ small enough so that
\[\Ex{X_t}^2 \ge (1-\eps)\cdot\frac14\cdot m^{4t}\cdot\left(\sum_{i=1}^n p_i^4\right)^2\cdot\left(\sum_{i=1}^n p_i^2\right)^{4t-4}.\]
Thus,
\begin{align*}
	& \sum_{\substack{\text{snakes $A$, $B$}\colon\\ |E(G_{F_A\cap F_B})|=l,\ |V(G_{{F_A\cap F_B}})|\ge l+2}}{\Prob{X_A=1 \wedge X_B=1}} \\
	& \le 16\cdot \frac{1+\eps}{1-\eps}\cdot (3t)^{5(j-l+1)}\cdot m^{-l}\cdot\left(\sum_{i=1}^n p_i^2\right)^{-j+1}\cdot\left(\sum_{i=1}^n p_i^4\right)^{-1}\cdot p_1^{2(j-l+1)} \cdot\Ex{X_t}^2\\
	\intertext{and since $\sum_{i=1}^n p_i^4\ge p_1^4$, we get}
	& \le 16\cdot \frac{1+\eps}{1-\eps}\cdot (3t)^{5(j-l+1)}\cdot \left(m\cdot\left(\sum_{i=1}^n p_i^2\right)\right)^{-l}\cdot\frac{p_1^{2(j-l+1)}}{p_1^4\cdot\left(\sum_{i=1}^n p_i^2\right)^{j-l-1}}\cdot\Ex{X_t}^2.\\
	\intertext{Again, we can use $m=\eps_m\cdot\left(\sum_{i=1}^n p_i^2\right)^{-1}$ to get}
	& = 16\cdot \frac{1+\eps}{1-\eps}\cdot (3t)^{5(j-l+1)}\cdot \eps_m^{-l}\cdot\frac{p_1^{2(j-l-1)}}{\left(\sum_{i=1}^n p_i^2\right)^{j-l-1}}\cdot\Ex{X_t}^2\\
	\intertext{and $f=p_1^2/\left(\sum_{i=1}^n p_i^2\right)$, which yields}
	& = 16\cdot \frac{1+\eps}{1-\eps}\cdot (3t)^{5(j-l+1)}\cdot \eps_m^{-l}\cdot f^{-(j-l-1)}\cdot\Ex{X_t}^2.\\
	\intertext{Since we know that $t=f^{1/78}$ we get}
	& = 16\cdot \frac{1+\eps}{1-\eps}\cdot \eps_m^{-l}\cdot (3t)^{10}\cdot\left(\frac{(3t)^5}{t^{78}}\right)^{j-l-1}\cdot\Ex{X_t}^2.\\
\end{align*}
Since we know that $j-l\ge 2$ and $\eps_m>1$, we can make this expression at most $\eps_E\cdot\Ex{X_t}^2/t^2$ for any $\eps_E>0$ by making $t$ sufficiently large.
The same holds if we make $\eps_1$ sufficiently small, because $t=f^{1/78}\ge {\eps_1}^{-1/78}$. 
As we do so, $\eps$ decreases as well.

The last case we consider is $j-l=1$.
This happens if we either only have one connected component in $G_{F_A\cap F_B}$ that does not form a cycle or if $G_{F_A\cap F_B}$ contains a cycle and one other connected component.
In the latter case, \eq{snake-estimate2} gives us
\begin{align*}
	& \sum_{\substack{\text{snakes $A$, $B$}\colon\\ |E(G_{F_A\cap F_B})|=l,\ |V(G_{{F_A\cap F_B}})|=l+1\\
	\text{cycle in $G_{F_A\cap F_B}$}}}{\Prob{X_A=1 \wedge X_B=1}} \\
	& \le 4\cdot (1+\eps)\cdot m^{4t-l}\cdot (3t)^{5(j-l+1)}\cdot\left(\sum_{i=1}^n p_i^4\right)\cdot\left(\sum_{i=1}^n p_i^2\right)^{4t-j-3}\cdot p_1^{2(j-l+1)}\\
	& = 4\cdot (1+\eps)\cdot m^{4t-l}\cdot (3t)^{10}\cdot\left(\sum_{i=1}^n p_i^4\right)\cdot\left(\sum_{i=1}^n p_i^2\right)^{4t-l-4}\cdot p_1^{4},\\
\intertext{where we can choose the value of $\eps\in(0,1)$ by making $\eps_1$ sufficiently small.
As before, we can use the estimate
\[\Ex{X_t}^2 \ge (1-\eps)\cdot\frac14\cdot m^{4t}\cdot\left(\sum_{i=1}^n p_i^4\right)^2\cdot\left(\sum_{i=1}^n p_i^2\right)^{4t-4}.\]
from \lemref{exp-snakes-sharp} to achieve an upper bound of}
	& \le 16\cdot3^{10}\cdot \frac{1+\eps}{1-\eps}\cdot t^{10}\cdot\left(m\cdot\sum_{i=1}^n p_i^2\right)^{-l}\cdot\frac{p_1^4}{\sum_{i=1}^n p_i^4}\cdot \Ex{X_t}^2.\\
	\intertext{Since a cycle can only exist for $l\ge t$, due to the requirement $m=\eps_m/\sum_{i=1}^n p_i^2$ for some $\eps_m>1$, and with $p_1^4\le \sum_{i=1}^n p_i^4$ it holds that this is}
	& \le 16\cdot3^{10}\cdot \frac{1+\eps}{1-\eps}\cdot t^{10}\cdot\eps_m^{-t}\cdot\Ex{X_t}^2.
\end{align*}
As in the case of $j-l=0$, where $G_{F_A\cap F_B}$ also contained a cycle, we see that for any $\eps_E>0$ we can bound this expression by $\eps_E\cdot\Ex{X_t}^2/t^2$ as desired by making $\eps_1$ sufficiently small, which makes $t$ sufficiently large.

If $j-l=1$ and $G_{F_A\cap F_B}$ does not contain a cycle, we have to look a bit more closely, since we cannot guarantee a large enough $l$ to make the expression sufficiently small.
Instead, we will consider different cases for mapping the central variable of $B$.
These cases will result in slightly better bounds than the one in \eq{snake-estimate2}.

First, we assume that $B$'s central variable is a free variable, \ie the central variable of $B$ does not appear in \emph{any} shared clauses of $A$ and $B$.
This means, we can actually choose $B$'s central variable freely and it will appear at least $4$ times in $F_A\cup F_B$.
In \eq{snake-estimate2} we assumed that each of our free variables only contributed $\sum_{i=1}^n p_i^2$.
However, in the current case, one of them (the central one) contributes $\sum_{i=1}^n p_i^4$.
Thus, we can substitute a factor of $(\sum_{i=1}^n p_i^2)\cdot p_1^2$ in \eq{snake-estimate} with $\sum_{i=1}^n p_i^4$ to get
\begin{align}
	&\sum_{\substack{\text{snakes $A$, $B$}\colon\\ |E(G_{F_A\cap F_B})|=l,\ |V(G_{{F_A\cap F_B}})|=l+1,\\\text{central of $B$ is free}}}{\Prob{X_A=1 \wedge X_B=1}}\notag\\
	& \le 4\cdot (1+\eps)\cdot m^{4t-l}\cdot (3t)^{10}\cdot\left(\sum_{i=1}^n p_i^4\right)^2\cdot\left(\sum_{i=1}^n p_i^2\right)^{4t-l-5}\cdot p_1^{2}.\label{eq:free-central}
\end{align}
As in the cases before, we use the lower bound on $\Ex{X_t}$ from \lemref{exp-snakes-sharp} and our definition $(\sum_{i=1}^n p_i^2)/p_1^2=f=t^{78}$ to get
\begin{align*}
	&\sum_{\substack{\text{snakes $A$, $B$}\colon\\ |E(G_{F_A\cap F_B})|=l,\ |V(G_{{F_A\cap F_B}})|=l+1,\\\text{central of $B$ not in $F_A\cap F_B$}}}{\Prob{X_A=1 \wedge X_B=1}}\\
	& \le 16\cdot 3^{10}\cdot \frac{1+\eps}{1-\eps}\cdot t^{10}\cdot\left(m\cdot\sum_{i=1}^n p_i^2\right)^{-l}\cdot\frac{p_1^2}{\sum_{i=1}^n p_i^2}\cdot\Ex{X_t}^2\\
	&= 4\cdot 3^{10}\cdot \frac{1+\eps}{1-\eps}\cdot\eps_m^{-l}\cdot t^{10}\cdot f^{-1}\cdot \Ex{X_t}^2\\
	&= 4\cdot 3^{10}\cdot \frac{1+\eps}{1-\eps}\cdot\eps_m^{-l}\cdot t^{-68}\cdot \Ex{X_t}^2.
\end{align*}
It is obvious that for any $\eps_E>0$ this is at most $\eps_E\cdot\Ex{X_t}^2/t^2$ as desired if we choose $t\ge {\eps_1}^{-1}$ sufficiently large or, conversely, $\eps_1$ sufficiently small. 

Now we assume that the central variable in $B$ is not free.
What could happen? It could coincide with a non-central variable from $A$ or with the central variable from $A$.
Thus, the central variable of $B$ could already appear once or twice in shared clauses in the first and one to four times in the second case.

Let us start with the case that it coincides with a non-central variable in $A$.
Then, one of the variables that appears twice in $A$ appears an additional (not in shared clauses) $2$ or $3$ times as the central node in $B$, depending on the number of shared clauses it already appears in.
In total it either appears $4$ times or $5$ times in $F_A\cup F_B$. 
For a representation of those two cases, see \figref{central-1} and \figref{central-2}.

\begin{figure}[t]
\centering
\begin{subfigure}[t]{0.2\textwidth}
	\centering
	\includegraphics[width=\textwidth]{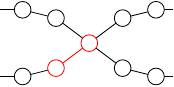}
	\caption{The central variable of $B$ appears in one shared clause.}
	\label{fig:central-1}
\end{subfigure}
\hfill
\begin{subfigure}[t]{0.2\textwidth}
	\centering
	\includegraphics[width=\textwidth]{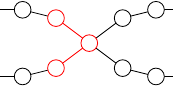}
	\caption{The central variable of $B$ appears in two shared clauses.}
	\label{fig:central-2}
\end{subfigure}
\hfill
\begin{subfigure}[t]{0.2\textwidth}
	\centering
	\includegraphics[width=\textwidth]{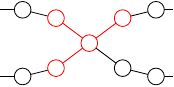}
	\caption{The central variable of $B$ appears in three shared clauses.}
\end{subfigure}
\hfill
\begin{subfigure}[t]{0.2\textwidth}
	\centering
	\includegraphics[width=\textwidth]{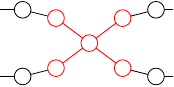}
	\caption{The central variable of $B$ appears in four shared clauses.}
\end{subfigure}
\caption{Snapshot of $B$'s central variable in $G_{F_B}$. Shared clauses of $F_A$ and $F_B$ are highlighted in red. If the central variable appears in $x$ shared clauses, then there are $x$ variables that appear exactly once in shared clauses.
Then, $B$'s central variable appears an additional $4-x$ times in $B$ and the variables that appear only once in shared clauses, each appear one additional time in $B$.}
\label{fig:central-unfree}
\end{figure}

Thus, we can replace the two appearances of a variable in $A$ and $2$ resp. $3$ appearances of unfree variables in $B$ with $4$ resp. $5$ appearances of a variable in total ($A$ and $B$).
That is, we multiply the expression from \eq{snake-estimate2} with $(\sum_{i=1}^n p_i^4)/(p_1^2\cdot \sum_{i=1}^n p_i^2)$ resp. $(\sum_{i=1}^n p_i^5)/(p_1^3\cdot \sum_{i=1}^n p_i^2)$.
Since $\sum_{i=1}^n p_i^5\le p_1\sum_{i=1}^n p_i^4$, the former case gives us an upper bound.
We get
\begin{align*}
	&\sum_{\substack{\text{snakes $A$, $B$}\colon\\ |E(G_{F_A\cap F_B})|=l,\ |V(G_{{F_A\cap F_B}})|=l+1,\\\text{central of $B$ not free and not central of $A$}}}{\Prob{X_A=1 \wedge X_B=1}}\\
	& \le 4\cdot (1+\eps)\cdot m^{4t-l}\cdot (3t)^{10}\cdot\left(\sum_{i=1}^n p_i^4\right)^2\cdot\left(\sum_{i=1}^n p_i^2\right)^{4t-l-5}\cdot p_1^{2}.
\end{align*}
This is the same upper bound we had in the previous case, \eq{free-central}, when the central variable of $B$ was free.
Thus, we already know that for any $\eps_E>0$ we can choose $\eps_1$ small enough to get a bound of at most $\eps_E\cdot \Ex{X_t}^2/t^2$.

The last case is that the central variable of $B$ coincides with the central variable of $A$.
Then, the variable that appears $4$ times in $A$ might appear $0$ to $3$ additional times (\ie not in shared clauses) in $B$, depending on the number of shared clauses it already appears in.
It cannot appear an additional $4$ times, since the central variable of $A$ must appear in a shared clause at least once for the variable to not be free.
Remember that we are in the case where $G_{F_A\cap F_B}$ only contains one connected component that is not a cycle.
This means, we have at most $4$ variables in shared clauses that each appear one additional time in $B$.
See \figref{central-unfree} for a visual representation of those cases.
Let $x\in\left\{1,2,3,4\right\}$ be the number of times that the central variable of $A$ appears in shared clauses.
Then, it appears an additional $4-x$ times in $B$. 
In addition to the central variable, there are now $x$ other unfree variables that each appear one additional time in $B$.
Each of these variables actually appears $3$ times in $A$ and $B$ together instead of $2$ times in $A$ and once as a single predetermined variable in $B$.
As before, we can substitute their appearances by multiplying \eq{snake-estimate2} with a factor of $(\sum_{i=1}^n p_i^3)/(p_1\cdot \sum_{i=1}^n p_i^2)$ for each of them.
By handling the shared central variable of $A$ and $B$ in the same way, we get an additional factor of $(\sum_{i=1}^n p_i^{8-x})/(p_1^{4-x}\sum_{i=1}^n p_i^{4})$.
We now get
\begin{align*}
	&\sum_{\substack{\text{snakes $A$, $B$}\colon\\ |E(G_{F_A\cap F_B})|=l,\ |V(G_{{F_A\cap F_B}})|=l+1,\\\text{central of $B$ is central of $A$, appears in $x$ shared clauses}}}{\Prob{X_A=1 \wedge X_B=1}}\\
	& \le 4\cdot (1+\eps)\cdot m^{4t-l}\cdot (3t)^{10}\cdot\left(\sum_{i=1}^n p_i^{8-x}\right)\cdot\left(\sum_{i=1}^n p_i^2\right)^{4t-l-4-x}\cdot\left(\sum_{i=1}^n p_i^3\right)^{x}.
	\intertext{Again, we can use the lower bound on $\Ex{X_t}$ from \lemref{exp-snakes-sharp} to get}
	& \le 16 \cdot3^{10}\cdot\frac{1+\eps}{1-\eps}\cdot t^{10}\cdot\left(m\cdot\sum_{i=1}^n p_i^2\right)^{-l}\cdot\frac{\left(\sum_{i=1}^n p_i^{8-x}\right)\cdot\left(\sum_{i=1}^n p_i^3\right)^{x}}{\left(\sum_{i=1}^n p_i^4\right)^{2}\left(\sum_{i=1}^n p_i^2\right)^{x}}\cdot\Ex{X_t}^2\\
	\intertext{and $m\ge1/\sum_{i=1}^n p_i^2$ implies}
	&\le 16 \cdot3^{10}\cdot\frac{1+\eps}{1-\eps}\cdot t^{10}\cdot\frac{\left(\sum_{i=1}^n p_i^{8-x}\right)\cdot\left(\sum_{i=1}^n p_i^3\right)^{x}}{\left(\sum_{i=1}^n p_i^4\right)^{2}\left(\sum_{i=1}^n p_i^2\right)^{x}}\cdot\Ex{X_t}^2.
\end{align*}

It remains to show that for any $\eps_E>0$ we can choose $\eps_1$ small enough so that
\[t^{10}\cdot\frac{\left(\sum_{i=1}^n p_i^{8-x}\right)\cdot\left(\sum_{i=1}^n p_i^3\right)^{x}}{\left(\sum_{i=1}^n p_i^4\right)^{2}\left(\sum_{i=1}^n p_i^2\right)^{x}}\le \frac{\eps_E}{t^2}.\]
First, note that $\sum_{i=1}^n p_i^{8-x}\le p_1^{4-x}\cdot\sum_{i=1}^n p_i^{4}$ and thus
\[t^{10}\cdot\frac{\left(\sum_{i=1}^n p_i^{8-x}\right)\cdot\left(\sum_{i=1}^n p_i^3\right)^{x}}{\left(\sum_{i=1}^n p_i^4\right)^{2}\left(\sum_{i=1}^n p_i^2\right)^{x}} 
\le t^{10}\cdot\frac{p_1^{4-x}\cdot\left(\sum_{i=1}^n p_i^3\right)^{x}}{\left(\sum_{i=1}^n p_i^4\right)\left(\sum_{i=1}^n p_i^2\right)^{x}}\]
In order to further bound this expression, we consider the probability vector $\vec{p}^{(n)}=p_1,p_2,\ldots,p_n$.
We now split the probabilities into those with $p_i\ge p_1/f^{1/6}$ and those with $p_i< p_1/f^{1/6}$.
Let $N=|\left\{i\in[n]\mid p_i\ge p_1/f^{1/6}\right\}|$ be the number of probabilities in $\vec{p}^{(n)}$ larger than the bound we set.
We now distinguish two cases: $N\ge f^{5/6}$ and $N<f^{5/6}$.

Assume the first case, $N\ge f^{5/6}$.
It holds that 
\[\sum_{i=1}^n p_i^{4} \ge N\cdot\left(\frac{p_1}{f^{1/6}}\right)^4 = p_1^4 \cdot f^{1/6}.\]
Together with $\sum_{i=1}^n p_i^3\le p_1\cdot \sum_{i=1}^n p_i^2$ this implies
\begin{align*}
t^{10}\cdot\frac{p_1^{4-x}\cdot\left(\sum_{i=1}^n p_i^3\right)^{x}}{\left(\sum_{i=1}^n p_i^4\right)\left(\sum_{i=1}^n p_i^2\right)^{x}}
&\le t^{10}\cdot\frac{p_1^{4}\cdot \left(\sum_{i=1}^n p_i^2\right)^{x}}{p_1^4\cdot f^{1/6}\left(\sum_{i=1}^n p_i^2\right)^{x}}\\
& = t^{10}\cdot f^{-1/6} = t^{-3}\le\eps_1^{1/78}\cdot t^{-2}
\end{align*}
as desired, due to our choice $t=f^{1/78}$ and since we can make $\eps_1$ as small as necessary.

Now assume $N<f^{5/6}$.
It holds that
\begin{align*}
\sum_{i=1}^n p_i^3 &
< N\cdot p_1^3 + \frac{p_1}{f^{1/6}}\cdot \sum_{i=1}^n p_i^2\\
& \le p_1^3\cdot f^{5/6} + p_1^3\cdot f^{5/6}=2\cdot p_1^3\cdot f^{5/6},
\end{align*}
where we used $\sum_{i=1}^n p_i^2=f\cdot p_1^2$.
With $\sum_{i=1}^n p_i^4\ge p_1^4$ and $\sum_{i=1}^n p_i^2=f\cdot p_1^2$ this readily implies
\begin{align*}
t^{10}\cdot\frac{p_1^{4-x}\cdot\left(\sum_{i=1}^n p_i^3\right)^{x}}{\left(\sum_{i=1}^n p_i^4\right)\left(\sum_{i=1}^n p_i^2\right)^{x}}
&\le t^{10}\cdot\frac{p_1^{4-x}\cdot\left(2\cdot p_1^3\cdot f^{5/6}\right)^{x}}{p_1^4\cdot\left(p_1^2\cdot f\right)^{x}}\\
& \le t^{10}\cdot 2^4 \cdot f^{-x/6} \le t^{10}\cdot 2^4 \cdot f^{-1/6}\le 16\cdot\eps_1^{1/78}\cdot t^{-2}.
\end{align*}
Again we can make this as small as any $\eps_E/t^2$ if we choose $\eps_1$ sufficiently small.

Finally, we took care of all the cases for $j-l=1$ and showed
\[\sum_{\substack{\text{snakes $A$, $B$}\colon\\ |E(G_{F_A\cap F_B})|=l,\ |V(G_{{F_A\cap F_B}})|=l+1}}{\Prob{X_A=1 \wedge X_B=1}}\le\eps_E\cdot\frac{\Ex{X_t}^2}{t^2}\]
as desired.
This implies 
\[\sum_{A}\sum_{B\colon B\sim A}{\Prob{X_A=1 \wedge X_B=1}}\le \eps_E\cdot \Ex{X_t}^2\]
and concludes the proof.
\end{proof}

\lemref{threshold-sharp} and \lemref{sharp-lb} now establish the existence of a sharp threshold at $m=\left(\sum_{i=1}^n p_i^2\right)^{-1}$ as we will see in \secref{main-theorem}. 
However, we first have to show an upper bound for $p_1^2\in\Theta(\sum_{i=1}^n p_i^2)$ and $p_2^2\in\Theta(\sum_{i=2}^n p_i^2)$, or more generally, for the case that we are given constants $\eps_1,\eps_2\in(0,1)$ with $p_1^2\ge\eps_1\cdot\sum_{i=1}^n p_i^2$ and $p_2^2\ge\eps_2\cdot\sum_{i=2}^n p_i^2$.
This case will be handled in the next section.

\section{A Simple Upper Bound on the Satisfiability Threshold}\label{sec:coarse}

This section handles the case that there are constants $\eps_1,\eps_2\in(0,1)$ with $p_1^2\ge\eps_1\cdot\sum_{i=1}^n p_i^2$ and $p_2^2\ge\eps_2\cdot\sum_{i=2}^n p_i^2$.
This especially includes $p_1^2\in\Theta(\sum_{i=1}^n p_i^2)$ and $p_2^2\in\Theta(\sum_{i=2}^n p_i^2)$.
This case is particularly easy, since it implies $(C\cdot p_1\cdot(\sum_{i=2}^n p_i^2)^{1/2})^{-1}\in\Theta(q_{\max}^{-1})$.
That means, the probability for a formula to be unsatisfiable is dominated by the highest clause probability.

We are going to show that there is a coarse threshold at $m^\star=(C\cdot p_1\cdot(\sum_{i=2}^n p_i^2)^{1/2})^{-1}\in\Theta(q_{\max}^{-1})$.
Note that \lemref{bicycle-coarse} only assumes $p_1^2\ge\eps_1\cdot\sum_{i=1}^n p_i^2$.
Thus, the lemma already handles $m<\eps_m\cdot (C\cdot p_1\cdot(\sum_{i=2}^n p_i^2)^{1/2})^{-1}$ for sufficiently small constants $\eps_m\in(0,1)$.
Now we only have to see what happens for $m\in\Omega(m^\star)$.

In the following lemma we give a lower bound on the probability to generate an unsatisfiable instance by showing the existence of an unsatisfiable sub-formula consisting only of clauses with the highest clause probability.
These are the clauses consisting of the two most-probable Boolean variables.
The lemma generally holds for $k\ge2$, but it especially serves our purpose of considering this easy case.

\begin{lemma}\label{lem:coarseness1}
Let $\Phi\sim\mathcal{D}^N(n,k,(\vec{p}^{(n)})_{n\in\N},m)$ be a non-uniform random k-SAT formula and let $q_{\max}$ denote the maximum clause probability. 
Then, $\Phi$ is unsatisfiable with probability at least
\[\left(1-e^{-q_{\max}\cdot m}\right)^{2^k}-q_{\max}^2\cdot 2^{2k}\cdot m \cdot\left(1+e^{-q_{\max}\cdot m}\right)^{2^k}.\]
\end{lemma}
\begin{proof}

Let $c$ be the clause with maximum probability. 
Since the signs of literals are chosen with probability $1/2$ independently at random, it holds that each clause with the same variables as $c$ has the same probability.
Our lower bound is now just a lower bound on the probability of having each of the $2^k$ clauses with these variables, which constitute an unsatisfiable sub-formula.
Let us enumerate the different clauses $c_1,\ldots,c_{2^k}$ with variables $X_1,\ldots,X_k$ in an arbitrary order.
Now let $\n{A_j}$ denote the event that $c_j$ is \emph{not} appearing in $\Phi$ and let $\n{A}=\bigcup_{j\in[2^k]}\n{A_j}$ denote the event that at least one of these clauses does not appear.
Due to the principle of inclusion and exclusion it holds that
\[\Pr{\n{A}}=\sum_{l=1}^{2^k}(-1)^{l+1}\sum_{J\subseteq[2^k]\colon|J|=l}\Pr{\bigcap_{j\in J} \n{A_j}}=\sum_{l=1}^{2^k} (-1)^{l+1}\left(\binom{2^k}{l} \cdot\left(1-l\cdot q_{\max}\right)^m\right),\]
because the clauses $c_1,\ldots,c_{2^k}$ have the same probability $q_{\max}$ of appearing and all clauses are drawn independently at random.
It now holds that
\begin{align*}
\Pr{\Phi\text{ unsat}}
\ge \Pr{A} &= 1-\left(\sum_{l=1}^{2^k} \binom{2^k}{l} \cdot(-1)^{l+1}\cdot\left(1-l\cdot q_{\max}\right)^m\right)\\
&=\sum_{l=0}^{2^k}\left(\binom{2^k}{l}\cdot(-1)^{l}\cdot\left(1-l\cdot q_{\max}\right)^m\right).
\end{align*}
We can now estimate
\[-\left(1-q_{\max}\cdot l\right)^m\ge-e^{-q_{\max}\cdot l\cdot m}\]
and, due to~\cite[Proposition B.3]{MotwaniR95},
\begin{equation*}
\left(1-q_{\max}\cdot l\right)^m
\ge e^{-q_{\max}\cdot l\cdot m}\cdot\left(1-q_{\max}^2\cdot l^2 \cdot m\right)
\ge e^{-q_{\max}\cdot l\cdot m}\cdot\left(1-q_{\max}^2\cdot 2^{2k}\cdot m\right).
\end{equation*}
In total, we get
\begin{align*}
\Pr{\Phi\text{ unsat}}&\\
\ge&\sum_{l=0}^{2^k}\Bigg(\binom{2^k}{l}\cdot(-1)^{l}\cdot e^{-q_{\max}\cdot l\cdot m}-\binom{2^k}{l}\cdot q_{\max}^2\cdot 2^{2k}\cdot m \cdot e^{-q_{\max}\cdot l\cdot m}\Bigg)\\
=&\left(1-e^{-q_{\max}\cdot m}\right)^{2^k}-q_{\max}^2\cdot 2^{2k}\cdot m \cdot\left(1+e^{-q_{\max}\cdot m}\right)^{2^k}.
\end{align*}
\end{proof}

Note that the former lemma implies the statement we want only if $q_{\max}\in o(1)$. 
Since 
\[q_{\max}=\frac{\prod_{i=1}^k p_i}{2^k\sum_{J\in{\mathcal{P}_k\left(\left\{1,2,\ldots,n\right\}\right)}}{\prod_{j\in{J}}{p_{j}}}},\]
it is also a function in $n$.
For $k=2$ the expression simplifies to $(p_1\cdot p_2)/(2\cdot(1-\sum_{i=1}^n p_i^2))$.
More generally than $q_{\max}\in o(1)$, we will now assume that we can choose an $\eps_q\in(0,1/2^k)$ so that $q_{\max}\le \eps_q$.
We will handle the case $q_{\max}\notin o(1)$ afterward.
The former lemma now yields the following corollary.

\begin{corollary} \label{cor:coarse1}
Let $\left(\vec{p}^{(n)}\right)_{n\in\N}$ be an ensemble of probability distributions.
Let $\Phi\sim\mathcal{D}^N(n,k,\left(\vec{p}^{(n)}\right)_{n\in\N},m)$ be a non-uniform random $k$-SAT formula.
Then,
\begin{enumerate}
\item for any $\eps_P\in(0,\left(1-e^{-\eps_m}\right)^{2^k})$ and for any $\eps_m>0$ so that $m=\eps_m/q_{\max}$ we can choose $\eps_q\in(0,1/2^k)$ with $q_{\max}\le\eps_q$ sufficiently small so that $\Pr{\Phi\text{ unsatisfiable}}\ge \eps_P$.
\item for any $\eps_P\in(0,1)$, we can choose $\eps_m>0$ with $m=\eps_m/q_{\max}$ sufficiently large and $\eps_q\in(0,1/2^k)$ with $q_{\max}\le\eps_q$ sufficiently small so that $\Pr{\Phi\text{ unsatisfiable}}\ge\eps_P$.
\end{enumerate}
\end{corollary}
\begin{proof}
Let $m^\star=q_{\max}^{-1}$ and fix a constant $\eps_m>0$ so that $m=\eps_m\cdot m^\star$ and a constant $\eps_P\in(0,\left(1-e^{-\eps_m}\right)^{2^k})$.
\lemref{coarseness1} tells us that 
\begin{align*}
\Pr{\Phi\text{ unsatisfiable}}
&\ge\left(1-e^{-q_{\max}\cdot m}\right)^{2^k}-q_{\max}^2\cdot 2^{2k}\cdot m \cdot\left(1+e^{-q_{\max}\cdot m}\right)^{2^k}\\
&=\left(1-e^{-\eps_m}\right)^{2^k}-q_{\max}\cdot 2^{2k}\cdot \eps_m \cdot\left(1+e^{-\eps_m}\right)^{2^k}\\
&\ge\left(1-e^{-\eps_m}\right)^{2^k}-\eps_q\cdot 2^{2k}\cdot \eps_m \cdot\left(1+e^{-\eps_m}\right)^{2^k}.
\end{align*}
If we choose $\eps_q$ sufficiently small, we can reach any value $\eps_P<\left(1-e^{-\eps_m}\right)^{2^k}$ as desired.

Now we turn to the case that we are only given $\eps_P\in(0,1)$.
It still holds that 
\begin{align*}
\Pr{\Phi\text{ unsatisfiable}}&\ge\left(1-e^{-\eps_m}\right)^{2^k}-\eps_q\cdot 2^{2k}\cdot \eps_m \cdot\left(1+e^{-\eps_m}\right)^{2^k}.
\end{align*}
We can see that if we choose $\eps_m$ sufficiently large and $\eps_q$ sufficiently small, we can make this expression at least $\eps_P$.
\end{proof}

This lemma already captures the case $q_{\max}\in o(1)$.
Let us now assume that there is some $\eps_q\in(0,1/2^k)$ so that $q_{\max}\ge \eps_q$.
It then holds that $m^\star=q_{\max}^{-1}\le 1/\eps_q$.
Remember that $q_{\max}\le1/2^k$ also still holds.
This means, the threshold function is bounded by a constant.
It is easy to see that for $\Phi\sim\mathcal{D}\left(n, k, (\vec{p}^{(n)})_{n\in\N}, m\right)$ and a constant $m\ge2^k$ it holds that $\Pr{\Phi\text{ unsatisfiable}}\ge q_{\max}^m\ge\eps_q^m$, since this is the probability of an unsatisfiable instance, where the most probable clause appears with all $2^k$ combinations of signs and then one of these clauses appears an additional $m-2^k$ times.
Similarly, $\Pr{\Phi\text{ satisfiable}}\ge q_{\max}^m\ge\eps_q^m$, as this is the probability of a satisfiable instance, where the same most probable clause appears $m$ times with the same sign.
Since $0<q_{\max}\le 1/2^k$ is a constant, the probability is a constant bounded away from zero and one.

It remains to show that $\Phi$ is unsatisfiable with probability $1-o(1)$ for $m\in \omega(1)$. 
More generally, we want to show that for any $\eps_P\in(0,1)$ we can choose an $\eps_m>0$ with $m=\eps_m\cdot m^\star$ large enough so that $\Phi$ is unsatisfiable with probability at least $\eps_P$.
The following lemma implies this.
Again, this lemma also holds for $k\ge2$ in general and without assuming anything for $q_{\max}$.
\begin{lemma}\label{lem:const-clause}
Consider a non-uniform random k-SAT formula $\Phi$.
Then $\Phi$ is unsatisfiable with probability at least
\[2-\left(1+\exp\left(-q_{\max}\cdot m\right)\right)^{2^k}.\]
\end{lemma}
\begin{proof}
As in \lemref{coarseness1}, it holds that
\begin{align*}
\Pr{\Phi\text{ unsat}}
&\ge \sum_{l=0}^{2^k}\left(\binom{2^k}{l}(-1)^{l}\left(1-l\cdot q_{\max}\right)^m\right).
\end{align*}
We can now estimate
\begin{align*}
\sum_{l=0}^{2^k}\left(\binom{2^k}{l}(-1)^{l}\left(1-l\cdot q_{\max}\right)^m\right)
&\ge 1-\sum_{l=1}^{2^k}\left(\binom{2^k}{l}\left(1-l\cdot q_{\max}\right)^m\right)\\
&\ge 1-\sum_{l=1}^{2^k}\left(\binom{2^k}{l}\exp\left(-m\cdot l\cdot q_{\max}\right)\right)\\
&= 2 -\left(1+\exp\left(-m\cdot q_{\max}\right)\right)^{2^k}
\end{align*}
\end{proof}

We can now see that our desired statement holds. 
The former implies it if we can choose $\eps_m$ large enough.
\begin{corollary}\label{cor:coarse2}
Let $\left(\vec{p}^{(n)}\right)_{n\in\N}$ be an ensemble of probability distributions.
For any constant $\eps_P\in(0,1)$ we can choose a constant $\eps_m>0$ with $m=\eps_m/q_{\max}$ sufficiently large so that the probability to generate an unsatisfiable formula $\Phi\sim\mathcal{D}^N(n,k,\left(\vec{p}^{(n)}\right)_{n\in\N},m)$ is at least $\eps_P$.
\end{corollary}
\begin{proof}
\lemref{const-clause} tells us
\[\Pr{\Phi\text{ unsatisfiable}}\ge 2 -\left(1+\exp\left(-m\cdot q_{\max}\right)\right)^{2^k}=2 -\left(1+\exp\left(-\eps_m\right)\right)^{2^k},\]
since $m\cdot q_{\max}=\eps_m$.
We can now simply make $\eps_m$ large enough to make this expression at least $\eps_P$.
\end{proof}

\section{Putting it All Together}\label{sec:main-theorem}

In this section we put the upper and lower bounds of the previous sections together.
This will show our main result of this chapter, the existence and sharpness of a satisfiability threshold for non-uniform random $2$-SAT depending on the ensemble of probability distributions $(\vec{p}^{(n)})_{n\in\N}$.
We could have unified the proof to capture all cases with $p_1^2\notin o(\sum_{i=1}^n p_i^2)$.
However, the current setting has the advantage that we can also state reasons for the threshold being coarse in the second and third case.
In the second case it is due to the emergence of a snake of size $2$, \ie an unsatisfiable sub-formula that looks like this
\[\left(w_2,w_1\right),\left(\n{w_1},w_2\right),\left(\n{w_2},w_3\right),\left(\n{w_3},\n{w_2}\right)\]
for literals $w_1$, $w_2$, and $w_3$ of distinct Boolean variables.
In the third case, the coarseness comes from the emergence of an unsatisfiable sub-formula, where the clause with the two most probable variables appears with all four combinations of signs.

\begin{theorem} \label{thm:main1}
Given an ensemble of probability distributions $\left(\vec{p}^{(n)}\right)_{n\in\N}$.
\begin{enumerate}
\item If $p_1^2\in o(\sum_{i=1}^n p_i^2)$, then non-uniform random 2-SAT has a sharp satisfiability threshold at $m^\star=1/\sum_{i=1}^n p_i^2$.
\item If $p_1^2\in \Theta\left(\sum_{i=1}^n p_i^2\right)$ and $p_2^2\in o\left(\sum_{i=2}^n p_i^2\right)$, then non-uniform random 2-SAT has a coarse satisfiability threshold at $m^\star=(C\cdot p_1\cdot(\sum_{i=2}^n p_i^2)^{1/2})^{-1}$.
Furthermore, for any large enough $n$ there is a range of size $\Theta(m^\star)$ around the threshold, where the probability to generate satisfiable instances is bounded away from zero and one.
\item If $p_1^2\in \Theta\left(\sum_{i=1}^n p_i^2\right)$ and $p_2^2\in \Theta\left(\sum_{i=2}^n p_i^2\right)$, then non-uniform random 2-SAT has a coarse satisfiability threshold at $m^\star=(q_{\max}^{-1})\in \Theta((C\cdot p_1\cdot(\sum_{i=2}^n p_i^2)^{1/2})^{-1})$.
Furthermore, for any large enough $n$ there is a range of size $\Theta(m^\star)$ around the threshold, where the probability to generate satisfiable instances is bounded away from zero and one.
\item Otherwise, non-uniform random 2-SAT has a coarse satisfiability threshold at $m^\star=(C\cdot\sum_{i=2}^n p_i^2+ C\cdot p_1\cdot(\sum_{i=2}^n p_i^2)^{1/2})^{-1}$.
\end{enumerate}
\end{theorem}
\begin{proof}
Remember our discussion in \secref{what}.
We want to show that $m^\star$ is an asymptotic threshold function of non-uniform random $2$-SAT with respect to parameter $m$.
This means:
\begin{enumerate}
\item for any function $m\colon \N\to\R^+$ with $m\in o(m^\star)$ and any $\eps_P\in(0,1)$ there is an $n_0\in\N$ so that for all $n\ge n_0$ the probability to generate a satisfiable instance is at least $\eps_P$.
\item and for all $m\colon \N\to\R^+$ with $m\in \omega(m^\star)$ and any $\eps_P\in(0,1)$ there is an $n_0\in\N$ so that for all $n\ge n_0$ the probability to generate an unsatisfiable instance is at least $\eps_P$.
\end{enumerate}

If we want to show a sharp threshold, we have to certify that:
\begin{enumerate}
\item for any given constant $\eps_m\in(0,1)$, any function $m\colon \N\to\R^+$ with $m\le\eps_m\cdot m^\star$, and any $\eps_P\in(0,1)$ there is an $n_0\in\N$ so that for all $n\ge n_0$ the probability to generate a satisfiable instance is at least $\eps_P$.
\item and for any given constant $\eps_m>1$, all $m\colon \N\to\R^+$ with $m\ge \eps_m\cdot m^\star$, and any $\eps_P\in(0,1)$ there is an $n_0\in\N$ so that for all $n\ge n_0$ the probability to generate an unsatisfiable instance is at least $\eps_P$.
\end{enumerate}

\paragraph*{Case~1: $p_1^2\in o(\sum_{i=1}^n p_i^2)$}
The first case we consider is $p_1^2\in o(\sum_{i=1}^n p_i^2)$.
We want to show a sharp threshold at $m^\star=1/\sum_{i=1}^n p_i^2$.
The requirement $p_1^2\in o(\sum_{i=1}^n p_i^2)$ implies that we can choose any $\eps_1\in(0,1)$ and for some $n_0\in\N$ it holds that $p_1^2\le \eps_1\cdot \sum_{i=1}^n p_i^2$ for all $n\ge n_0$.
Thus, \lemref{sharp-lb} directly implies the first requirement for sharpness and \lemref{threshold-sharp} directly implies the second requirement.

\paragraph*{Case~2: $p_1^2\in\Theta\left(\sum_{i=1}^n p_i^2\right)$ and $p_2^2\in o\left(\sum_{i=2}^n p_i^2\right)$}
The second case we consider is $p_1^2\in\Theta\left(\sum_{i=1}^n p_i^2\right)$ and $p_2^2\in o\left(\sum_{i=2}^n p_i^2\right)$.
We want to show that $m^\star=(C\cdot p_1\cdot(\sum_{i=2}^n p_i^2)^{1/2})^{-1}$ is a coarse satisfiability threshold.
The requirements imply that there is some $\eps_1\in(0,1)$ and that we can choose an $\eps_2\in(0,1)$ so that $p_1^2\ge \eps_1\cdot \sum_{i=1}^n p_i^2$ and $p_2^2\le \eps_2\cdot\sum_{i=2}^n p_i^2$ hold simultaneously for all sufficiently large $n$.
If $m\in o(m^\star)$, then for any $\eps_m\in(0,1)$ there is an $n_0\in\N$ so that $m\le\eps_m\cdot m^\star$ for all $n\ge n_0$.
Thus, we can apply \lemref{bicycle-coarse} to certify the first condition on $m^\star$ being an asymptotic threshold function.
Equivalently, if $m\in \omega(m^\star)$, then for any $\eps_m>1$, there is an $n_0\in\N$ so that $m\ge\eps_m\cdot m^\star$ for all $n\ge n_0$.
We can now apply \corref{threshold-coarse2}.
Note that the lemma assumes $m=\eps_m\cdot m^\star$.
However, since the probability to generate satisfiable instances is non-increasing in $m$ (\cf \lemref{monotone-drawing}), it suffices to consider $m'=\eps_m \cdot m$.
The probability to generate satisfiable (unsatisfiable) instances at the actual number of clauses $m\ge m'$ can only be smaller (larger).
Thus, \corref{threshold-coarse2} implies the second condition on $m^\star$ being an asymptotic threshold function.

It remains to show that the threshold is not sharp.
Essentially, we are going to show that there is a non-empty range of $\eps_m\in[\eps_m^{(1)},\eps_m^{(2)}]$ for which the probability to generate satisfiable instances at $m=\eps_m\cdot m^\star$ is bounded away from zero and one by a constant.
If the threshold was sharp, at least one of the probabilities at positions $m^{(1)}$ and $m^{(2)}$ that are a constant factor apart would approach zero or one in the limit.
Since in our case neither the probability at $m^{(1)}=\eps_m^{(1)}\cdot m^\star$ nor the one at $m^{(2)}=\eps_m^{(2)}\cdot m^\star$ does, the threshold must be coarse.
First, \lemref{bicycle-coarse} states that for any $\eps_P\in(0,1)$ we can choose $\eps_m\in(0,1)$ small enough so that the probability to generate a satisfiable instance at $m=\eps_{m}\cdot m^\star$ is at least $\eps_P$.
We can now choose $\eps_P^{(1)}>\eps_P^{(2)}$.
This will result in some $\eps_m^{(1)}<\eps_m^{(2)}$ so that the probability to generate a satisfiable instance is at least $\eps_P^{(1)}$ at $m=\eps_m^{(1)}\cdot m^\star$ and $\eps_P^{(2)}$ at $m=\eps_m^{(2)}\cdot m^\star$.
However, \lemref{coarse-threshold1} states that for the same values of $\eps_m$ we can choose $\eps_2$ with $p^2 \le \eps_2\cdot\sum_{i=2}^n p_i^2$ small enough so that the probability to generate an unsatisfiable instance at $m=\eps_m\cdot m^\star$ is at least $\eps_P$ for any constant 
\[\eps_P<\frac{\eps_m^4}{\eps_m^4+3\cdot\eps_m^2\left(1+\frac{1}{\eps_1}+\frac{1}{\eps_1^2}\right)+8}.\]
This requirement on $\eps_2$ holds for all sufficiently large $n$, since $p_2^2\in o(\sum_{i=2}^n p_i^2)$.
Thus, for $\eps_m^{(1)}$ and $\eps_m^{(2)}$ both the probability to generate a satisfiable and the probability to generate an unsatisfiable instance are at least some constant depending only on $\eps_1$ and $\eps_m$ if $n$ is large enough.
Since both $\eps_m$ and $\eps_1$ are fixed, these probabilities cannot approach zero or one in the limit.
This implies coarseness of the threshold as desired.

\paragraph*{Case~3: $p_1^2\in\Theta\left(\sum_{i=1}^n p_i^2\right)$ and $p_2^2\in\Theta\left(\sum_{i=2}^n p_i^2\right)$}
The third case we consider is $p_1^2\in\Theta\left(\sum_{i=1}^n p_i^2\right)$ and $p_2^2\in \Theta\left(\sum_{i=2}^n p_i^2\right)$.
We want to show a coarse satisfiability threshold at $m^\star=q_{\max}^{-1}$, where $q_{\max}=(C\cdot p_1\cdot p_2)/2$ is the maximum clause probability.
Note that in this case, $m^\star\in\Theta((C\cdot p_1\cdot(\sum_{i=2}^n p_i^2)^{1/2})^{-1})$.
As before, we can apply \lemref{bicycle-coarse} to certify the first condition on $m^\star$ being an asymptotic threshold function. 
The second condition is implied by our results in \secref{coarse}.
The second statement of \corref{coarse1} certifies the second condition if $\eps_q\in(0,1)$ with $q_{\max}\le\eps_q$ is sufficiently small and $\eps_m>0$ with $m=\eps_m/q_{\max}$ is sufficiently large.
If $q_{\max}\in o(1)$ and $m\in\omega(m^\star)=\omega(q_{\max}^{-1})$ both conditions hold for all sufficiently large values of $n$.
If $q_{\max}\notin o(1)$, the second condition holds as follows.
According to the second condition we are given an $m\in\omega(m^\star)$ and an $\eps_P\in(0,1)$.
We choose $\eps_m$ sufficiently large and $\eps_q$ sufficiently small so that we generate an unsatisfiable instance with probability at least $\eps_P$ according to \corref{coarse1}.
Then, we fix that value of $\eps_q$ and choose an $\eps_m$ sufficiently large so that we generate an unsatisfiable instance with probability at least $\eps_P$ according to \corref{coarse2}.
Since $m\in\omega(m^\star)$, we know that $m\ge\eps_m / q_{\max}$ holds for both values of $\eps_m$ we chose as soon as $n$ is sufficiently large. 
For all such values of $n$ we either have $q_{\max}\le \eps_q$ or $q_{\max}>\eps_q$.
Thus, the second condition holds either according to \corref{coarse1} or according to \corref{coarse2}.

As in the previous case we have to rule out that the threshold is sharp.
Again, we will show that there is a range of $m\in\Theta(m^\star)$ where the probability to generate satisfiable instances is bounded away from zero and one by constants.
However, depending on whether or not $q_{\max}\in o(1)$, this range can be at different positions in $\Theta(m^\star)$.
This is due to the fact that, if $q_{\max}\in\Omega(1)$, then $m^\star\in\Oh(1)$.
However, in order to have an unsatisfiable instance we need $m\ge4$.
At the same time \lemref{bicycle-coarse} might require us to choose an $\eps_m$ so small that this is not guaranteed anymore.
Thus, in the case that $q_{\max}\in\Omega(1)$ we choose a different range of $\eps_m$ with $m=\eps_m \cdot m^\star$.

We start with $q_{\max}\in o(1)$.
Now, note that $m^\star=q_{\max}^{-1}\in\Theta((C\cdot p_1\cdot(\sum_{i=2}^n p_i^2)^{1/2})^{-1})$.
Thus, there are constants ${\eps_l}^{(1)},{\eps_l}^{(2)}>0$ such that ${\eps_l}^{(1)}\cdot (C\cdot p_1\cdot(\sum_{i=2}^n p_i^2)^{1/2})^{-1}\le m^\star \le {\eps_l}^{(2)}\cdot (C\cdot p_1\cdot(\sum_{i=2}^n p_i^2)^{1/2})^{-1}$ for all sufficiently large values of $n$.
We now choose $m^{(1)}={\eps_m}^{(1)}\cdot m^\star$ and $m^{(2)}={\eps_m}^{(2)}\cdot m^\star$ with ${\eps_m}^{(1)}<{\eps_m}^{(2)}$.
It holds that $m^{(1)}\le{\eps_m}^{(1)}\cdot{\eps_l}^{(2)}\cdot (C\cdot p_1\cdot(\sum_{i=2}^n p_i^2)^{1/2})^{-1}$ and $m^{(2)}\le{\eps_m}^{(2)}\cdot{\eps_l}^{(2)}\cdot (C\cdot p_1\cdot(\sum_{i=2}^n p_i^2)^{1/2})^{-1}$.
Since \lemref{bicycle-coarse} only requires $p_1^2\in\Theta\left(\sum_{i=1}^n p_i^2\right)$, we can now use it equivalently to the second case.
That means, if we choose the constants ${\eps_m}^{(1)}$ and ${\eps_m}^{(2)}$ small enough, the probability to generate satisfiable instances at both number of clauses is at least a constant depending only on $\eps_1$, ${\eps_l}^{(2)}$ and $\eps_m$, all of which are constant for sufficiently large $n$.
%
%
For the same values of $m$ we want to have a constant lower bound on the probability to generate unsatisfiable instances.
Again, our results from \secref{coarse} provide us with these lower bounds.
According to \corref{coarse1} it holds for both $m^{(1)}$ and $m^{(2)}$ that the probability to generate unsatisfiable instances can be lower bounded by a constant that only depends on $\eps_m$ as soon as $\eps_q\in(0,1/2^k)$ with $q\le\eps_q$ is small enough.
This holds for all sufficiently large $n$, since we assumed $q_{\max}\in o(1)$.
Since $\eps_m^{(1)}$ and $\eps_m^{(2)}$ are fixed constants, the resulting probability is constant as well.
This gives us the desired result if $q_{\max}\in o(1)$.

Now we consider $q_{\max}\in \Omega(1)$.
It holds that $m^\star=1/q_{\max}\in\Oh(1)$.
Then, we can simply choose any two constants $\eps_m^{(1)},\eps_m^{(2)}>1$ that are sufficiently far apart for $m^{(1)}=\eps_m^{(1)}\cdot m^\star$ and $m^{(2)}=\eps_m^{(1)}\cdot m^\star$ to be different integers.
Both the probability to generate a satisfiable and an unsatisfiable instance are at least $q_{\max}^m$.
Since $q_{\max}\in\Omega(1)$, $q_{\max}$ is lower-bounded by a constant for all sufficiently large $n$.
The same holds for $m=\eps_m \cdot m^\star$, since $q_{\max}\le 1/2^k$ and $\eps_m$ is some fixed constant as well.
Thus, the probabilities to generate satisfiable and unsatisfiable instances at $m^{(1)}$ and $m^{(2)}$ are bounded away from zero and one as desired.

Last, we consider $q_{\max}\notin o(1)$ and $q_{\max}\notin \Omega(1)$.
First, we choose ${\eps_m}^{(1)},{\eps_m}^{(2)}>0$ as before and $\eps_q$ small enough so that the same bounds hold as in the case of $q_{\max}\in o(1)$.
Then, we assume $q_{\max}\ge \eps_q$ and choose $\eps_m^{(1)},\eps_m^{(2)}>1$ as in the case of $q_{\max}\in\Omega(1)$.
This implies probabilities of at least $\eps_q^{\eps_m/\eps_q}$ to generate a satisfiable/unsatisfiable instance.
For all sufficiently large $n$ we either have $q_{\max}\le \eps_q$ or $q_{\max}>\eps_q$.
Thus, the threshold is coarse either way.

\paragraph*{Case 4: Otherwise}
The last case we consider is that none of the three other cases hold.
We are going to show that there is a coarse threshold at $m^\star=(C\cdot\sum_{i=2}^n p_i^2+ C\cdot p_1\cdot(\sum_{i=2}^n p_i^2)^{1/2})^{-1}$.
The threshold function is chosen such that, depending on $p_1^2$ and $p_2^2$, either the first or the second term dominates.
That means, if $\eps_1\in(0,1)$ with $p_1^2\le\eps_1\sum_{i=1}^n p_i^2$ is small enough, then $m^\star\in\Theta((C\cdot\sum_{i=1}^n p_i^2)^{-1})$.
In that case, we have an asymptotic threshold as if $p_1^2\in o(\sum_{i=1}^n p_i^2)$.
Otherwise, $m^\star\in\Theta( (C\cdot p_1\cdot(\sum_{i=2}^n p_i^2)^{1/2})^{-1})$.
Then, we have an asymptotic threshold as if $p_1^2\in \Theta(\sum_{i=1}^n p_i^2)$.
To make things easier, let us investigate $m\in o(m^\star)$ and $m\in\omega(m^\star)$ separately.

Let us start with $m\in o(m^\star)$.
We are given an $\eps_P\in(0,1)$ and have to assure that the probability to generate a satisfiable instance at $m$ is at least $\eps_P$.
Thus, we first choose some $\eps_m\in(0,1)$ with $m=\eps_m\cdot m^\star$.
Furthermore, we assume that we can choose an $\eps_1$ with $p_1^2\le\eps_1\cdot\sum_{i=1}^n p_i^2$.
We now want to apply \lemref{sharp-lb}.
However, the lemma is stated with respect to the threshold function $(\sum_{i=1}^n p_i^2)^{-1}$.
Thus, we first have to relate our $m^\star$ to this function, assuming that we can choose $\eps_m$ and $\eps_1$ arbitrarily small.
It holds that 
\[\sum_{i=1}^n p_i^2=p_1^2+\sum_{i=2}^n p_i^2\le\eps_1\cdot\sum_{i=1}^n p_i^2+\sum_{i=2}^n p_i^2.\]
Thus, $\sum_{i=2}^n p_i^2\ge(1-\eps_1)\cdot \sum_{i=1}^n p_i^2$ and therefore 
\[m^\star=\left(C\cdot\sum_{i=2}^n p_i^2+ C\cdot p_1\cdot\left(\sum_{i=2}^n p_i^2\right)^{1/2}\right)^{-1}\le \frac{1}{1-\eps_1}\cdot \left(C\cdot \sum_{i=1}^n p_i^2\right)^{-1}.\]
Note that $C=1/(1-\sum_{i=1}^n p_i^2)\ge1$.
For any fixed $\eps_1$ this especially means $m\in o(m^\star)$ implies $m\in o(\left(\sum_{i=1}^n p_i^2\right)^{-1})$.
This allows us to apply \lemref{sharp-lb} with $\eps_m\in(0,1)$ and an $\eps_1\in(0,1)$ small enough to give us a probability of at least $\eps_P$.
The requirement $m\in o((\sum_{i=1}^n p_i^2)^{-1})$ guarantees that the condition on $\eps_m$ is fulfilled.
However, we can not guarantee that $p_1^2\le\eps_1\cdot\sum_{i=1}^n p_i^2$ holds.
Thus, we now assume $p_1^2\ge\eps_1\cdot\sum_{i=1}^n p_i^2$.
This is what we also assumed in the case $p_1^2\in\Theta(\sum_{i=1}^n p_i^2)$ and in fact, we can use the same results now.
That is, we can use \lemref{bicycle-coarse}.
Again, we have to relate $m^\star$ to the threshold function $(C\cdot p_1\cdot(\sum_{i=2}^n p_i^2)^{1/2})^{-1}$ the lemma uses.
However, we can easily see that $m^\star\le (C\cdot p_1\cdot(\sum_{i=2}^n p_i^2)^{1/2})^{-1}$.
Thus, any function $m\in o(m^\star)$ is also in $o((C\cdot p_1\cdot(\sum_{i=2}^n p_i^2)^{1/2})^{-1})$.
\lemref{bicycle-coarse} states that for the value of $\eps_1\in(0,1)$ we have chosen before and the given value $\eps_P$, we can now choose $\eps_m\in(0,1)$ with $m\le\eps_m\cdot m^\star\le\eps_m\cdot (C\cdot p_1\cdot(\sum_{i=2}^n p_i^2)^{1/2})^{-1}$ sufficiently small so that the probability to generate a satisfiable instance is at least $\eps_P$.
Thus, for all large enough values of $n$, both $m\le \eps_m\cdot \left(\sum_{i=1}^n p_i^2\right)^{-1}$ and $m\le \eps_m\cdot (C\cdot p_1\cdot(\sum_{i=2}^n p_i^2)^{1/2})^{-1}$ hold.
Then, the probability of at least $\eps_P$ at $m$ is guaranteed either by \lemref{sharp-lb} if $p_1^2\le\eps_1\cdot\sum_{i=1}^n p_i^2$ or by \lemref{bicycle-coarse} if $p_1^2>\eps_1\cdot\sum_{i=1}^n p_i^2$.

Let us now turn to $m\in\omega(m^\star)$.
Given an $\eps_P\in(0,1)$ we want to show that the probability to generate an unsatisfiable instance is at least $\eps_P$ at $m$.
Again, we assume that we can choose an $\eps_1$ with $p_1^2\le\eps_1\cdot\sum_{i=1}^n p_i^2$.
Then, we can apply \lemref{threshold-sharp}.
However, we first have to compare $m^\star$ to $(\sum_{i=1}^n p_i^2)^{-1}$ again.
First, it holds that $p_1^2\le \eps_1\cdot \sum_{i=1}^n p_i^2\le \eps_1\cdot p_1$ and thus $p_1\le\eps_1$.
This implies $C=1/(1-\sum_{i=1}^n p_i^2)\le 1/(1-\eps_1)$.
It also implies
\[m^\star=\left(C\cdot\sum_{i=2}^n p_i^2+ C\cdot p_1\cdot\left(\sum_{i=2}^n p_i^2\right)^{1/2}\right)^{-1}\ge \left(2\cdot C\cdot \sum_{i=1}^n p_i^2\right)^{-1}\ge\frac{1-\eps_1}{2\cdot\sum_{i=1}^n p_i^2},\]
since $p_1\le(\sum_{i=1}^n p_i^2)^{1/2}$ and $\sum_{i=2}^n p_i^2\le\sum_{i=1}^n p_i^2$. 
This means, $m\in\omega(m^\star)$ implies $m\in\omega(1/\sum_{i=1}^n p_i^2)$.
We can now choose some $\eps_m>1$ with $m\ge \eps_m/\sum_{i=1}^n p_i^2$ and apply \lemref{threshold-sharp} to show that the probability to generate an unsatisfiable instance at $\eps_m/\sum_{i=1}^n p_i^2$ is at least $\eps_P$ if $\eps_1\in(0,1)$ with $p_1^2\le\eps_1\cdot\sum_{i=1}^n p_i^2$ is small enough.
Note that this probability only holds at $\eps_m/\sum_{i=1}^n p_i^2$.
However, due to the monotonicity of the probability function in our model (\cf \lemref{monotone-drawing}), it also holds for all $m\ge\eps_m/\sum_{i=1}^n p_i^2$.
Since $m\in\omega(1/\sum_{i=1}^n p_i^2)$, $m\ge\eps_m/\sum_{i=1}^n p_i^2$ holds for all sufficiently large values of $n$.
Up to this point we assumed $p_1^2\le\eps_1\cdot\sum_{i=1}^n p_i^2$ for the value $\eps_1$ we needed in \lemref{threshold-sharp}.
Now we assume $p_1^2>\eps_1\cdot\sum_{i=1}^n p_i^2$ for that same value $\eps_1$.
However, we have to make another distinction depending on $p_2^2$.
First, we assume $p_2^2\le \eps_2\cdot\sum_{i=2}^n p_i^2$ for some $\eps_2\in(0,1)$ of our choice.
We want to use \corref{threshold-coarse2} to show the bound we need.
Again, we have to show that $m^\star=(C\cdot\sum_{i=2}^n p_i^2+ C\cdot p_1\cdot(\sum_{i=2}^n p_i^2)^{1/2})^{-1}$ is large enough compared to $(C\cdot p_1\cdot(\sum_{i=2}^n p_i^2)^{1/2})^{-1}$.
It holds that $\sum_{i=2}^n p_i^2\le\sum_{i=1}^n p_i^2\le p_1^2/\eps_1$.
Thus, $C\cdot\sum_{i=2}^n p_i^2\le C\cdot p_1\cdot(\sum_{i=2}^n p_i^2)^{1/2}/\sqrt{\eps_1}$ and 
\[m^\star\ge \frac{1}{1+1/\sqrt{\eps_1}}\cdot \left(C\cdot p_1\cdot\left(\sum_{i=2}^n p_i^2\right)^{1/2}\right)^{-1}.\]
Thus, for our fixed value $\eps_1$ it holds that $m\in\omega(m^\star)$ implies $m\in\omega((C\cdot p_1\cdot(\sum_{i=2}^n p_i^2)^{1/2})^{-1})$.
We can now apply \corref{threshold-coarse2} for some sufficiently large $\eps_m>0$ and some sufficiently small $\eps_2\in(0,1)$ to have a probability of at least $\eps_P$ for generating an unsatisfiable instance.
As with $p_1^2$, we now assume the contrary for $p_2^2$, \ie $p_2^2>\eps_2\cdot\sum_{i=2}^n p_i^2$ for the value $\eps_2$ we just chose.
We want to use \corref{coarse2} to show a probability of at least $\eps_P$ for generating an unsatisfiable instance.
The lemma holds if we have $m\ge\eps_m/q_{\max}$ for an $\eps_m>0$ large enough.
Under our assumptions $p_1^2>\eps_1\cdot\sum_{i=1}^n p_i^2$ and $p_2^2>\eps_2\cdot\sum_{i=2}^n p_i^2$ it holds that
\[m^\star=\left(C\cdot\sum_{i=2}^n p_i^2+ C\cdot p_1\cdot\left(\sum_{i=2}^n p_i^2\right)^{1/2}\right)^{-1}\ge \left(\left(\frac{1}{\sqrt{\eps_1\cdot \eps_2}}+\frac{1}{\sqrt{\eps_2}}\right)\cdot 2\cdot q_{\max}\right)^{-1},\]
because $\sum_{i=2}^n p_i^2\le(\sum_{i=1}^n p_i^2)^{1/2}\cdot(\sum_{i=2}^n p_i^2)^{1/2}$ and $q_{\max}=C\cdot p_1\cdot p_2/2$.
Therefore, $m\in\omega(m^\star)$ implies $m\in\omega(q_{\max}^{-1})$ in this case.
Thus, for any $\eps_m>0$ it holds that $m\ge\eps_m/q_{\max}$ for all sufficiently large $n$ and \corref{coarse2} gives us a probability of at least $\eps_P$ as desired.
Note that, depending on the lemma or corollary we used, we made different choices for $\eps_m$.
However, all these choices are satisfied for all sufficiently large $n$, since $m$ always grows asymptotically faster than the respective threshold function in all three cases.
From this point, either \lemref{threshold-sharp}, or \corref{threshold-coarse2}, or \corref{coarse2} guarantees that the probability to generate an unsatisfiable instance is at least $\eps_P$.

It remains to show that the threshold is not sharp in the last case.
If none of the first three cases hold, then either 
\begin{enumerate}
\item $p_1^2\notin\Theta(\sum_{i=1}^n p_i^2)$ and $p_1^2\notin o(\sum_{i=1}^n p_i^2)$ or 
\item $p_1^2\in\Theta(\sum_{i=1}^n p_i^2)$, but $p_2^2\notin\Theta(\sum_{i=2}^n p_i^2)$ and $p_2^2\notin o(\sum_{i=2}^n p_i^2)$.
\end{enumerate}
If $p_1^2\notin o(\sum_{i=1}^n p_i^2)$, then there is a constant $\eps_1$ so that for any $n_0\in\N$ there must be an $n\ge n_0$ such that $p_1^2\ge \eps_1\cdot \sum_{i=1}^n p_i^2$.
We now consider only the values of $n$, where $p_1^2\ge \eps_1\cdot \sum_{i=1}^n p_i^2$ holds.
Essentially, we treat this as an ensemble with $p_1^2\in\Theta(\sum_{i=1}^n p_i^2)$.
Now we either have $p_2^2\in o(\sum_{i=2}^n p_i^2)$, or $p_2^2\in\Theta(\sum_{i=2}^n p_i^2)$, or neither of the two.
For $p_2^2\in o(\sum_{i=2}^n p_i^2)$ we know from case~2 that there is a range of $m$ of size $\Theta((C\cdot p_1\cdot(\sum_{i=2}^n p_i^2)^{1/2})^{-1})$, where the probability function approaches neither zero nor one.
For $p_2^2\in\Theta(\sum_{i=2}^n p_i^2)$ the same holds for $\Theta(q_{\max}^{-1})$ due to case~3.
From proving that $m^\star$ is an asymptotic threshold function we know that, depending on $p_1^2$ and $p_2^2$, $m^\star\in\Theta((C\cdot p_1\cdot(\sum_{i=2}^n p_i^2)^{1/2})^{-1})$ or $m^\star\in\Theta(q_{\max}^{-1})$, respectively.
Thus for all sufficiently large values of $n$ we consider, there are ranges of $m\in\Theta(m^\star)$, where the probability function approaches neither zero nor one.
If we considered all values of $n$ now, the ones we selected before prevent our probability function from approaching zero or one in the chosen ranges.
Thus, the threshold cannot be sharp.
If $p_1^2\in\Theta(\sum_{i=1}^n p_i^2)$, or if we have infinitely many values of $n$ where this holds as for $p_1^2\notin o(\sum_{i=1}^n p_i^2)$, and $p_2^2\notin o(\sum_{i=2}^n p_i^2)$ a similar argumentation holds for the chosen values of $n$ with $p_2^2\ge \eps_2\cdot \sum_{i=2}^n p_i^2$, \ie we can assume $p_1^2\in\Theta(\sum_{i=1}^n p_i^2)$ and $p_2^2\in\Theta(\sum_{i=2}^n p_i^2)$ for those values.
\end{proof}

\section{Example Applications of our Theorem} \label{sec:2-SAT-examples}

We now apply \thmref{main1} to determine the satisfiability threshold behavior of non-uniform random 2-SAT with different ensembles of probability distributions.

\subsection{Random 2-SAT}

For random $2$-SAT the probability distribution for $n\in\N$ is $\vec{p}^{(n)}=\left(\frac1n,\frac1n,\ldots,\frac1n\right)$.
This means $p_1^2 = \frac{1}{n^2}$ and $\sum_{i=1}^n p_i^2=\frac1n$.
We see that $p_1^2\in o(\sum_{i=1}^n p_i^2)$.
The first case of our theorem now tells us that there is a sharp threshold at $m^\star=(\sum_{i=1}^n p_i^2)^{-1}=n$.
This is exactly what \citet{chvatalreed92} found out as well.

\begin{figure}[t!]
  \centering
  \includegraphics[width=0.9\textwidth]{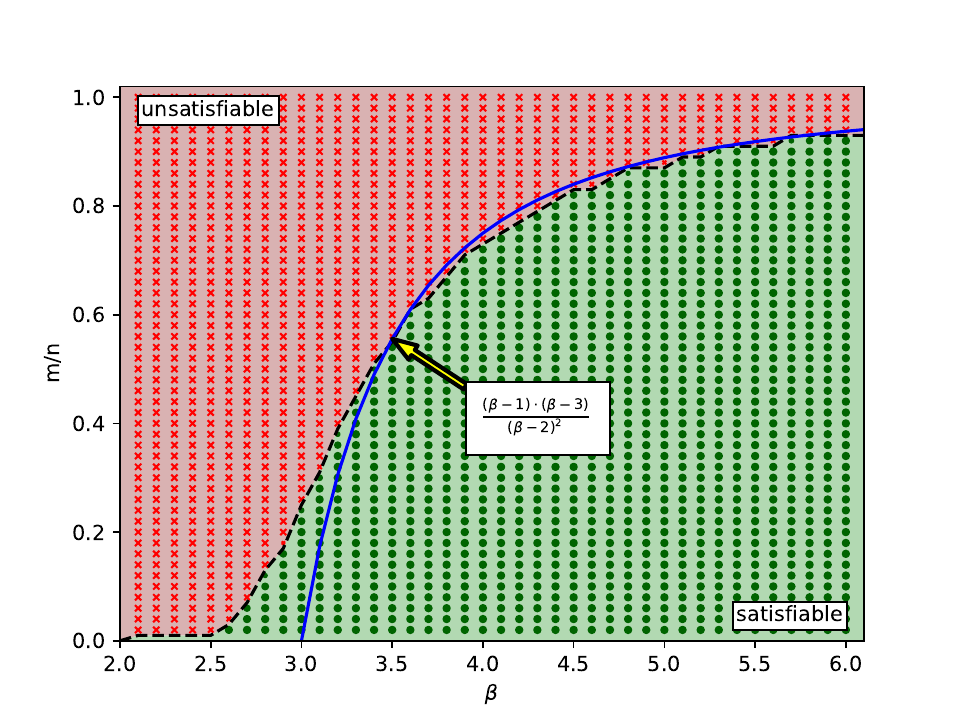}
  \caption{\label{fig:phase1} Phase diagram for power-law random 2-SAT formulas with $n=10^7$ variables.
	Each point is a sample of 100 random instances at the given parameter combination.
We drew a red cross if all instances were unsatisfiable and a green dot if at least one instance was satisfiable with the size of the dot scaling with the fraction of satisfiable instances.
We empirically observe a sharp phase transition
  (\protect\tikz[overlay]{\protect\path[draw=black,line width= 1.6pt,dash pattern=on 4pt off 4pt,line cap=round]
(0.04,0.1) -- (0.5,0.1);}\hspace*{5mm}), which closely matches
the theoretical bound of \thmref{main1}
(\protect\tikz[overlay]{\protect\path[draw=blue,line width= 1.6pt]
(0.04,0.1) -- (0.46,0.1);}\hspace*{5mm}).}\label{fig:pl-2-sat}
\end{figure}

\subsection{Power-law Random 2-SAT} \label{sec:pl2sat}

\thmref{main1} implies the following corollary.
\begin{corollary}\label{cor:pl2sat}
For power-law random 2-SAT, if
\begin{itemize}
\item $\beta<3$, then the threshold is coarse at $m^\star\in\Theta\left(q_{\max}^{-1}\right)\in\Theta\left(n^{2(\beta-2)/(\beta-1)})\right)$.
\item $\beta=3$, then the threshold is sharp at $m^\star=4\cdot\frac{n}{\ln n}$.
\item $\beta>3$, then the threshold is sharp at $m^\star=\frac{(\beta-1)\cdot(\beta-3)}{(\beta-2)^2}\cdot n$.
\end{itemize}
\end{corollary}
\begin{proof}
For power-law random $2$-SAT we assume some fixed $\beta>2$.
Then for $n\in\N$ the distribution is $\vec{p}^{(n)}=\left(p_1^{(n)},p_2^{(n)}\ldots,p_n^{(n)}\right)$ with
\[p_i^{(n)}=\frac{(n/i)^{\frac{1}{\beta-1}}}{\sum_{j=1}^n (n/j)^{\frac{1}{\beta-1}}}.\]
It holds that $p_1\ge p_2 \ge \ldots\ge p_n$.
\lemref{pl-aux1} tells us that
\begin{eqnarray*}
p_1^2 &=& (1\pm o(1))\cdot \left(\frac{\beta-2}{\beta-1}\right)^2\cdot n^{-2\frac{\beta-2}{\beta-1}},\\
p_2^2 &=& (1\pm o(1))\cdot \left(\frac{\beta-2}{\beta-1}\right)^2\cdot2^{-\frac{1}{\beta-1}}\cdot n^{-2\frac{\beta-2}{\beta-1}},\text{ and}\\
\sum_{i=1}^n p_i^2 &=& \begin{cases}\Theta\left(n^{-2\frac{\beta-2}{\beta-1}}\right)&\text{for }\beta<3\\
\left(1\pm o(1)\right)\cdot\frac14\cdot\frac{\ln n}{n}&\text{for }\beta=3\\
\left(1\pm o(1)\right)\cdot\frac{(\beta-2)^2}{(\beta-3)\cdot(\beta-1)}\cdot n^{-1} & \text{for }\beta>3.\end{cases}
\end{eqnarray*}

For $\beta<3$ it holds that $p_1^2\in\Theta\left(\sum_{i=1}^n p_i^2\right)$ and $p_2^2\in\Theta(\sum_{i=2}^n p_i^2)$.
Thus, there is a coarse threshold at $m^\star=q_{\max}^{-1}\in\Theta(n^{2(\beta-2)/(\beta-1)})$, since $q_{\max}=C\cdot p_1\cdot p_2/2$, $p_1,p_2\in\Theta\left(n^{-(\beta-2)/(\beta-1)}\right)$, and $C=1/(1-\sum_{i=1}^n p_i^2)=1+o(1)$.

For $\beta=3$ it holds that $p_1^2\in o(\sum_{i=1}^n p_i^2)$.
Thus, there is a sharp satisfiability threshold at $m^\star=4\cdot\frac{n}{\ln n}$.

For $\beta>3$ it also holds that $p_1^2\in o(\sum_{i=1}^n p_i^2)$.
Thus, there is a sharp satisfiability threshold at $m^\star=\frac{(\beta-1)\cdot(\beta-3)}{(\beta-2)^2}\cdot n$.
\end{proof}

\figref{pl-2-sat} visualizes the empirical threshold position compared to the theoretical position according to \corref{pl2sat}.

\subsection{Geometric Random 2-SAT}

\begin{figure}[t!]
  \centering
  \includegraphics[width=0.95\textwidth]{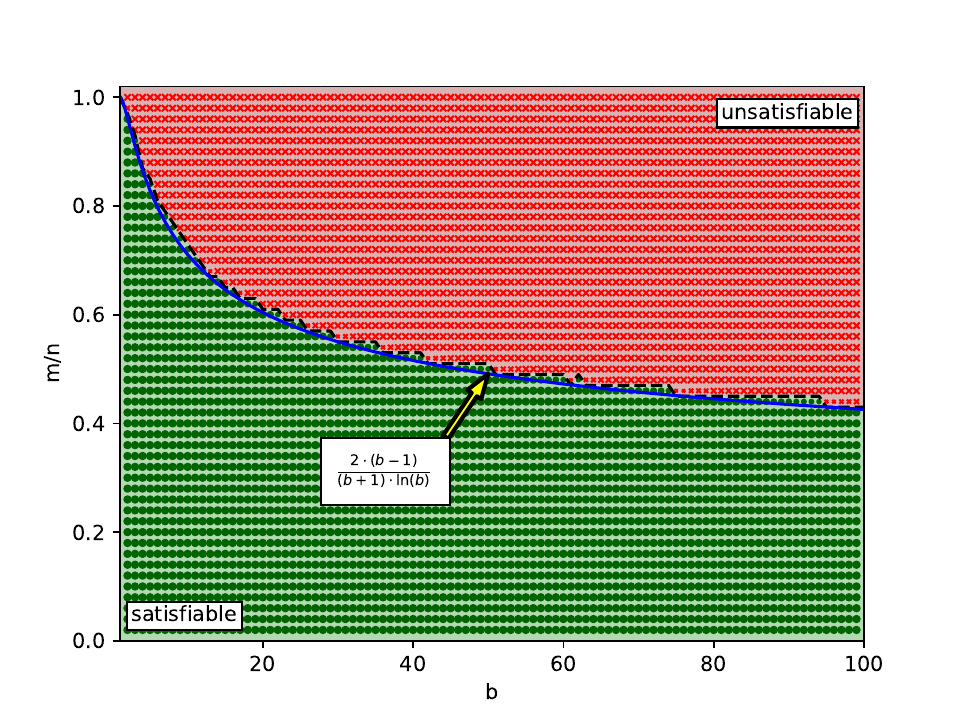}
  \caption{\label{fig:phase2} Phase diagram for geometric random 2-SAT formulas with $n=10^6$ variables.
	Each point is a sample of 100 random instances at the given parameter combination.
We drew a red cross if all instances were unsatisfiable and a green dot if at least one instance was satisfiable with the size of the dot scaling with the fraction of satisfiable instances.
We empirically observe a sharp phase transition
  (\protect\tikz[overlay]{\protect\path[draw=black,line width= 1.6pt,dash pattern=on 4pt off 4pt,line cap=round]
(0.04,0.1) -- (0.5,0.1);}\hspace*{5mm}), which closely matches
the theoretical bound of \thmref{main1}
(\protect\tikz[overlay]{\protect\path[draw=blue,line width= 1.6pt]
(0.04,0.1) -- (0.46,0.1);}\hspace*{5mm}).}\label{fig:geo-2-sat}
\end{figure}

\thmref{main1} implies the following corollary.
\begin{corollary}\label{cor:g2sat}
For geometric random 2-SAT with base $b>1$ there is a sharp threshold at $m^\star=\frac{2\cdot(b-1)}{(b+1)\cdot\ln b}\cdot n$.
\end{corollary} 
\begin{proof}
We assume some fixed $b>1$.
Then for $n\in\N$ the distribution is $\vec{p}^{(n)}=\left(p_1^{(n)},p_2^{(n)}\ldots,p_n^{(n)}\right)$ with
\[p_i^{(n)}=\frac{b\cdot(1-b^{-1/n})}{b-1}\cdot b^{-(i-1)/n}.\]
Again, it holds that $p_1\ge p_2 \ge \ldots\ge p_n$.
%
%

\lemref{g-aux1} tells us
\[p_1^2=(1-o(1))\cdot\left(\frac{b\cdot\ln b}{b-1}\right)^2\cdot n^{-2}.\]
and
\[\sum_{i=1}^n p_i^2 = (1\pm o(1))\cdot\frac{b+1}{b-1}\cdot\frac{\ln b}{2}\cdot n^{-1}.\]
Since $p_1^2\in o(\sum_{i=1}^n p_i^2)$, the threshold is sharp at 
$m^\star=\frac{2\cdot(b-1)}{(b+1)\cdot\ln b}\cdot n$.
\end{proof}
Again, \figref{geo-2-sat} visualizes the empirical threshold position compared to the theoretical one from \corref{g2sat}.

\section{Discussion and Future Work}

We showed a dichotomy of coarse and sharp thresholds for the non-uniform random $2$-SAT model depending on the variable probability distribution.
In the case of a coarse threshold, the coarseness either stems from two variables being present in too many clauses and forming an unsatisfiable sub-formula of size $4$ with constant probability or from a snake with three variables which emerges with constant probability.
Furthermore we determined the exact position of the satisfiability threshold in the case of a sharp threshold.
Hence, our result generalizes the seminal works by \citet{chvatalreed92} and by \citet{Goerdt1996threshold} to arbitrary variable probability distributions.
It allows us to prove or disprove an equivalent of the satisfiability threshold conjecture for non-uniform random $2$-SAT.
For example for power-law random 2-SAT, an equivalent of the conjecture holds for power law exponents $\beta\ge 3$ and the satisfiability threshold is at exactly $\frac{(\beta-3)\cdot(\beta-1)}{(\beta-2)^2}\cdot n$ for $\beta>3$ and exactly at $4\cdot\frac{n}{\ln n}$ for $\beta=3$.

The grand goal of our works is to show similar results for higher values of $k$, where we already made a first step by showing sharpness for certain variable probability distributions~\cite{SAT18}.
Another direction we are interested in for $k\ge3$ is proving bounds on the average computational hardness of formulas around the threshold, for example by showing resolution lower bounds like Mull et al.~\cite{communityHardness}.
We achieved first results in that direction for Power-Law Random $k$-SAT in \cite{BlaesiusFGLR21}.

\bibliography{2SAT}

\appendix
\section{Omitted Proofs}

\stateplbounds*
\begin{proof}
It holds that
\[1+\int_{i=1}^{n}\left(\frac{n}{i}\right)^{1/(\beta-1)}\,di\le\sum_{i=1}^n \left(\frac{n}{i}\right)^{1/(\beta-1)}\le n^{1/(\beta-1)}+\int_{i=1}^{n}\left(\frac{n}{i}\right)^{1/(\beta-1)}\,di.\]
Since
\[\int_{i=1}^{n}\left(\frac{n}{i}\right)^{1/(\beta-1)}\,di=\frac{\beta-1}{\beta-2}\cdot \left(n-n^{1/(\beta-1)}\right),\]
for $\beta>2$, we have
\[\sum_{i=1}^n \left(\frac{n}{i}\right)^{1/(\beta-1)}= (1-o(1))\cdot \frac{\beta-1}{\beta-2}\cdot n\]
and thus
\[p_i=(1+o(1))\cdot\frac{\beta-2}{\beta-1}\cdot n^{-1}\cdot\left(\frac{n}{i}\right)^{1/(\beta-1)}.\]
For $\beta=3$ it holds that
\[\sum_{i=1}^n p_i^2 = (1+o(1))\cdot\left(\frac{\beta-2}{\beta-1}\right)^2\cdot n^{-1} \cdot \sum_{i=1}^n \frac1i=(1+ o(1))\cdot\frac14\cdot \frac{\ln}{n}.\]
Otherwise, we consider the function
\[\left(\frac{\beta-2}{\beta-1}\right)^2\cdot n^{-2\frac{\beta-2}{\beta-1}}\cdot i^{-\frac{2}{\beta-1}}\]
and the integral
\[\int_{i=1}^{n} \left(\frac{\beta-2}{\beta-1}\right)^2\cdot n^{-2\frac{\beta-2}{\beta-1}}\cdot i^{-\frac{2}{\beta-1}}\,di=\frac{(\beta-2)^2}{(\beta-3)\cdot(\beta-1)}\cdot\left(n^{-1}-n^{-2\frac{\beta-2}{\beta-1}}\right).\]
Again, we can use the relation of sum and integral to derive
\[\sum_{i=1}^n p_i^2 =\left(1+ o(1)\right)\cdot\frac{(\beta-2)^2}{(\beta-3)\cdot(\beta-1)}\cdot n^{-1}\]
for $\beta>3$ and
\[\sum_{i=1}^n p_i^2 \in \Theta\left(n^{-2\frac{\beta-2}{\beta-1}}\right)\]
for $\beta<3$ as desired.
Although we do not derive the exact leading factor in the last case, the asymptotic expression is sufficient for our results.
\end{proof}

\stategeombounds*
\begin{proof}
It holds that
\[\sum_{i=1}^n p_i^2=\frac{b^2 \cdot(1-b^{-1/n})^2}{(b-1)^2}\cdot\sum_{i=0}^{n-1} \left(\frac{1}{b^{2/n}}\right)^i=\frac{b+1}{b-1}\cdot\frac{1-b^{-1/n}}{1+b^{-1/n}},\]
since this is a simple geometric series.
We get
\begin{align*}
\sum_{i=1}^n p_i^2=\frac{b+1}{b-1}\cdot\frac{1-b^{-1/n}}{1+b^{-1/n}}
&=\frac{b+1}{b-1}\cdot\frac{1-e^{-\ln (b)/n}}{1+e^{-\ln(b)/n}}\\
&\le \frac{b+1}{b-1}\cdot\frac{1-(1-\ln (b)/n)}{1+(1-\ln(b)/n)}\\
&= \left(1+\frac{\ln b}{2n-\ln b}\right)\cdot\frac{b+1}{b-1}\cdot\frac{\ln (b)}{2n}
\end{align*}
and
\begin{align*}
\sum_{i=1}^n p_i^2=\frac{b+1}{b-1}\cdot\frac{1-b^{-1/n}}{1+b^{-1/n}}
&=\frac{b+1}{b-1}\cdot\frac{b^{1/n}-1}{b^{1/n}+1}\\
&=\frac{b+1}{b-1}\cdot\frac{e^{\ln (b)/n}-1}{(1+(b-1))^{1/n}+1}\\
&\ge\frac{b+1}{b-1}\cdot\frac{(1+\ln (b)/n)-1}{1+(b-1)/n+1}\\
&=\left(1-\frac{b-1}{2n+b-1}\right)\cdot\frac{b+1}{b-1}\cdot\frac{\ln (b)}{2n},
\end{align*}
where we used Bernoulli's inequality in the denominator of the third line.
This establishes the first statement.
For the second statement observe
\[p_1=\frac{b\cdot(1-b^{-1/n})}{b-1}=\frac{b\cdot(1-e^{-\ln(b)/n})}{b-1}\le \frac{b\ln (b)}{b-1}\cdot n^{-1}\]
and
\begin{align*}
p_1=\frac{b\cdot(1-b^{-1/n})}{b-1}&=\frac{b\cdot(b^{1/n}-1)}{(b-1)\cdot b^{1/n}}\\
&=\frac{b\cdot(e^{\ln(b)/n}-1)}{(b-1)\cdot (1+(b-1))^{1/n}}\\
&\ge\frac{b\cdot\ln(b)/n}{(b-1)\cdot (1+(b-1)/n)}\\
&=\left(1-\frac{b-1}{n+b-1}\right)\cdot\frac{b\ln (b)}{b-1}\cdot n^{-1}.
\end{align*}
\end{proof}

\statenondec*
\begin{proof}
Let $n\in\N$, $k\in\N$, and $\left(\vec{p}^{(n)}\right)_{n\in\N}$ be arbitrary, but fixed.
Since $m$ is the only free parameter, we let $\mathcal{D}^N(m)$ denote $\mathcal{D}^N\left(n, k, \left(\vec{p}^{(n)}\right)_{n\in\N}, m\right)$ for the sake of simplicity.
Now choose some $m\in\N$ arbitrarily.
We are going to show that
\[\Pro{\Phi\sim\mathcal{D}^N(m+1)}{P(\Phi)=1}\ge\Pro{\Phi\sim \mathcal{D}^N(m)}{P(\Phi)=1}.\]
We interpret each formula $\Phi$ as a sequence of (not necessarily distinct) clauses $(c_1,c_2,\ldots,c_m)$.
Since clauses are drawn independently with replacement in $\mathcal{D}^N$, for two formulas $x=(c_1,c_2,\ldots,c_m)$ and $y=(c_1,c_2,\ldots,c_m, c_{m+1})$ it holds that
\[\Pro{\Phi\sim \mathcal{D}^N(m+1)}{\Phi=y}=\Pro{\Phi\sim \mathcal{D}^N(m)}{\Phi=x}\cdot\Pro{\Phi\sim \mathcal{D}^N(1)}{\Phi=(c_{m+1})}.\]
Now let $P_{k,m}$ denote the set of all formulas with property $P$ in $k$-CNF with at most $m$ clauses and let $\mathcal{C}$ denote the set of all possible $k$-clauses over $n$ variables.
Then,
\begin{align*}
&\Pro{\Phi\sim \mathcal{D}^N(m+1)}{P(\Phi)=1}\\
& = \sum_{y=(c_1,\ldots,c_m,c)\in P_{k,m+1}}\Pro{\Phi\sim \mathcal{D}^N(m+1)}{\Phi=y}\\
& = \sum_{x=(c_1,\ldots,c_m)\in\mathcal{C}^m}\left(\Pro{\Phi\sim \mathcal{D}^N(m)}{\Phi=x}\cdot \sum_{\substack{c\in\mathcal{C}\colon\\ (c_1,\ldots,c_m,c)\in P_{k,m+1}}}\left(\Pro{\Phi\sim \mathcal{D}^N(1)}{\Phi=(c)}\right)\right)\\
& \ge \sum_{x=(c_1,\ldots,c_m)\in P_{k,m}}\left(\Pro{\Phi\sim \mathcal{D}^N(m)}{\Phi=x}\cdot \sum_{\substack{c\in\mathcal{C}\colon\\ (c_1,\ldots,c_m,c)\in P_{k,m+1}}}\left(\Pro{\Phi\sim \mathcal{D}^N(1)}{\Phi=(c)}\right)\right).\\
\intertext{Due to the monotonicity of $P$, if $x\in P_{k,m}$, then any $y$ which extends $x$ by one clause is in $P_{k,m+1}$. Thus,}
& = \sum_{y=(c_1,\ldots,c_m, c_{m+1})\in P_{k,m}}\left(\Pro{\Phi\sim \mathcal{D}^N(m)}{\Phi=x}\cdot \sum_{c\in\mathcal{C}}\left(\Pro{\Phi\sim \mathcal{D}^N(1)}{\Phi=(c)}\right)\right)\\
&=\Pro{\Phi\sim \mathcal{D}^N(m)}{P(\Phi)=1},
\end{align*}
since 
\[\sum_{c\in\mathcal{C}}\Pro{\Phi\sim \mathcal{D}^N(1)}{\Phi=(c)}=1.\]
This proves that the probability for $P$ is non-decreasing in $\mathcal{D}^N$ as $m$ increases.
\end{proof}

\end{document}